\documentclass[a4paper,twocolumn,11pt,accepted=2024-08-19]{quantumarticle}
\pdfoutput=1
\usepackage{mathtools}
\usepackage[utf8]{inputenc}
\usepackage{graphicx}%
\usepackage{dcolumn}%
\usepackage{bm}%
\usepackage{mathrsfs}
\usepackage{amsfonts}
\usepackage{color}
\usepackage{xcolor}
\usepackage{braket}
\usepackage{paralist}
\usepackage{xcolor}
\usepackage{color}
\usepackage{graphicx}
\graphicspath{{./figs/}}
\usepackage{algorithm}
\usepackage{algorithmic}
\usepackage{comment}
\usepackage{multirow}
\usepackage{tcolorbox}
\usepackage{bbm}
\usepackage{array}
\usepackage{subcaption}
\usepackage{pifont}
\newcommand{\cmark}{\ding{51}}
\newcommand{\xmark}{\ding{55}}
\usepackage [numbers,sort&compress]{natbib}

\usepackage[english]{babel}
\usepackage{url}
\definecolor{darkblue}{rgb}{0,0,0.5}
\usepackage{hyperref}
\hypersetup{
colorlinks=true,
linkcolor=black,
filecolor=blue,
citecolor=darkblue,  
urlcolor=darkblue,
}
\usepackage{cleveref}
\usepackage{amsmath,amsthm,amssymb,amsbsy}

\usepackage[absolute]{textpos}

\newtheorem{theorem}{Theorem}%
\newtheorem{lemma}{Lemma}%

\theoremstyle{remark}

\newcommand{\R}{\mathbb{R}}

\newcommand{\C}{\mathbb{C}}

\newcommand{\e}{\begin{equation}}
\newcommand{\ee}{\end{equation}}
\newcommand{\en}{\begin{equation*}}
\newcommand{\een}{\end{equation*}}
\newcommand{\eqn}{\begin{eqnarray}}
\newcommand{\eeqn}{\end{eqnarray}}
\newcommand{\bmat}{\begin{bmatrix}}
\newcommand{\emat}{\end{bmatrix}}

\DeclareMathAlphabet\mathbfcal{OMS}{cmsy}{b}{n}

\newcommand{\E}{\operatorname{\mathbb{E}}}

\newcommand{\vct}[1]{\boldsymbol{#1}}
\newcommand{\mtx}[1]{\boldsymbol{#1}}
\newcommand{\trace}{\operatorname{tr}}

\def \vec       {\operatorname*{vec}}

\DeclareMathOperator*{\argmin}{\text{arg~min}}

\newcommand{\wh}{\widehat}

\newcommand{\norm}[2]{\left\| #1 \right\|_{#2}}

\newcommand{\bracket}[1]{\left( #1 \right)}
\newcommand{\parans}[1]{\left(#1\right)}

\newcommand{\innerprod}[2]{\trace\parans{ #1  #2}}

\newcommand{\calA}{\mathcal{A}}

\newcommand{\calM}{\mathcal{M}}

\newcommand{\ve}{\vct{e}}

\newcommand{\vp}{\vct{p}}

\newcommand{\vu}{\vct{u}}

\newcommand{\vphi}{\vct{\phi}}

\newcommand{\vrho}{\vct{\rho}}

\newcommand{\vone}{\vct{1}}

\newcommand{\mA}{\mtx{A}}
\newcommand{\mB}{\mtx{B}}

\newcommand{\mU}{\mtx{U}}

\newcommand{\mLambda}{\mtx{\Lambda}}

\newcommand{\mId}{{\bf I}}

\newcommand{\setA}{\mathbb{A}}

\newcommand{\setU}{\mathbb{U}}

\graphicspath{{./figs/}}

\newlength{\imgwidth}
\newboolean{twoColVersion}
\setboolean{twoColVersion}{false}
\newcommand{\twoCol}[2]{\ifthenelse{\boolean{twoColVersion}} {#1} {#2} }

\begin{document}

\title{On the connection between least squares, regularization, and classical shadows}

\author{Zhihui Zhu}
\email{zhu.3440@osu.edu}
\affiliation{Department of Computer Science and Engineering, The Ohio
State University, Columbus, Ohio 43210, USA}
\author{Joseph M. Lukens}
\email{joseph.lukens@asu.edu}
\affiliation{Research Technology Office and Quantum Collaborative, Arizona State University, Tempe, Arizona 85287, USA}
\affiliation{Quantum Information Science Section, Oak Ridge National Laboratory, Oak Ridge, Tennessee 37831, USA}
\author{Brian T. Kirby}
\email{brian.t.kirby4.civ@army.mil}
\affiliation{DEVCOM Army Research Laboratory, Adelphi, MD 20783, USA}
\affiliation{Tulane University, New Orleans, LA 70118, USA}

\begin{abstract}
Classical shadows (CS) offer a resource-efficient means to estimate quantum observables, circumventing the need for exhaustive state tomography. 
Here, we clarify and explore the connection between CS techniques and least squares (LS) and regularized least squares (RLS) methods commonly used in machine learning and data analysis. 
By formal identification of LS and RLS ``shadows'' completely analogous to those in CS---namely, point estimators calculated from the empirical frequencies of single measurements---we show that both RLS and CS can be viewed as regularizers for the underdetermined regime, replacing the pseudoinverse with invertible alternatives. 
Through numerical simulations, we evaluate RLS and CS from three distinct angles: the tradeoff in bias and variance, mismatch between the expected and actual measurement distributions, and the interplay between the number of measurements and number of shots per measurement.

Compared to CS, RLS attains lower variance at the expense of bias, is robust to distribution mismatch, and is more sensitive to the number of shots for a fixed number of state copies---differences that can be understood from the distinct approaches taken to regularization. Conceptually, our integration of LS, RLS, and CS under a unifying ``shadow'' umbrella aids in advancing the overall picture of CS techniques, while practically our results highlight the tradeoffs intrinsic to these measurement approaches, illuminating the circumstances under which either RLS or CS would be preferred, such as unverified randomness for the former or unbiased estimation for the latter.
\end{abstract}

\maketitle

\section{Introduction}
As experimentally accessible quantum systems continue to increase in size and complexity, methods for characterizing these systems as efficiently as possible have assumed primary importance. One of the leading state characterization approaches is quantum state tomography, which provides a complete density matrix describing a system from which all observable properties can be extracted via classical calculations~\cite{nielsen2000quantum}. 
However, the exponential scaling in the required number of measurements and the classical computational cost of determining the density matrix most consistent with the given measurement results make state tomography an unrealistic approach for large systems. 

The practical challenges associated with full state tomography have spurred the development of various methods for estimating properties and observables of quantum systems without needing to reconstruct the entire density matrix \cite{aaronson2019shadow}. 
For example, nonlinear functions of a density matrix, such as various entanglement measures, can be estimated directly without reconstruction but at the cost of requiring multiple copies of a given state and joint measurements between them \cite{horodecki2002method,horodecki2003measuring,mintert2007observable}.
Further, techniques based on randomized measurements (e.g., application of random unitaries to states with fixed measurement bases) have been developed to estimate observables from single-copy systems without requiring state reconstruction. 
These methods include those that do not incorporate the explicit set of randomized measurements selected into the estimation procedure \cite{van2012measuring,wyderka2023complete,ketterer2019characterizing,imai2021bound,ketterer2022statistically,liang2010nonclassical,tran2015quantum,seshadri2021theory,cieslinski2023analysing} as well as those, such as classical shadows (CS), that do utilize this information and hence still require a shared frame of reference \cite{huang2020predicting,acharya2021informationally,Nguyen2022,struchalin2021experimental}.

The CS approach to estimating observables of an unreconstructed density matrix is especially attractive due to its experimental simplicity and demonstrated predictive power \cite{huang2020predicting}. 
The original CS proposal
leverages random single-shot measurements to construct ``shadows'' of a quantum state that then stand-in for a complete density matrix reconstruction for the purpose of calculating observables. Even though the CS density matrix is not even constrained to be positive semidefinite (PSD)---an ostensibly surprising feature critical to its unique scaling behavior~\cite{lukens2021bayesian}---it has been shown to provide accurate estimates for many quantities of interest in a quantum system. 

Since its initial development, CS techniques have been studied intensely in various scenarios including---but not limited to---experimental data \cite{elben2020mixed, struchalin2021experimental, Zhang2021, Stricker2022, Zhu2022}, compared to  approaches such as Bayesian mean estimation~\cite{lukens2021bayesian}, extended to positive operator-valued measures (POVMs)~\cite{acharya2021informationally, Nguyen2022} and multiple shots per measurement setting~\cite{zhou2023performance}, and derandomized to remove the necessity of randomly chosen measurements~\cite{huang2021efficient}. 

From a physical point of view, the original CS proposal is relatively straightforward~\cite{huang2020predicting}. %
The randomness of the measurement procedure induces a depolarizing channel with a simple inversion consisting of subtraction of the identity. Intuitively, since the projections are restricted to a single shot, they naturally act as a noisy channel; they are a low-dimensional projection of a higher-dimensional probability distribution, and the randomness of these projections ensures the noise is isotropic. %
Importantly from an operational perspective, the inverse of an appropriately chosen random channel can be computed analytically, thereby obviating the need for a computationally intensive inverse calculation. 

Here we consider the fundamental causes for the estimation power of CS in light of standard least-squares (LS) and regularized least-squares (RLS) formulations of the same problem. %
Through a formal derivation of the LS and RLS solutions to a generic quantum measurement scenario, we find that both rely on their own ``shadows''---i.e., linear transformations of individual measurement results---which are averaged to obtain the final estimate. As overviewed in Fig.~\ref{fig:concept} and detailed in Secs. \ref{sec:background} and \ref{sec:stabilizing} below, all three techniques (LS, RLS, and CS) follow strikingly similar workflows, differing only in the respective inversion operation in each's shadow formula. Under this viewpoint, these techniques are seen to comprise a general family in which RLS and CS stabilize the LS shadow in the underdetermined regime by replacing the pseudoinverse with invertible and well-conditioned operators. Not only does our work reveal the unifying framework that the idea of ``shadows'' provides for traditional LS techniques; it also uncovers an profitable interpretation of CS as a regularizer for low-measurement quantum estimation in the tradition of RLS. 
Collectively, our results contribute to the fundamental understanding of CS while simultaneously offering practical guidance for quantum estimation.

This article is organized as follows. Section~\ref{sec:background} introduces the general measurement problem in terms of POVMs and derives the LS solution, expressing the result as an average of LS shadows. Simulations of a five-qubit system reveal high variance and error from the double descent phenomenon, which is mitigated by the RLS and CS stabilization techniques introduced in Sec.~\ref{sec:stabilizing}. Section~\ref{sec:Compare} then compares the advantages and disadvantages of RLS and CS with respect to three specific features: (i)~the bias-variance tradeoff, (ii)~the impact of misspecified measurement distributions, and (iii)~the scaling with reallocations of the number of measurements and the number of shots per measurement. Concluding thoughts appear in Sec.~\ref{sec:conclusion}.

\begin{figure*}[t!]
{\centering
\includegraphics[width=6.5in]{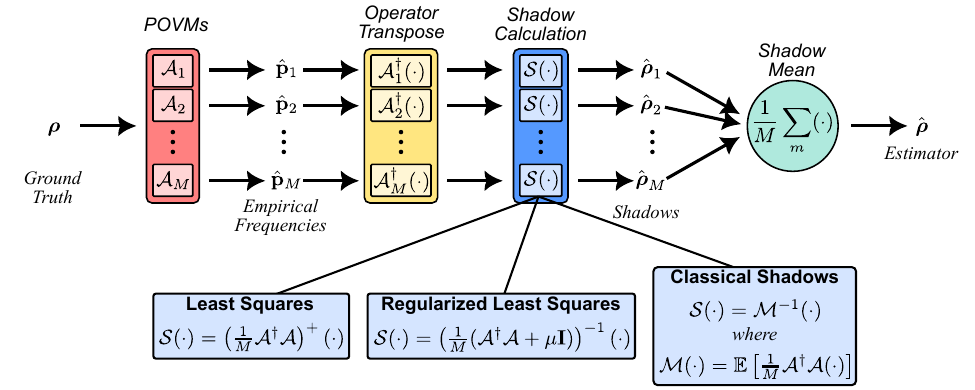}
}
\caption{Shadow picture of quantum estimation. POVMs $\{\mathcal{A}_1,\mathcal{A}_2,\ldots,\mathcal{A}_M\}\equiv\mathcal{A}$ are measured via repeated preparation of a ground truth quantum state $\boldsymbol{\rho}$. The observed frequencies for each POVM produce a single shadow state $\wh{\boldsymbol{\rho}}_m = \mathcal{S}\left(\mathcal{A}_m^\dagger(\wh{\vp}_m)\right)$, the collection of which are averaged for the final estimate $\wh{\boldsymbol{\rho}}$. The only difference between each technique lies in the specific shadow operation chosen: (i)~least squares (LS) performs the (pseudo)inverse on the POVMs directly; (ii)~regularized least squares (RLS) ensures invertibility through the addition of a term proportional to the identity; and (iii)~classical shadows (CS) inverts according to a simulated channel $\mathcal{M}$ defined in expectation over all possible measurement settings.}
\label{fig:concept}
\end{figure*}

\section{Background: POVM Measurements and LS Estimation}
\label{sec:background}
\subsection{POVM Measurements}
Consider an $n$-qubit quantum state $\vrho\in\C^{D\times D}$ with dimension $D = 2^n$. The probabilistic nature of quantum measurements can be described using POVMs~\cite{nielsen2000quantum}.  A POVM is a set of PSD operators $\{\mA_1,\ldots,\mA_K \}$---abbreviated as $\{\mA_k\}_{k\in[K]}$ with $[K] :=\{1,\ldots,K\}$ or simply $\{\mA_k\}$ when clear from context---such that $
\sum_{k=1}^K \mA_k = \mId.$
Each POVM element $\mA_k$ is associated with a possible measurement outcome, and the probability $p_k$ of detecting the $k$-th outcome when measuring the density operator $\vrho$ is given by
\begin{eqnarray}
\label{The defi of POVM 2}
p_k = \innerprod{\mA_k}{\vrho}.
\end{eqnarray}
We can repeat the measurement process $L$ times, observe  the $k$-th outcome $f_k$ number of times, and take the average of the outcomes to generate the empirical frequencies
\begin{equation}
\wh p_{k} = \frac{f_k}{L}, \  k \in[K].
\label{eq:empirical-prob}\end{equation}
Collectively, the random variables $f_1,\ldots,f_K$ are characterized by a multinomial distribution  with parameters $L$  and $\{p_k\}$. When $L = 1$, the measurements $\{\wh p_k = f_k\}$ form a one-hot vector, or a delta-like distribution, with all entries zero except for one entry being one. 
\textit{Orthogonal rank-1 POVMs.---}
A special case common in practice focuses on rank-1 POVMs of the form $\{ \mA_k=\vu_k \vu_k^\dagger \}$ with $\vu_k \in \C^D$ and $\sum_{k=1}^K \vu_k \vu_k^\dagger = \mId$, where $\dagger$ denotes the Hermitian transpose. We note that in the physics literature when $\vu$ represents a quantum state it is often represented as a ket $\ket{u}$; however, we adopt vector notation throughout for convenience.
When $\mU = \begin{bmatrix} \vu_1 & \cdots & \vu_K \end{bmatrix}^\dagger\in\C^{D\times  K}$ further forms an orthonormal basis, in which case $K = D$, the probability $p_k$ can be written as
\begin{equation}
p_k = \innerprod{\mA_k}{\vrho}  = \vu_k^\dagger \vrho \vu_k = \ve_k^\top \bracket{\mU \vrho \mU^\dagger} \ve_k,
\label{Probability of rank-1 orth POVM}
\end{equation}
where the last equation implies that the measurement is equivalent to first applying the unitary $\mU$ to the unknown state  $\vrho \mapsto \mU\vrho\mU^\dagger$ (the reason for the Hermitian transpose in the definition $\mU = \begin{bmatrix} \vu_1 & \cdots & \vu_K \end{bmatrix}^\dagger$) and then performing measurements in the canonical (or computational) basis $\ve_1,\ldots,\ve_{D}$. Both steps can be
implemented on a universal quantum computer, though the complexity of synthesizing $\mU$ by quantum circuits is matrix-dependent. As an aside, we note that rank-1 orthonormal POVMs with $L=1$ comprised the focus of the original CS proposal~\cite{huang2020predicting}, although CS extensions to both more generic POVMs~\cite{Nguyen2022} and $L>1$ are possible~\cite{elben2023randomized, zhou2023performance}. 

\textit{POVM ensembles.---}As in the case of the above rank-1 POVM, an individual POVM might not achieve informational completeness; therefore measuring with multiple POVMs can be used to acquire a more holistic understanding of the quantum state. For simplicity, consider $M$ POVMs, indexed by $m\in[M]$, where each POVM $\{\mA_{m,k}\}_{k\in[K]}$ contains the same number of PSD operators $K$ and is probed with $L$ shots, returning the empirical frequencies $\wh \vp_m$ as described in Eq.~\eqref{eq:empirical-prob}, for $m \in [M]$ where the bold notation indicates a vector: $\wh \vp_m = [f_{m,1} \cdots f_{m,K}]^\top/L  = [\wh p_{m,1} \cdots \wh p_{m,K}]^\top$. 

To simplify the notation, we collect the probabilities for each POVM $\{ \innerprod{\mA_{m,k}}{\vrho}\}$, %
into a single linear map $\calA_m:  \C^{D\times  D} \rightarrow \R^K$ of the form %
\begin{equation}
\calA_m(\vrho)= \begin{bmatrix}
         \innerprod{\mA_{m,1}}{\vrho} \\
          \vdots \\
          \innerprod{\mA_{m,K}}{\vrho}
        \end{bmatrix}.
\end{equation}
If we vectorize $\vrho$ and $\mA_{m,k}$ into $\vec(\vrho)$ and $\vec(\mA_k)$  such that $\innerprod{\mA_k}{\vrho} = (\vec(\mA_k))^\dagger \vec(\vrho)$ and define $\mB = \begin{bmatrix}  
\vec(\mA_{m,1}) & \cdots & \vec(\mA_{m,K}) 
\end{bmatrix}\in \C^{D^2\times K}$, then $\calA_m(\vrho)$ can be written as matrix-vector product of form
\e
\calA_m(\vrho) = \mB^\dagger \vec(\vrho).
\label{eq:vector form}\ee
Stacking all the empirical frequencies $\{\wh \vp_m\}$ and the linear operators $\{\calA_m\}$ as a single linear map  $\calA:  \C^{D\times  D} \rightarrow \R^{MK}$, we can write 
\e
\wh \vp = \begin{bmatrix}
    \wh \vp_1 \\ \vdots \\ \wh \vp_M
\end{bmatrix}, \quad \calA(\vrho) =  \begin{bmatrix}
    \calA_1(\vrho) \\ \vdots \\ \calA_M(\vrho)
\end{bmatrix}.
\label{eq:map-M-POVM}\ee
It is important to note that no assumptions about informational (tomographic) completeness have been applied in the formalism so far. Specifically, the linear map $\calA$ is informationally complete
iff $\mathrm{rank}(\calA) = D^2$; i.e., we can form exactly $D^2$
linearly independent operators by linearly combining the set of POVMs \cite{acharya2021informationally}.
In the regime of interest to CS, $\mathrm{rank}(\calA)\ll D^2$ typically holds, although we note that it is possible to formally define a \emph{single informationally complete} POVM so that $\mathrm{rank}(\calA)= D^2$ even with $M=1$---a construction that has been shown valuable for both theoretical analyses~\cite{acharya2021informationally,Nguyen2022} and experimental implementation~\cite{Stricker2022} of CS techniques. In this case, any estimator error stems solely from the number of shots $L$. In our analysis, we do not restrict to informational completeness and in the numerical simulations below consider rank-1 POVMs with $K=D$ outcomes.  
Nonetheless, the formalism developed applies to any combination of $M$, $K$, and $L$ and thus can be explored for any POVMs of potential interest.

\subsection{LS Estimation}
Without any prior information about $\vrho$, we can estimate it from the measurements $\wh \vp$ by the LS estimator
\e
\wh \vrho = \argmin_{\vrho' \in \C^{D\times  D}} \norm{\wh \vp - \calA(\vrho')}{2}^2,
\label{eq:ls-problem}\ee
by writing $\calA(\vrho')$ as matrix-vector product as in Eq.~(\ref{eq:vector form}). While one can explicitly enforce $\vrho'$ to be Hermitian and trace one, 
we will show in Lemma~\ref{lemma:1} that solutions to Eq.~(\ref{eq:ls-problem}) automatically adhere to these properties.
Additionally, the solution $\wh \vrho$ satisfies the following normal equation 
\e
\calA^\dagger\calA(\wh\vrho) = \calA^\dagger(\wh \vp),
\ee
where
\begin{align}
\begin{split}
    & \calA^\dagger(\wh \vp) = \sum_{m=1}^M \calA_m^\dagger(\wh \vp_m)  = \sum_{m=1}^M \sum_{k=1}^K \wh p_{m,k} \mA_{m,k},\\
    &\calA^\dagger \calA(\wh\vrho) = \sum_{m=1}^M \sum_{k=1}^K  \innerprod{\mA_{m,k}}{\wh\vrho} \mA_{m,k}.
\end{split}
\end{align}
When the ensemble of $M$ POVMs is informationally complete, %
$\calA^\dagger \calA$ is invertible and the solution is unique, given by $\wh \vrho = \parans{\calA^\dagger\calA}^{-1} \parans{\calA^\dagger(\wh \vp)}$. On the contrary, when the $M$ POVMs are not informationally complete, 
the operator $\calA^\dagger \calA$ is rank-deficient and the above problem has an infinite number of solutions. 
Among all possibilities, a common choice is to select the one that has the smallest norm or energy, also known as minimum-norm estimator, which can be obtained by applying the pseudoinverse %
$(\calA^\dagger \calA)^{+}$:
\begin{align}
\begin{split}
    \wh \vrho & = \parans{\calA^\dagger\calA}^{+} \parans{\calA^\dagger(\wh \vp)} \\ & = \sum_{m=1}^M  \parans{\calA^\dagger \calA}^{+}\parans{\calA_m^\dagger(\wh \vp_m)} \\ & = \frac{1}{M} \sum_{m=1}^M  \underbrace{\left(\frac{1}{M}\calA^\dagger \calA\right)^{+}\parans{\calA_m^\dagger(\wh \vp_m)}}_{\text{LS shadow}},\label{eq:shadow-v1}
\end{split}
\end{align}
where the significance of defining an LS ``shadow'' operator is elaborated on below.
When the POVMs are informationally complete, $(\calA^\dagger \calA)^{+}$ becomes $(\calA^\dagger \calA)^{-1}$. Thus, Eq.~\eqref{eq:shadow-v1} holds for both informationally complete and incomplete cases.

\begin{lemma}\label{lemma:1} The LS estimator $\wh\vrho$ is always Hermitian. Moreover, if $\wh \vp$ lies in the range space of $\calA$, then the LS estimator $\wh\vrho$ also has trace $1$. 
\end{lemma}

\begin{proof} 
Using the two equivalent forms for the pseudoinverse $\calA^+ = \parans{\calA^\dagger\calA}^{+} \calA^\dagger = \calA^\dagger \parans{\calA\calA^\dagger}^{+}$  \cite{albert1972regression}, we can rewrite the LS estimator as $\widehat \vrho = \calA^\dagger \parans{\parans{\calA\calA^\dagger}^{+}(\wh \vp)}$. Since $\wh \vp$ is a real vector and $\calA\calA^\dagger$ is a linear map of $\R^{MK}\rightarrow \R^{MK}$, $\parans{\calA\calA^\dagger}^{+}(\wh \vp)$ is also a real vector. As $\widehat \vrho$ lies in the range space of $\calA^\dagger$ that can be written as $\sum_{m=1}^M \sum_{k=1}^K \alpha_{m,k}\mA_{m,k}$ with real $\alpha_{m,k}$ and Hermitian $\mA_{m,k}$, $\wh \vrho$ is always Hermitian.

Now if $\wh \vp$ lies in the range space of $\calA$, then the LS estimator $\wh\vrho$ satisfies $\calA(\wh \vrho) = \wh \vp$, which further implies that $\vone^\top \calA(\wh \vrho) = \vone^\top \wh \vp$. Since $\vone^\top \wh \vp = M$ and  $\vone^\top \calA(\wh \vrho) = \sum_{m=1}^M \sum_{k=1}^K  \innerprod{\mA_{m,k}}{\wh\vrho} = M \trace(\wh \vrho)$, we have $\trace(\wh \vrho) = 1$. 
\end{proof}

When the operators  $\mA_{1,1},\ldots,\mA_{M,K}$ are linearly independent, $\wh \vp$ lies in the range space of $\calA$. On the other hand, even when $\wh \vp$ does not lie in the range space of $\calA$---which could happen when $MK>D^2$ and in which case the second half of Lemma~\ref{lemma:1} does not apply---we have observed in numerical experiments that the LS estimator $\wh \vrho$ either has trace 1 or is very close to 1 in practice. 

\textit{LS shadow.---}In anticipation of the CS formalism introduced in Sec.~\ref{sec:CS}, we may define $\wh\vrho_m \vcentcolon= \mathcal{S}\left(\calA_m^\dagger(\wh \vp_m)\right) = (\frac{1}{M}\calA^\dagger \calA)^{+}\calA_m^\dagger(\wh \vp_m)$ the {\it LS shadow} of $\vrho$ associated with measurement $m$. Unlike CS shadows, these LS shadows can be {\it biased} estimators of the state $\vrho$ as their average $\wh \vrho$  biases towards the minimum norm \cite{schwemmer2015systematic}. Yet like CS, the final state estimate $\wh\vrho$ is simply the average of the available individual shadows.

Intuitively, one would expect the performance of the LS estimator to improve with more POVM measurements $M$. We can quantify 
the agreement of the total estimator $\wh\vrho = \frac{1}{M}\sum_{m=1}^M \wh\vrho_m$ with the ground truth $\vrho$ through the Frobenius error $\|\wh\vrho-\vrho\|_F$, and the value of any observable $\lambda = \trace(\mLambda \vrho)$ through the mean squared error (MSE) $\E [(\wh\lambda - \lambda)^2]$, where $\wh \lambda = \trace(\mLambda \wh\vrho)$. 

Computing the error of all estimators with respect to ground truth quantities provides an objective standard by which we can compare all estimation methods in this study. In this vein, it is important to note that the three approaches under consideration (LS, RLS, and CS) correspond to different estimation techniques under a common physical model; i.e., all approaches seek to estimate the same unknown quantum state $\rho$ and assume the same mapping from state to probabilities [Eq.~\eqref{The defi of POVM 2}].
In other words, the problem of interest concerns estimation techniques and not model selection, so model identification tools such as the Akaike information criterion~\cite{Akaike1974}---considered in a variety of quantum state estimation contexts~\cite{Yin2011, vanEnk2013, Scholten2018, Yano2023}---do not apply.

On another note, Ref.~\cite{huang2020predicting} employed an additional statistical technique, ``median of means,'' to reduce the impact of outliers by partitioning the shadows into several groups and taking the median as the estimate. Incidentally, in recent experimental tests of CS, no significant difference was observed in the performance of the two approaches (mean versus median of means)~\cite{struchalin2021experimental}. Roughly speaking, the median-of-means approach is not designed to reduce the variance of the estimator, but rather  %
obtain a better concentration bound than the sample mean alone \cite{huang2020predicting}. Thus, as this paper focuses on the variance (i.e., MSE) instead of the concentration bound for performance quantification, the sample mean represents the most suitable estimator for our purposes. All that said,
we do expect similar phenomena to hold for all comparisons below with the median of means.

For our simulated experiments, we invoke the setup of the original CS proposal \cite{huang2020predicting} with ($K=D$)-outcome, rank-1 POVMs $\{\mA_{m,1},\ldots,\mA_{m,D} \}$ defined according to $\mA_{m,k} = \vu_{m,k}\vu_{m,k}^\dagger$, where each $\mU_m = \begin{bmatrix} \vu_{m,1} \cdots \vu_{m,D} \end{bmatrix}^\dagger$ is a randomly chosen $D\times  D$ unitary matrix.
Each POVM measures the state only once (i.e., $L = 1$) so that $\wh \vp_m$ becomes a one-hot vector, in which case $\calA_m^\dagger (\wh \vp_m)$ can be rewritten as 
\e
\calA_m^\dagger (\wh \vp_m) = \sum_{k=1}^D \wh p_{m,k} \vu_{m,k}\vu_{m,k}^\dagger = (\mU_m^\dagger \wh \vp_m)(\mU_m^\dagger \wh \vp_m)^\dagger.
\label{eq:Aadjoint-rank-one}\ee

Motivated by our previous study~\cite{lukens2021bayesian}, we consider $n = 5$ qubits, Haar-random unitaries, and a fixed ground truth state 
$\vrho = \ve_0 \ve_0^\dagger$. We focus on three rank-1 observables,  $\mLambda_i = \vphi_i \vphi_i^\dagger, i = 0,1,2$, where 
\begin{align}
    \vphi_0 & = \ve_0, 
    \vphi_1  = \frac{1}{\sqrt{2}}\ve_0 + \frac{1}{\sqrt{2(D-1)}} \sum_{j=1}^{D-1}\ve_j,
    \vphi_2 = \ve_1.
\end{align}
These possess ground truth values $\lambda_0=1$, $\lambda_1=1/2$, and $\lambda_2=0$ regardless of dimension $D$,
providing an informative range for exploration.

\begin{figure*}[t]
\centering
\begin{subfigure}{0.24\textwidth}
\centering
\includegraphics[width= \textwidth]{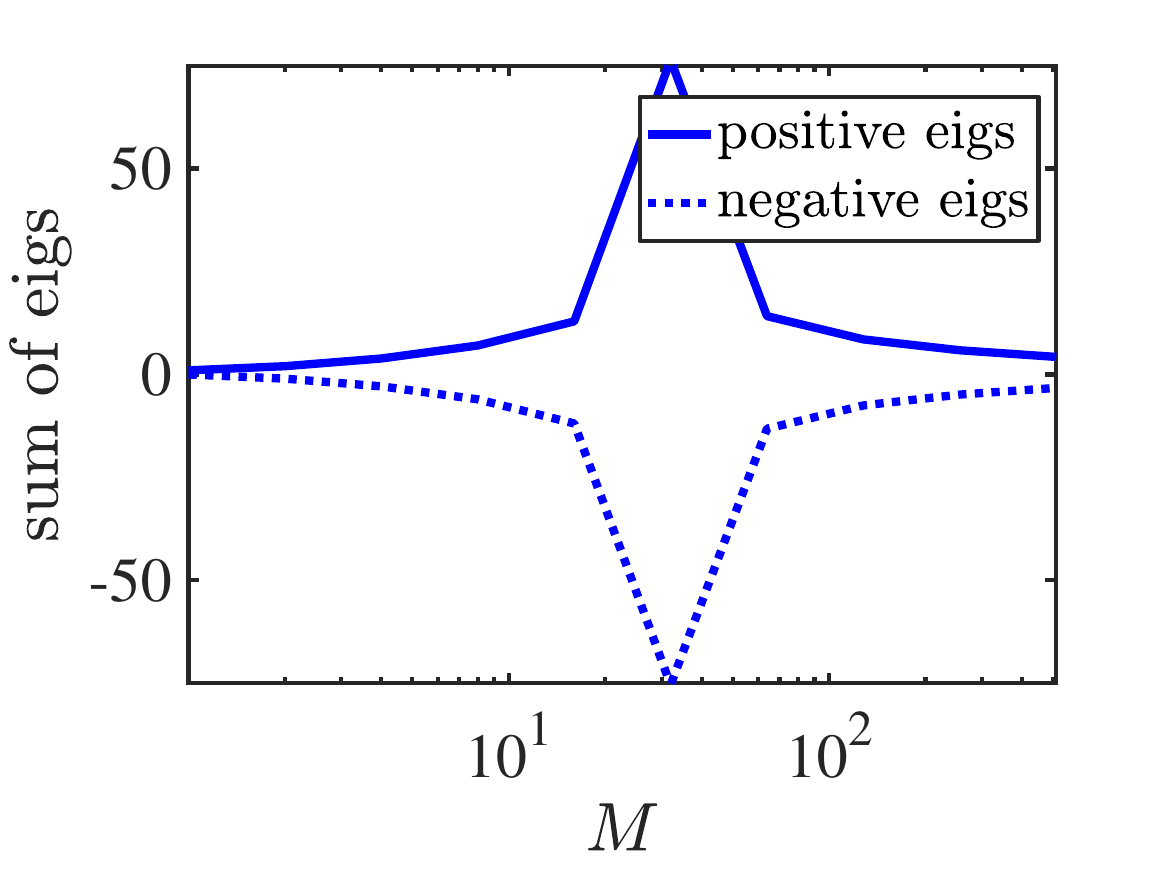}
\end{subfigure}
\begin{subfigure}{0.24\textwidth}
\centering
\includegraphics[width=\textwidth]{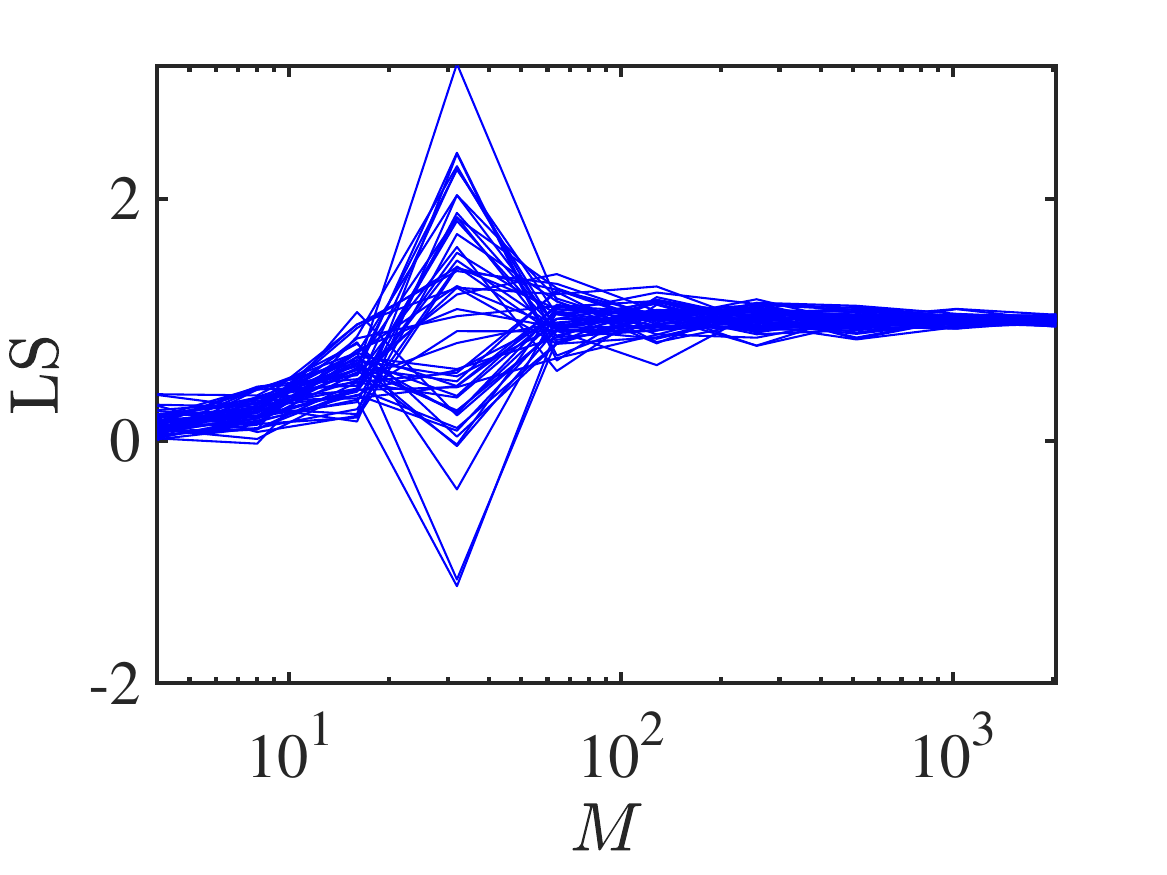}
\end{subfigure}
\begin{subfigure}{0.24\textwidth}
\includegraphics[width=\textwidth]{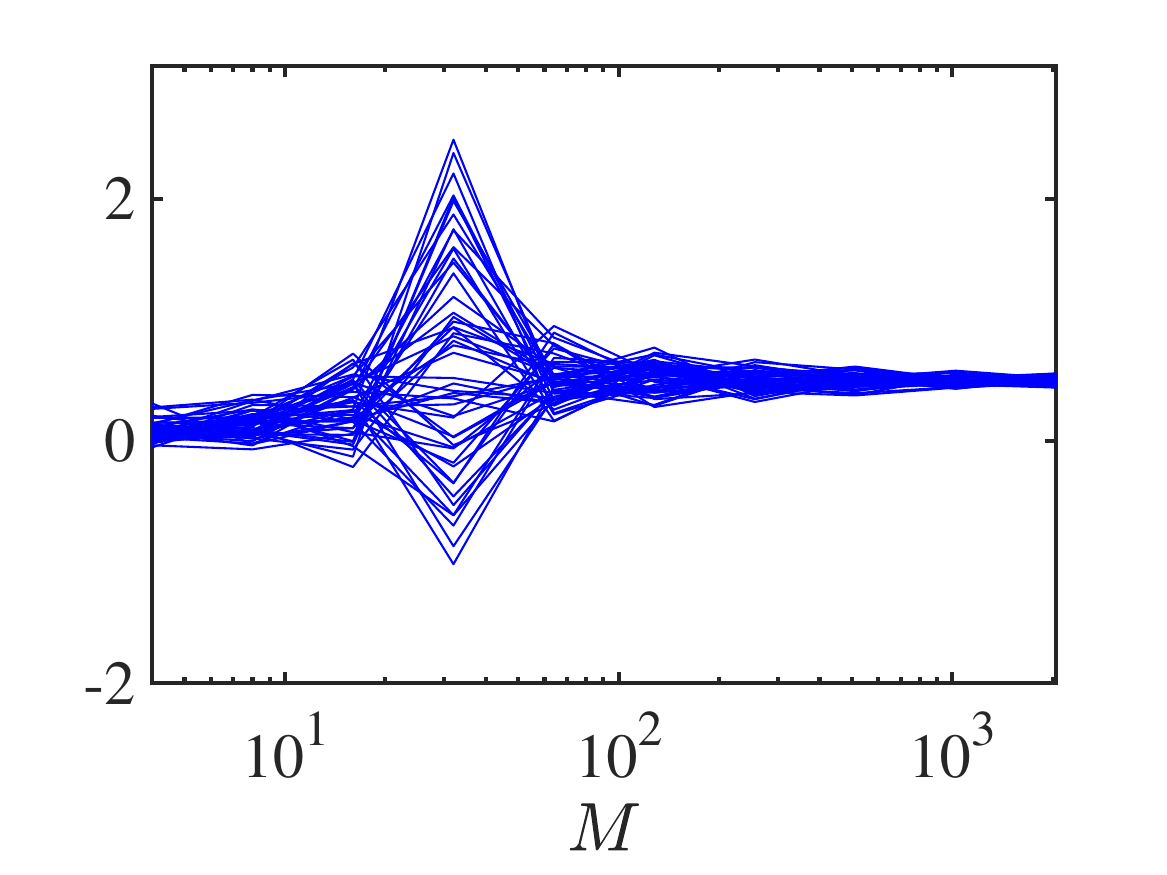}
\end{subfigure}
\begin{subfigure}{0.24\textwidth}
\includegraphics[width=\textwidth]{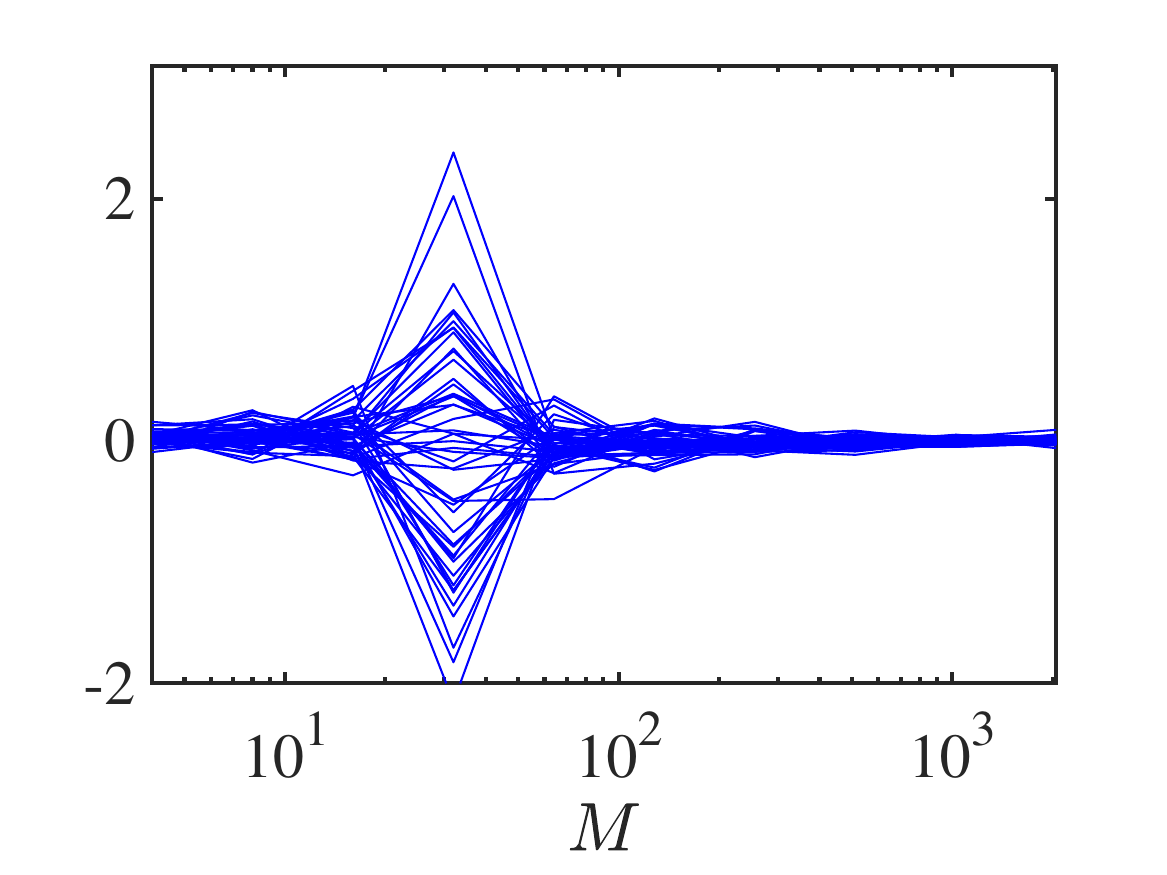}
\end{subfigure}
\vfill
\begin{subfigure}{0.24\textwidth}
\includegraphics[width=\textwidth]{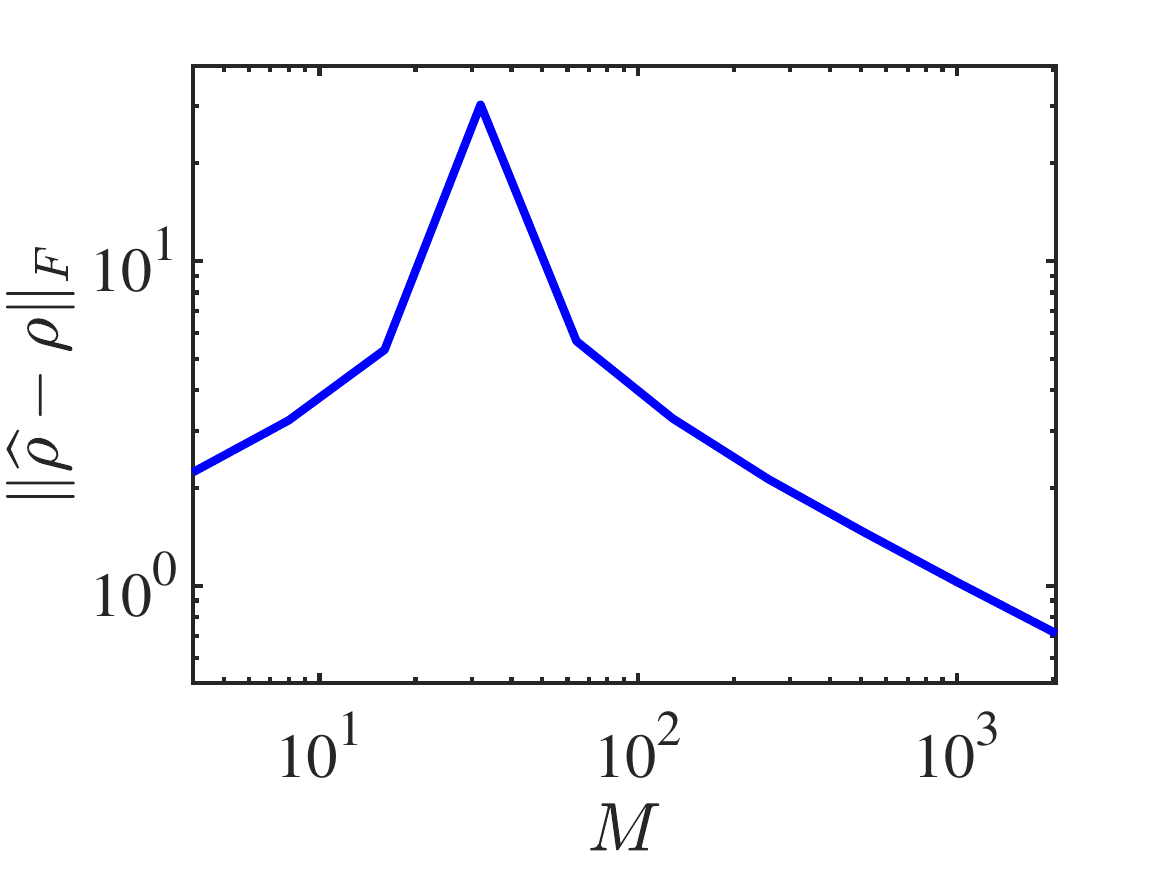}
\caption{$\wh \vrho$}
\end{subfigure}
\begin{subfigure}{0.24\textwidth}
\includegraphics[width=\textwidth]{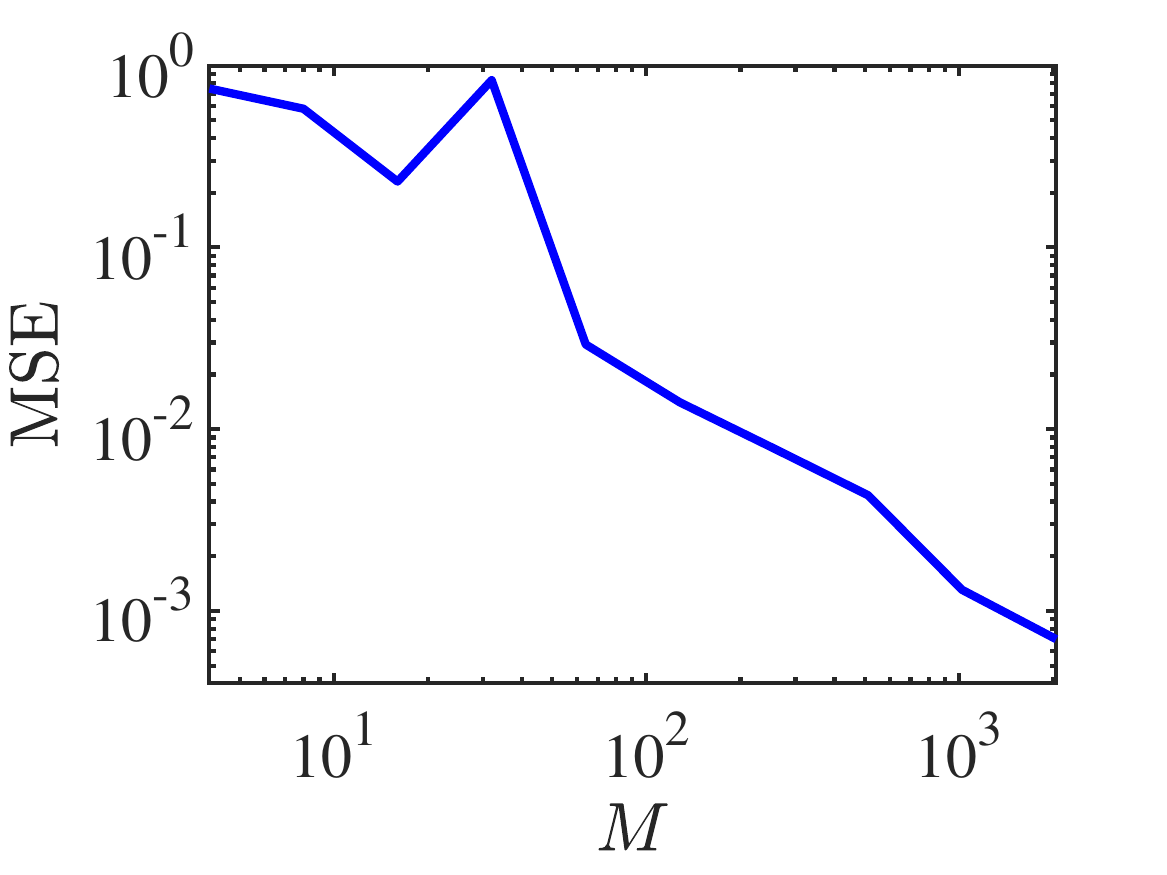}
\caption{$\wh\lambda_0$}
\end{subfigure}
\begin{subfigure}{0.24\textwidth}
\includegraphics[width=\textwidth]{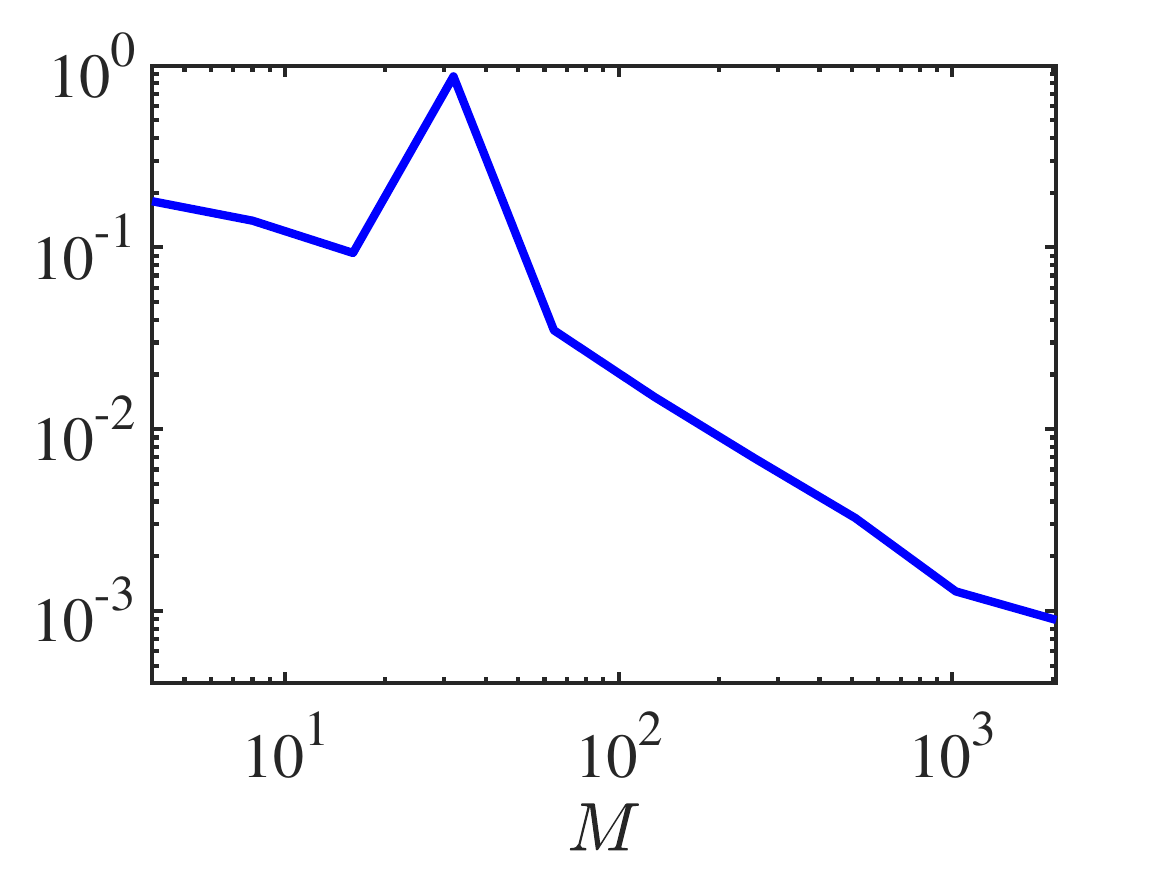}
\caption{$\wh\lambda_1$}
\end{subfigure}
\begin{subfigure}{0.24\textwidth}
\includegraphics[width=\textwidth]{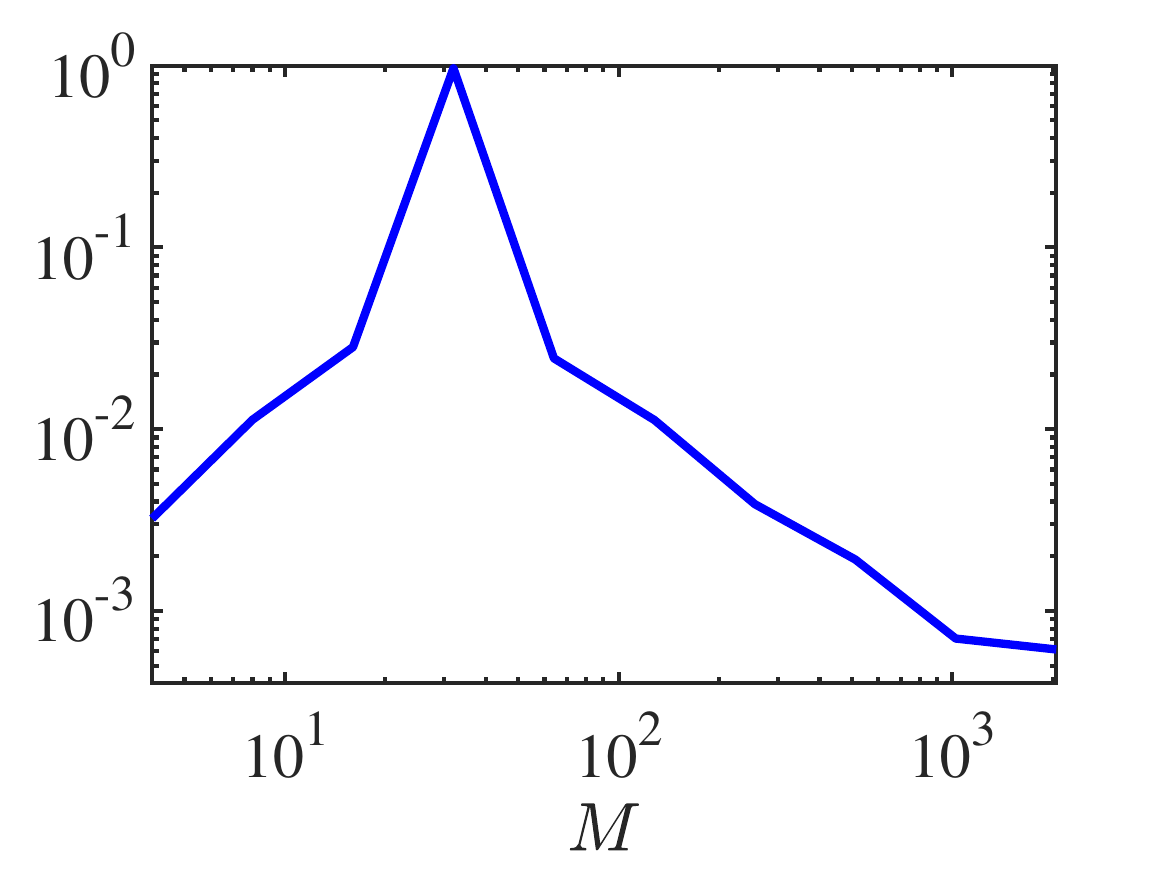}
\caption{$\wh\lambda_2$}
\end{subfigure}
\caption{Illustration of the performance of the LS shadow for estimating the state $\vrho$ and the linear observables $\lambda_i = \trace{\mLambda_i \vrho}$ with  $\mLambda_i = \vphi_i \vphi_i^\dagger$, where $\vphi_0 = \ve_0, \vphi_1 = \frac{1}{\sqrt{2}}\ve_0 + \frac{1}{\sqrt{2(D-1)}} \sum_{j=1}^{D-1}\ve_j, \vphi_2 = \ve_1$: (a) the sum of positive eigenvalues and negative eigenvalues of $\wh \vrho$, and $\|\wh \vrho - \vrho\|_F$, (b-d) $\wh \lambda_i$ (estimator for $\lambda_i$) from 50 independent trials, and the corresponding MSE $(\wh \lambda_i - \lambda_i)^2$ averaged over the 50 trials. 
}
\label{fig:LS}
\end{figure*}

\textit{LS and ``double descent'' phenomena.---}We simulate experiments over 50 independent trials, in each of which we compute the LS estimator for ten collections of measurements $M\in\{2^2,2^3,\ldots,2^{11}\}$.  Notably, Fig.~\ref{fig:LS} shows that the estimation errors for both the state $\vrho$ and the linear observables $\lambda_i$ do not monotonically decrease with the number of POVMs $M$. In particular, for $\vrho$ and $\lambda_2$, when $M$ increases, the estimation error first increases in the underdetermined regime ($M < D$), peaks at the interpolation regime when $M \approx D$, and then decreases when entering the overdetermined regime ($M > D$). Here, the three regimes are defined according to the relationship between the total number of outcomes $MD$ and the size of the state $D^2$, corresponding to the number of equations and parameters in the LS problem [Eq.~\eqref{eq:ls-problem}]. %
The curves of the estimation error for $\lambda_0$ and $\lambda_1$ first decrease, then increase, likewise peaking in the interpolating regime ($M\approx D$), and finally decrease with $M$. This resembles the ``double descent'' phenomenon observed in deep neural networks: performance first improves, then gets worse, and then improves again with increasing model or data size \cite{belkin2019reconciling,nakkiran2021deep,yang2020rethinking,singhphenomenology,chen2024gibbs}. This phenomenon has been formally studied for linear regression problems under certain statistical models~\cite{nakkiran2019more,belkin2020two,bartlett2020benign,sonthalia2023under,curth2024u}, and has recently been observed in neural networks for quantum state tomography~\cite{Smith2021} and polarimetry of vector beams~\cite{Pierangeli2023}.

Roughly speaking, in the interpolating regime $M\approx D$, $(\calA^\dagger \calA)^{+}$ is unstable (with very large singular values) and the LS estimator $\wh \vrho$ becomes highly nonphysical, i.e., it has large negative eigenvalues, resulting in large errors with respect to the ground truth [Fig.~\ref{fig:LS}(a)]. This appears in Fig.~\ref{fig:LS}(b--d) on the estimators $\wh\lambda_i$ as well, which become unstable around the interpolating regime, varying widely between trials. A formal analysis of this phenomenon is beyond the scope of the present investigation and is reserved for future work. 
Nonetheless, its presence in the context of the LS shadow estimator forms an important springboard for the techniques described in the following section. Both RLS and CS estimation procedures mitigate the issue of double descent by replacing the pseudoinverse $(\calA^\dagger \calA)^{+}$ in the LS shadow calculation $\mathcal{S}(\cdot)$ with a stabilized alternative: the similarities---and differences---between the RLS and CS solutions in turn reveal an interesting picture of CS estimation as a complementary and computationally efficient ``regularizer'' for the LS shadow.

\section{Stabilizing the LS Estimator}
\label{sec:stabilizing}

\subsection{RLS Estimation}
\label{sec:RLS}
In the underdetermined regime where the number of tested outcomes $MD$ is smaller than the size of the state $D^2$, regularization has been widely adopted for constraining the resulting solution. In the literature of quantum state tomography, regularization or constraint has been exploited for stable, low-measurement reconstruction under the assumption of specific structural features---such as low-rank states \cite{gross2010quantum,liu2011universal,KuengACHA17, haah2017sample,guctua2020fast,francca2021fast} and matrix product states and operators \cite{verstraete2004matrix,pirvu2010matrix,cramer2010efficient,baumgratz2013scalable,lidiak2022quantum,noh2020efficient,wang2020scalable,jarkovsky2020efficient,qin2023stable}. Without assuming any particular state structure,%
a common regularization 
in statistics and machine learning is $\ell_2$ regularization, resulting in the so-called RLS or ridge-regression estimator. 
This has also been widely used in quantum state tomography~\cite{opatrny1997least,mu2020quantum}.
The $\ell_2$ regularization often leads to a dense solution; in the context of quantum states, it tends to push toward mixed states (lower purity). Specifically, for a given $\mu\ge 0$,
\begin{align}
\begin{split}
\wh \vrho &= \argmin_{\vrho' \in \C^{D\times  D}} \left\{ \norm{\wh \vp - \calA(\vrho')}{{2}}^2 +\mu \norm{\vrho'}{F}^2 \right\}\\
& = (\calA^\dagger \calA + \mu\mId)^{-1}\calA^\dagger(\wh \vp) \\ & = \frac{1}{M} \sum_{m=1}^M  \underbrace{\left(\frac{1}{M}(\calA^\dagger \calA + \mu\mId)\right)^{-1}\calA_m^\dagger(\wh \vp_m)}_{\text{RLS shadow}}.\label{eq:shadow-v2}
\end{split}
\end{align} 
With the introduced $\ell_2$ regularization, $\calA^\dagger \calA + \mu\mId$ is invertible and its condition number decreases as $\mu$ increases, achieving the purpose of stabilization. Following the discussion of LS shadows, we may call the induced $\wh \vrho_m \vcentcolon= \mathcal{S}\left(\calA_m^\dagger(\wh\vp_m)\right)= \left(\frac{1}{M}(\calA^\dagger \calA + \mu\mId)\right)^{-1}\calA_m^\dagger(\wh \vp_m)$ the ``RLS shadow.'' Similar to LS shadows, these RLS shadows are biased estimators of the ground truth $\vrho$.

\begin{figure*}[t]
\centering
\begin{subfigure}{0.24\textwidth}
\centering
\includegraphics[width= \textwidth]{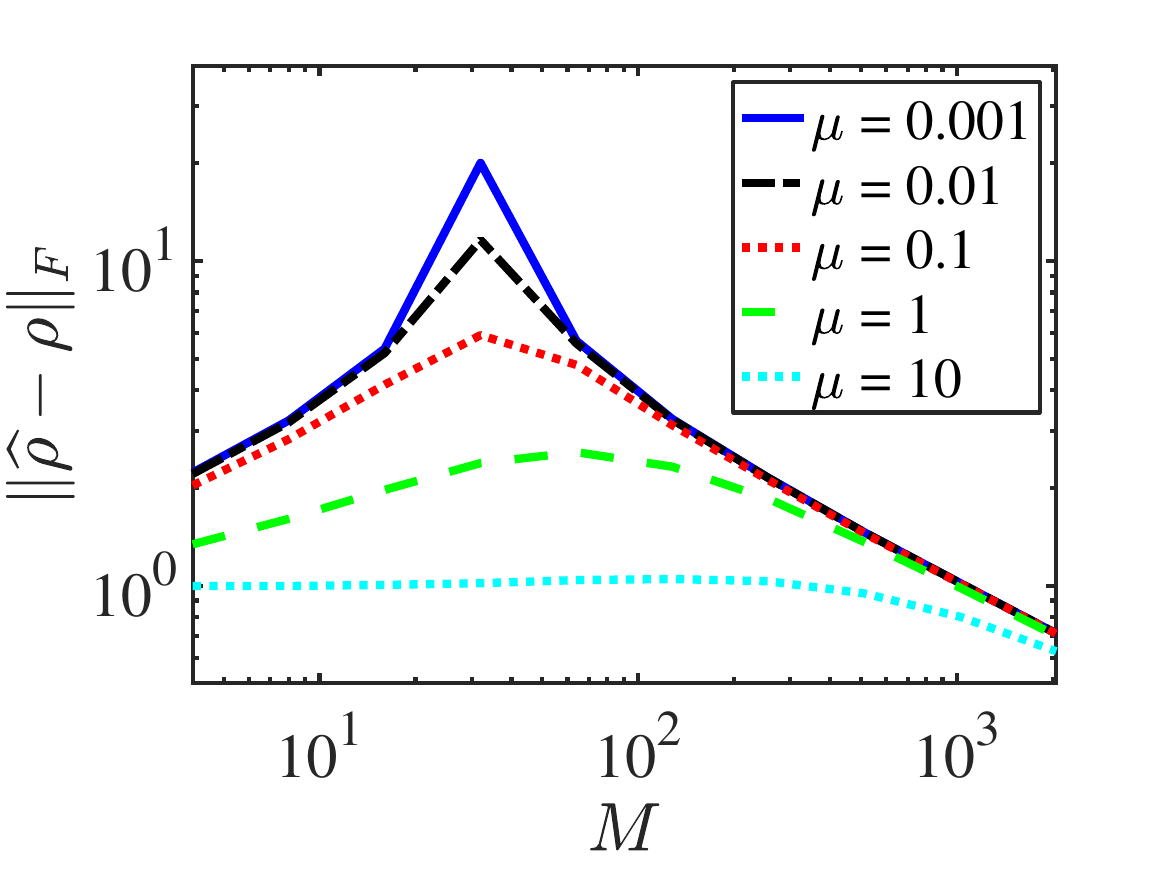}
\caption{$\wh \vrho$}
\end{subfigure}
\begin{subfigure}{0.24\textwidth}
\centering
\includegraphics[width=\textwidth]{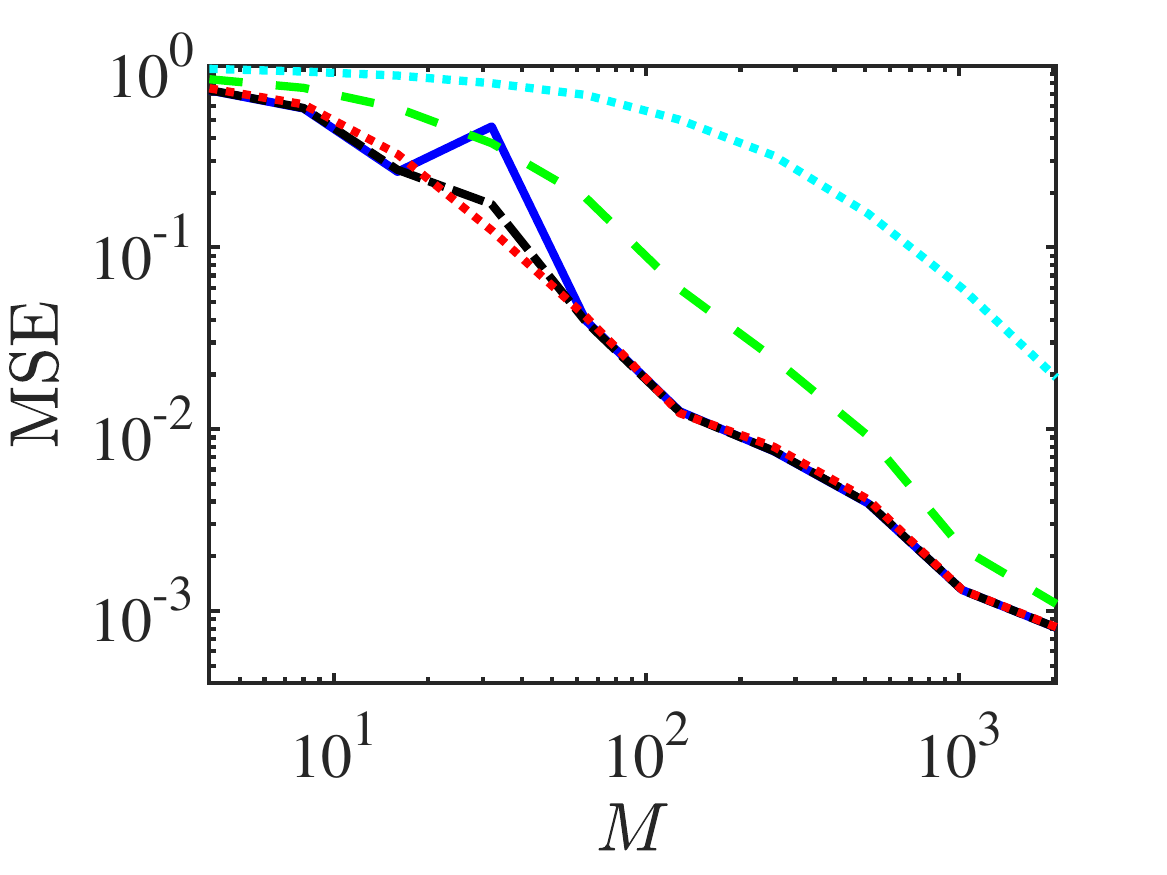}
\caption{$\wh\lambda_0$}
\end{subfigure}
\begin{subfigure}{0.24\textwidth}
\includegraphics[width=\textwidth]{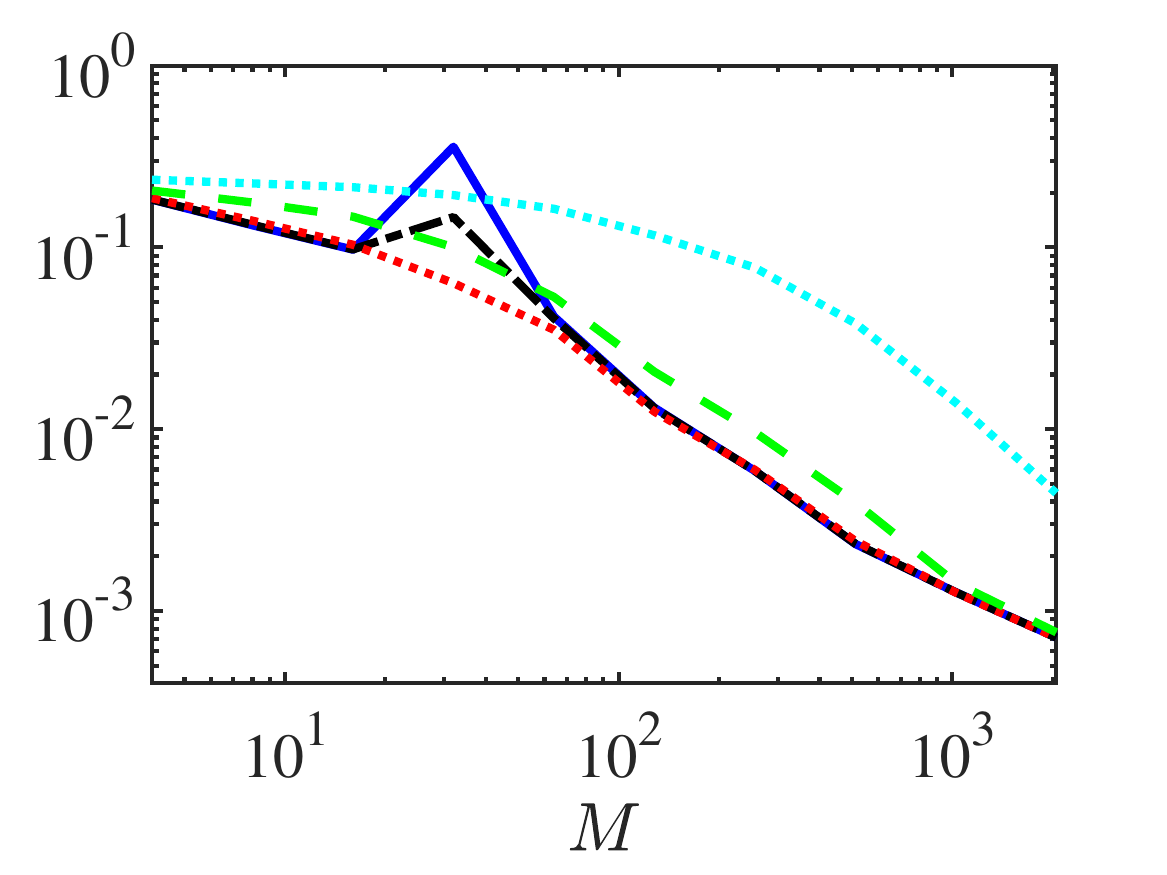}
\caption{$\wh\lambda_1$}
\end{subfigure}
\begin{subfigure}{0.24\textwidth}
\includegraphics[width=\textwidth]{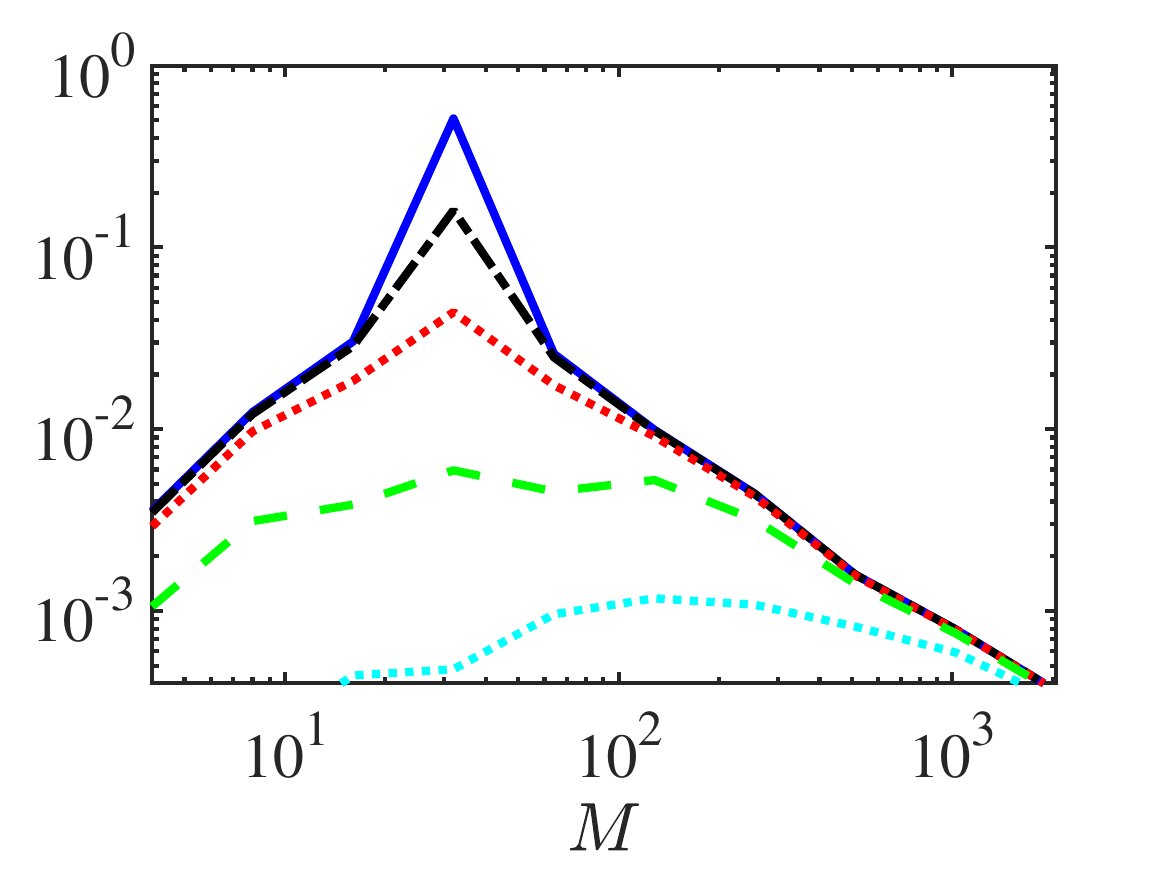}
\caption{$\wh\lambda_2$}
\end{subfigure}
\caption{Illustration of the performance of the RLS shadow with different regularization parameter $\mu$ for estimating the state $\vrho$ and the three linear observables as in Fig.~\ref{fig:LS}.
}
\label{fig:RLS-mu}
\end{figure*}

Choosing a suitable regularization parameter $\mu$ requires balancing stabilization of the operator $\calA^\dagger \calA + \mu\mId$ and fitting the measurements, for which techniques like cross-validation or grid search can be used. 
Recent work~\cite{nakkiran2020optimal} shows that an optimal regularization (which may vary with $M$) that minimizes Frobenius error $\|\wh \vrho-\vrho\|_F$ can mitigate the double descent phenomenon for linear regression under certain statistical assumptions. While such a formal strategy seems impractical here since the state $\vrho$ is unknown \emph{a priori}, numerical experiments demonstrate that a small $\mu$ can indeed make RLS estimators stable. As shown in Fig.~\ref{fig:RLS-mu}, $\mu = 0.1$ leads to monotonically decreasing MSEs for $\lambda_0$ and $\lambda_1$. While the MSE for $\lambda_2$ still peaks in the interpolating regime ($M\approx D$), it is already small (comparable to those for $\lambda_0$ and $\lambda_1$) and significantly smaller than the one achieved by LS. We could achieve monotonically decreasing estimation error for both $\lambda_2$ and $\vrho$ through even larger values of $\mu$ (e.g., the $\mu=1$ case in Fig.~\ref{fig:RLS-mu}). However, this will substantially bias the estimator $\wh \vrho$ to zero, leading to worse estimation for $\lambda_0$ and $\lambda_1$. In order to balance the performance across different observables, throughout all subsequent experiments, we simply take $\mu = 0.1$ for RLS.

\subsection{CS Estimation}
\label{sec:CS}
CS estimation utilizes a different approach to stabilize the pseudoinverse $(\frac{1}{M}\calA^\dagger \calA)^{+}$. Assume each POVM $\{\mA_{m,k}\}_{k\in[K]}$ is independently and randomly generated from an ensemble of POVMs $\setA$ according to a certain probability distribution $P(\setA)$. We may approximate $\frac{1}{M}\calA^\dagger \calA$ by its expectation
\begin{align}
\begin{split}
\calM(\vrho) & =\E\left[\frac{1}{M}\calA^\dagger \calA(\vrho)\right] \\ & = \E_{\{\mA_{k}\}\sim P(\setA) } \left[\sum_{k=1}^K  \innerprod{\mA_{k}}{\vrho} \mA_{k} \right],
\label{eq:Exp-AtA}
\end{split}
\end{align}
where $\{\mA_k\}$ represents a random POVM generated from the ensemble of POVMs $\setA$ according to the probability distribution $P(\setA)$.
Here $\calM$ is called the quantum channel. If $\setA$ is tomographically complete, then $\E[\calA^\dagger \calA]$ is full rank and invertible. The shadow in CS introduced in Ref.~\cite{huang2020predicting} can then be defined by replacing $\frac{1}{M}\calA^\dagger \calA$ in Eq.~\eqref{eq:shadow-v1} with its expectation $\E[\calA^\dagger \calA(\vrho)]$, i.e.,
\e
\wh \vrho = \frac{1}{M} \sum_{m=1}^M \underbrace{\calM^{-1}\left(\calA_m^\dagger(\wh \vp_m)\right)}_{\text{CS shadow}},
\label{eq:shadow-v3}
\ee
so that we obtain the ``CS shadow'' $\wh\vrho_m \vcentcolon= \mathcal{S}\left(\calA_m^\dagger(\wh \vp_m)\right) = \calM^{-1}\left(\calA_m^\dagger(\wh \vp_m)\right)$, for which the original CS proposal is named. (We acknowledge the inherent repetitiveness of the term ``CS shadow''---``classical shadows shadow''---but adopt it for consistency with LS shadow and RLS shadows.)

As $\calM$ defined in Eq.~\eqref{eq:Exp-AtA} is a linear operator, its inverse $\calM^{-1}$ is also linear.  These CS shadows are independent and unbiased estimators of $\vrho$. Specifically, noting that the randomness of each shadow comes from two sources---the randomly selected POVM $\calA_m$ %
and the random experimental outcome $\wh \vp_{m}$---we can take the expectation to obtain
\begin{widetext}
\begin{align}
\E_{\{\mA_{m,k}\}\sim P(\setA),\wh \vp_m} \left[\calM^{-1}\parans{\calA_m^\dagger (\wh \vp_m) } \right]  = \calM^{-1} \underbrace{ \E_{\{\mA_{m,k}\}\sim P(\setA),\wh \vp_m} \left[\parans{\sum_{k=1}^K \wh p_{m,k} \mA_{m,k}} \right] }_{\calM(\vrho)}  = \vrho,
\end{align}
\end{widetext}
where $ \E_{\{\mA_{m,k}\}\sim P(\setA),\wh \vp_m} \left[\sum_{k=1}^K \wh p_{m,k} \mA_{m,k} \right] = \calM(\vrho)$ can be obtained by noting the conditional expectation $\E_{\wh \vp_m } [\wh p_{m,k}\mA_{m,k} \mid \{\mA_{m,k}\}] = \trace(\mA_{m,k}\vrho)\mA_k$ according to the Born rule [see Eqs.~\eqref{The defi of POVM 2} and \eqref{eq:empirical-prob}].

At this point it is useful to pause and compare the three estimation approaches introduced and analyzed so far: LS, RLS, and CS. As articulated in Fig.~\ref{fig:concept} and Eqs.~(\ref{eq:shadow-v1}, \ref{eq:shadow-v2}, \ref{eq:shadow-v3}), each approach produces an estimate that can be viewed as an average over discrete shadows $\wh\vrho_m=\mathcal{S}\left(\calA_m^\dagger(\wh\vp_m)\right)$ each corresponding to one of the $M$ POVMs measured. Whereas LS computes each shadow by direct inversion of the total collection of measurements, both RLS and CS modify this procedure, through regularization and quantum channel inversion, respectively. This observation already offers interesting insights into CS features.

In our opinion, one of the initially most surprising aspects of CS lies in the way it treats the measurement operators post-experiment. Although the POVMs $\calA_m$ are selected at random during the measurement process in the canonical CS example, they are known to the user \emph{a posteriori} through the complete collection $\calA$. Yet this knowledge is intentionally ignored in the CS shadow operation $\mathcal{S}(\cdot)$; the inversion is instead performed on the \emph{a priori} quantum channel with completely random measurements---an essentially ``fictitious'' quantum channel from the perspective of the completed experiment. In the light of RLS, however, this channel selection acquires a more intuitive explanation in terms of stabilization: like RLS, CS allows for well-conditioned inversion under any set of measurements, opening the opportunity to improve stability in the estimation procedure and qualitatively accounting for the rigorous information-theoretic bounds it attains~\cite{huang2020predicting}.

\begin{table*}[tbh!]
\centering
\caption{Comparison of RLS and CS estimators.}
\label{tab:my-table}
\begin{tabular}{c|c|c|c|c}
 & \begin{tabular}[c]{@{}c@{}}distribution\\ independent\end{tabular} & \begin{tabular}[c]{@{}c@{}}computational\\ cost\end{tabular} & bias & variance \\ \hline
LS & \cmark & high & high & high \\ \hline
RLS & \cmark & high & high & low \\ \hline
CS & \xmark & low & low (zero) & high
\end{tabular}
\label{table:comparison}\end{table*}

\textit{CS with rank-1 POVMs.---}
The quantum channel $\calM$ in Eq.~(\ref{eq:Exp-AtA}) depends on the POVM ensemble and the corresponding sampling distribution. Consider rank-1 POVMs of the form $\{\mA_{1},\ldots,\mA_{d} \}$ with $\mA_{k} = \vu_{k}\vu_{k}^\dagger$, where each $\mU = \begin{bmatrix} \vu_{1} & \cdots & \vu_{d} \end{bmatrix}^\dagger$ is randomly chosen from an ensemble of $D\times  D$ unitary matrices $\setU$.  Various unitary ensembles have been explored in prior studies, including the local and global Clifford ensembles \cite{huang2020predicting}, fermionic Gaussian unitaries \cite{zhao2021fermionic}, chaotic Hamiltonian evolutions \cite{hu2023classical}, locally scrambled unitary ensembles \cite{hu2023classical}, and Pauli-invariant unitary ensembles \cite{bu2022classical}, for which explicit formulas for the quantum channel $\calM$ and its inverse $\calM^{-1}$ exist. %
For instance, if $\setU$ is the full unitary group and each unitary matrix $\mU$ is sampled independently according to the Haar measure on $\setU$, %
the quantum channel $\calM$ and its inverse $\calM$ can be computed as
\begin{align}
\begin{split}
 \calM(\vrho) & = \E_{\mU} \left[\sum_{k=1}^{D} \parans{\vu_{k}^\dagger \vrho \vu_{k} }\vu_k \vu_k^\dagger   \right] \\ & = \frac{1}{D+1}\vrho + \frac{\trace(\vrho)}{D+1}\mId,\label{eq:expectation-M}
\end{split}
\end{align}
\begin{equation}
 \calM^{-1}(\vrho) = (D + 1)\vrho - \trace(\vrho)\mId.
\label{eq:M-inverse}
\end{equation}
Plugging this explicit formula and the expression of $\calA^\dagger_m(\wh \vp_m)$ in Eq.~\eqref{eq:Aadjoint-rank-one} into Eq.~\eqref{eq:shadow-v2}, we can further simplify the CS shadow \cite{huang2020predicting}:
\e\begin{split}
 \wh \vrho_m & = (D + 1) (\mU_m^\dagger \wh \vp_m)(\mU_m^\dagger \wh \vp_m)^\dagger \\
 & \qquad- \trace((\mU_m^\dagger \wh \vp_m)(\mU_m^\dagger \wh \vp_m)^\dagger) \mId\\
 & = (D + 1) (\mU_m^\dagger \wh \vp_m)(\mU_m^\dagger \wh \vp_m)^\dagger - \mId,
\end{split}
\label{eq:shadow-unitary}
\ee
Since $(\mU_m^\dagger \wh \vp_m)(\mU_m^\dagger \wh \vp_m)^\dagger$ is rank-1 with its only non-zero eigenvalue equal to 1, each classical shadow satisfies $\trace(\wh\vrho_m)=1$, through the summation of one positive eigenvalue equal to $D$ and $D-1$ negative eigenvalues equal to $-1$.

\section{Comparing RLS and CS Stabilization Techniques}
\label{sec:Compare}
\subsection{Overview}
\label{sec:Overview}
As introduced and discussed in the previous section, both RLS and CS invoke ``fictitious'' quantum channels to stabilize the inverse operation in computing each shadow $\wh\rho_m$. Nevertheless, they do so in significantly different ways, leading to distinct advantages and disadvantages in specific use cases. To examine these aspects further, we test and compare RLS and CS methods in the estimation of quantum observables from three perspectives, each of which leverages targeted numerical simulations to reveal the  important behaviors of interest. We summarize the three features below and in Table~\ref{table:comparison}, and the relevant numerical simulations follow in the subsequent subsections.

\textit{Feature 1: bias and variance tradeoff.---}CS and RLS approaches trade off bias and variance in opposite ways. CS estimates are always unbiased but can exhibit relatively large variance with a limited number of measurements. On the other hand, RLS estimation controls variance through $\ell_2$ regularization, but also introduces bias.

\textit{Feature 2: handling of distribution mismatch.---}The quantum channel $\mathcal{M}$ in CS relies on information about how the POVMs are randomly generated. Such information is not necessary in RLS. In other words, RLS is more flexible as it only requires the POVMs actually used, regardless of whether they are generated randomly or deterministically. On the other hand, by averaging over all possible POVMs, the quantum channel in CS often has a simple explicit formulation that is independent of the specific POVMs measured, as shown in Eqs.~(\ref{eq:expectation-M}, \ref{eq:M-inverse}). In contrast, RLS must compute the inverse $(\mathcal{A}^\dagger\mathcal{A} + \mu\mId)^{-1}$ for each POVM realization. Thus, RLS and CS trade off flexibility in the distribution with computational efficiency.

\textit{Feature 3: scaling with multishot measurements.---}The measurement of $M$ POVMs with $L$ shots each requires a total of $ML$ state preparations. Although we focus primarily on the $L=1$ case---in line with the original CS formulation~\cite{huang2020predicting}---our derivations are completely generic with respect to $L$. In our third set of tests, we therefore examine the performance of both RLS and CS for multishot measurements ($L>1$). Although \emph{a priori} unclear to us whether RLS and CS shadows would show any differences in their respective dependencies on $L$, numerical tests in Sec.~\ref{sec:multishot} find that estimation errors in RLS are much more sensitive to $L$---both for better and worse---than their CS counterparts.

\subsection{Feature 1: Bias-Variance Tradeoff}
\label{sec:biasVariance}

\begin{figure*}[t]
\centering
\begin{subfigure}{0.24\textwidth}
\centering
\includegraphics[width=\textwidth]
{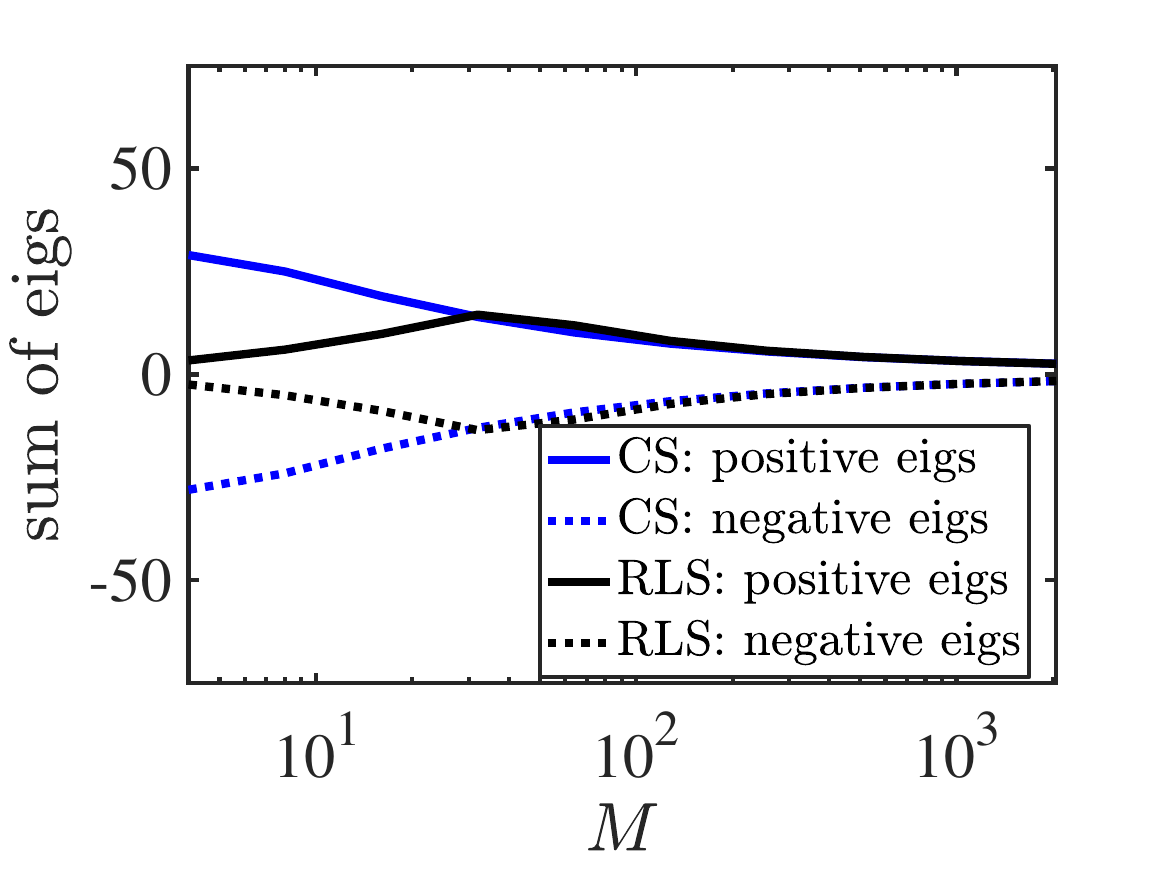}
\end{subfigure}
\begin{subfigure}{0.24\textwidth}
\centering
\includegraphics[width=\textwidth]{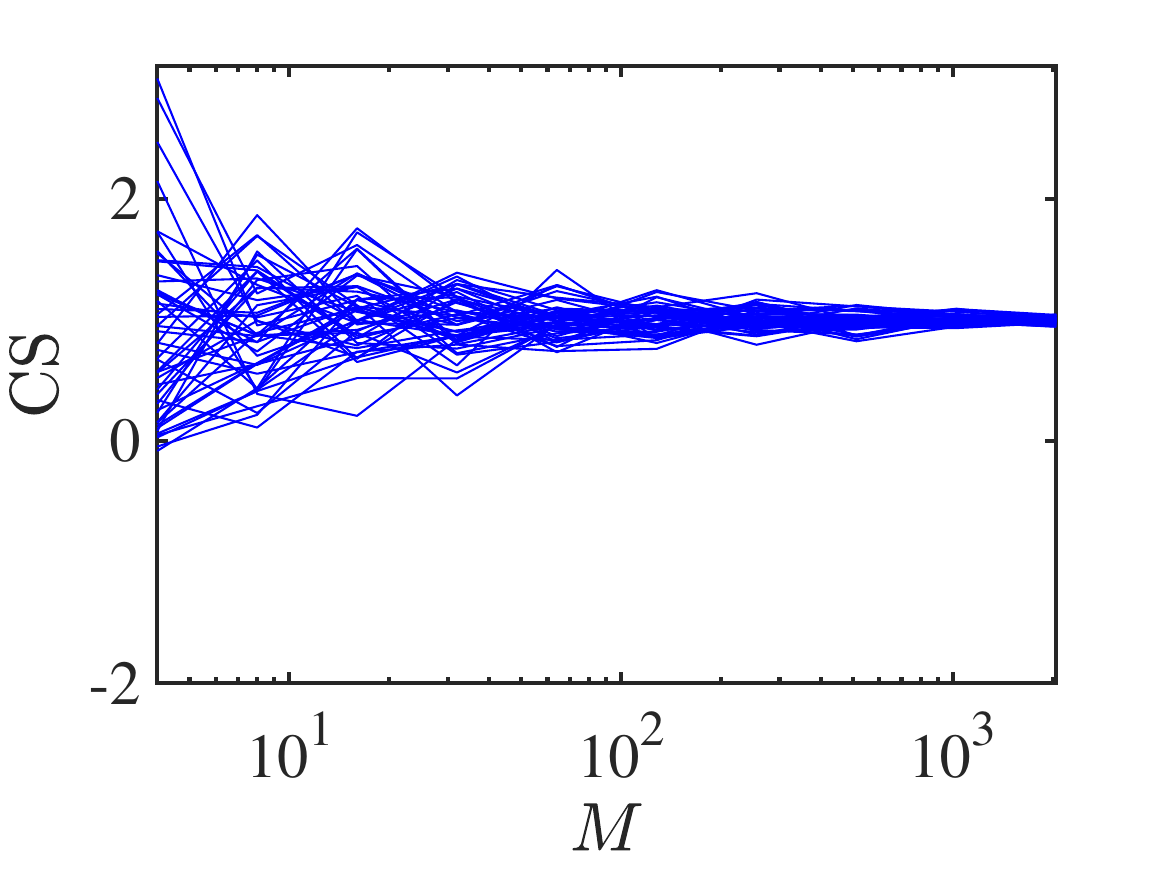}
\end{subfigure}
\begin{subfigure}{0.24\textwidth}
\includegraphics[width=\textwidth]{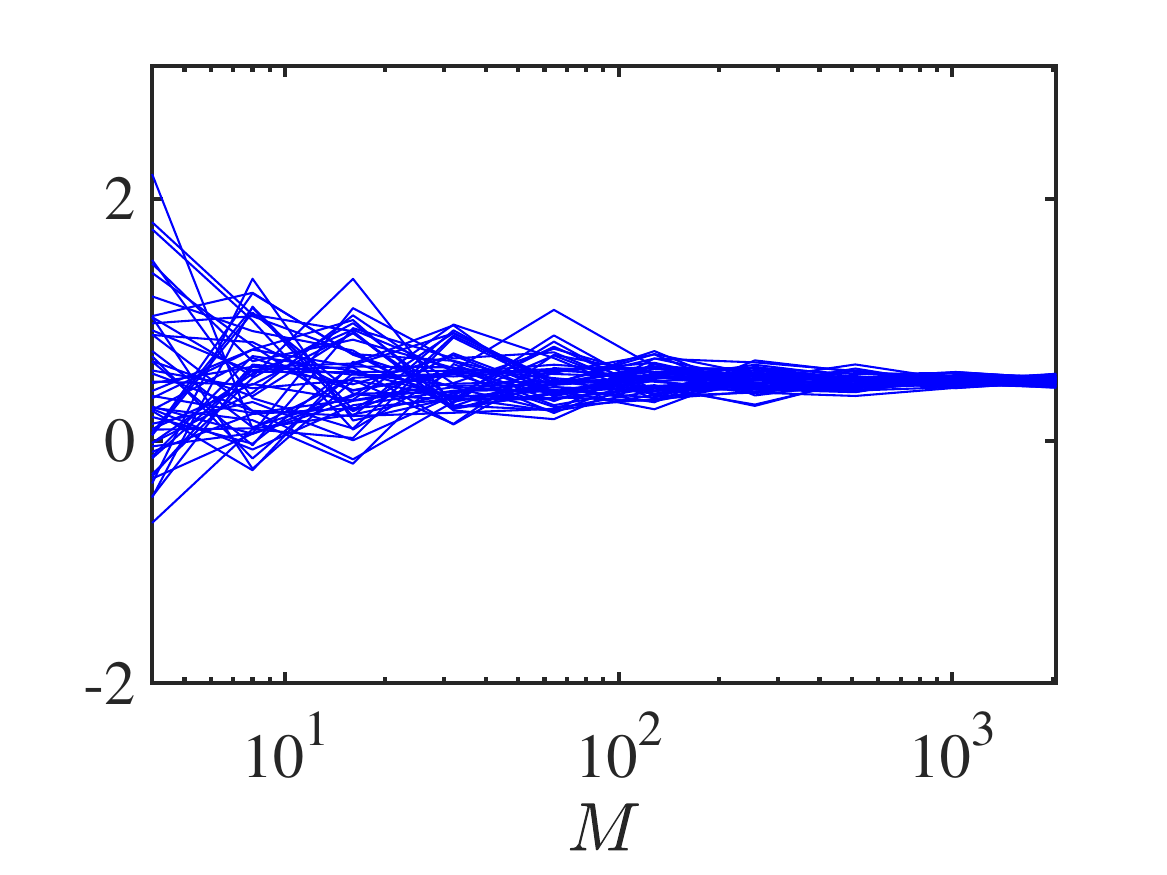}
\end{subfigure}
\begin{subfigure}{0.24\textwidth}
\includegraphics[width=\textwidth]{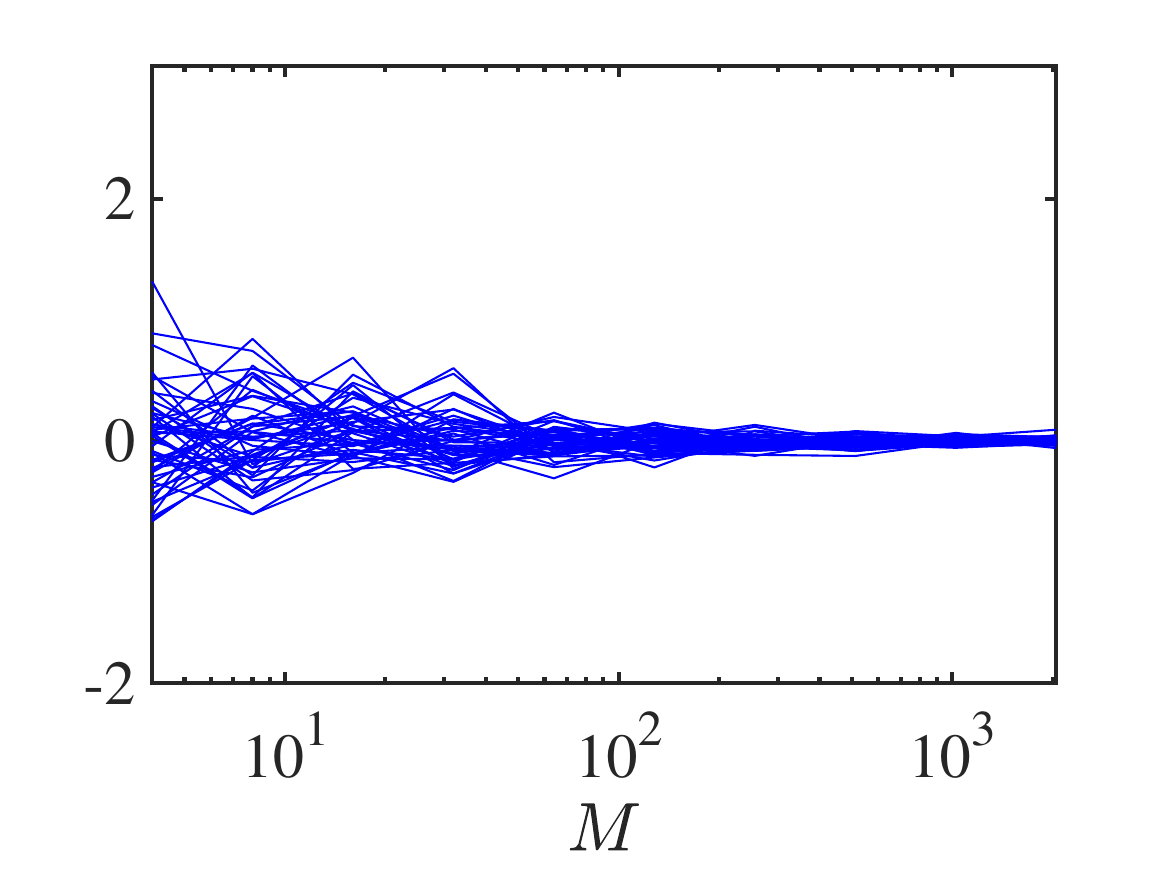}
\end{subfigure}
\vfill
\begin{subfigure}{0.24\textwidth}
\centering
\includegraphics[width= \textwidth]
{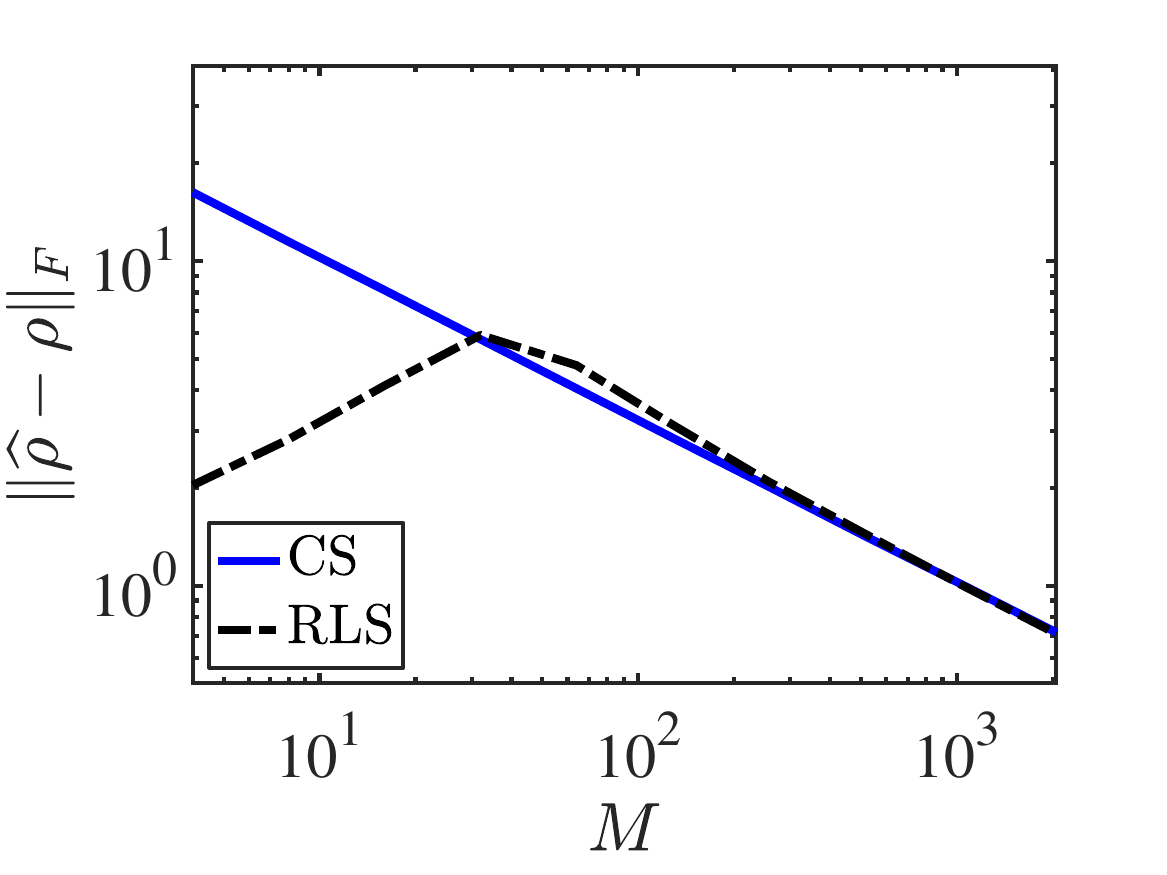}
\end{subfigure}
\begin{subfigure}{0.24\textwidth}
\centering
\includegraphics[width=\textwidth]{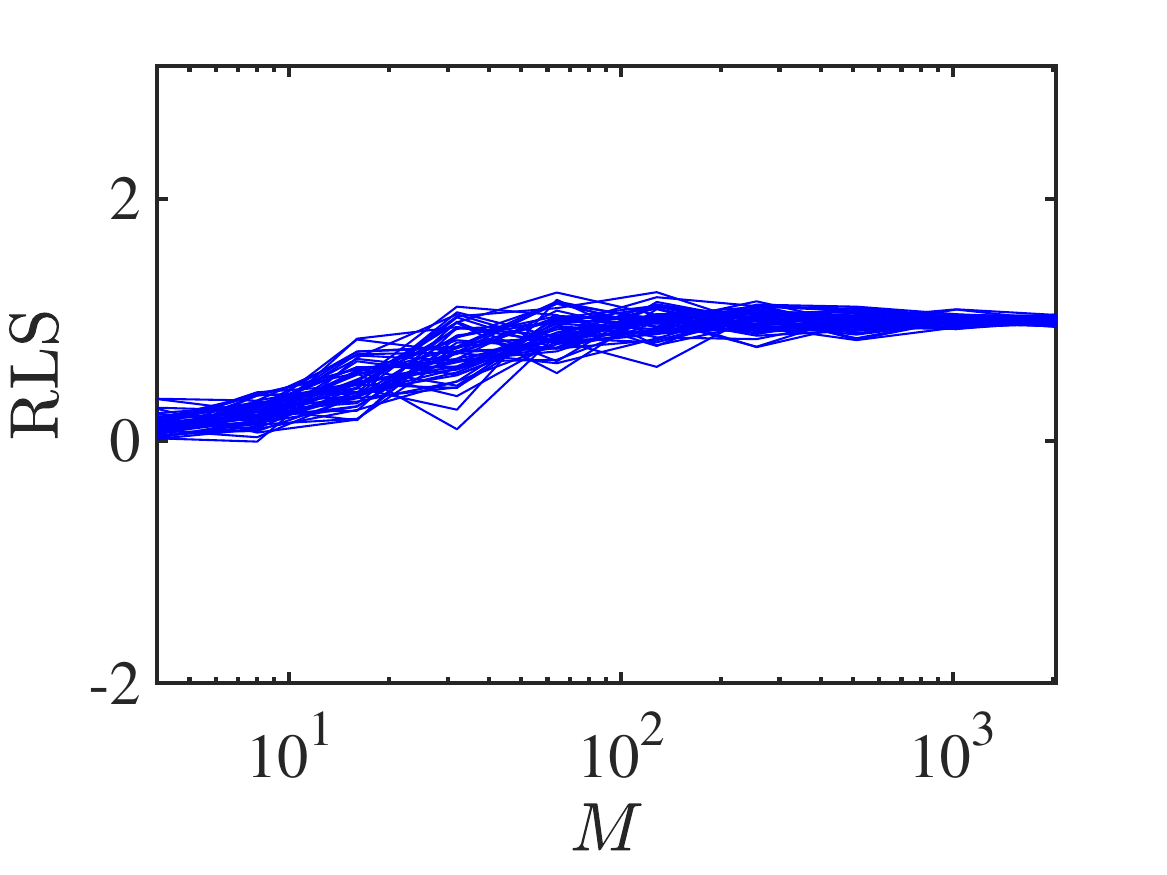}
\end{subfigure}
\begin{subfigure}{0.24\textwidth}
\includegraphics[width=\textwidth]{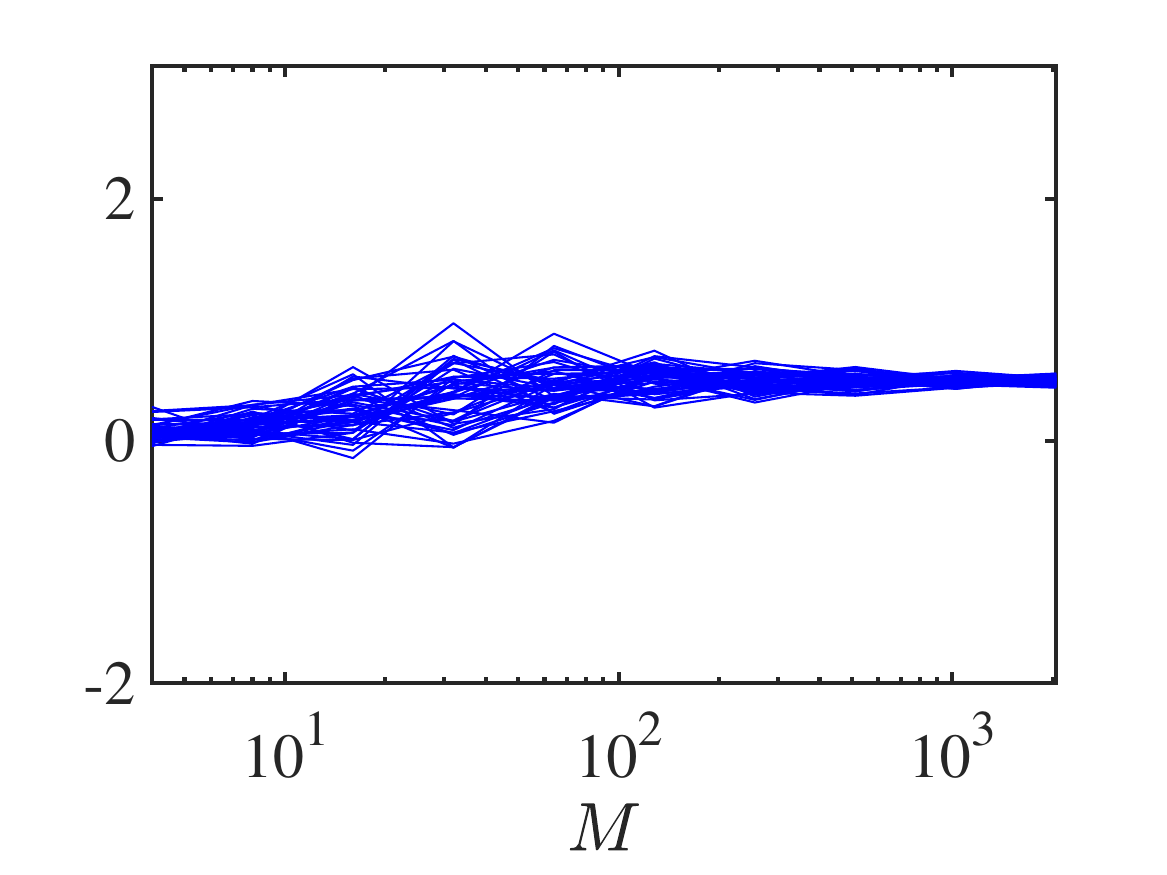}
\end{subfigure}
\begin{subfigure}{0.24\textwidth}
\includegraphics[width=\textwidth]{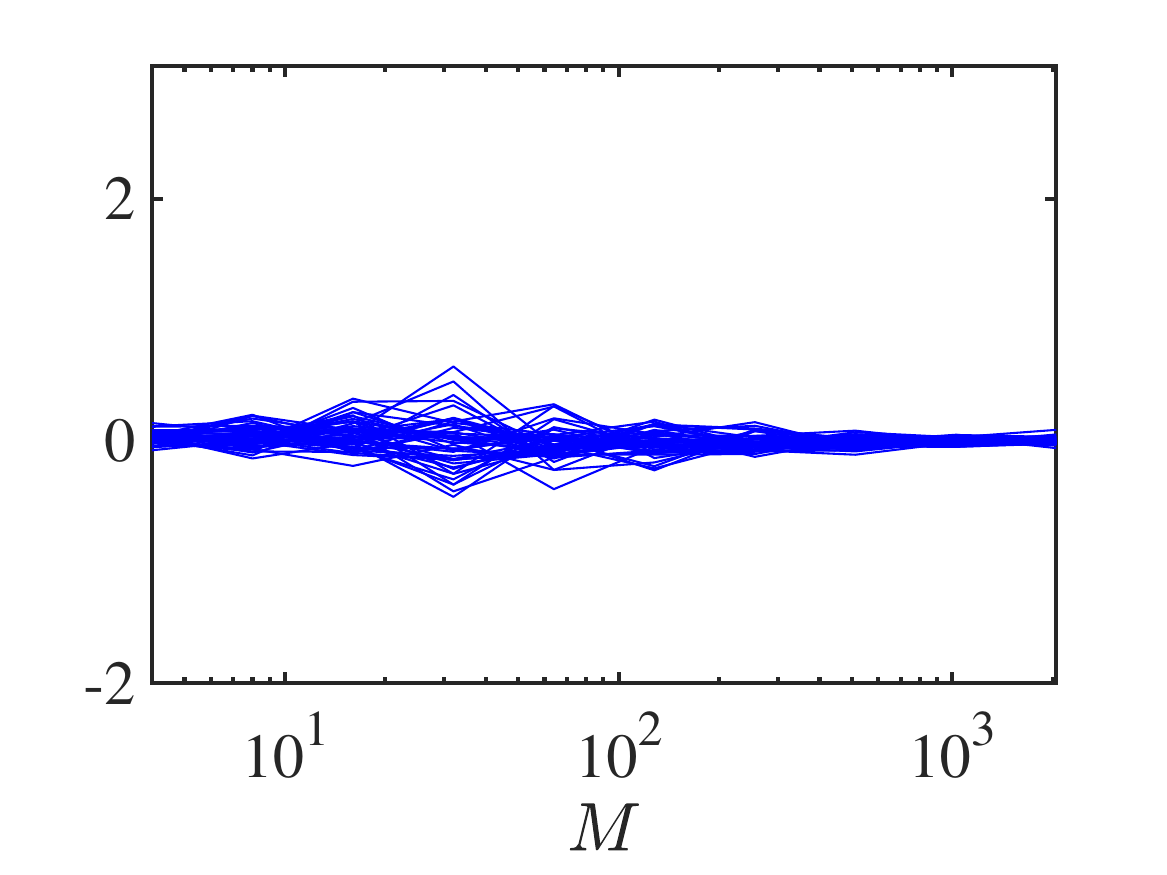}
\end{subfigure}
\vfill
\begin{subfigure}{0.24\textwidth}
\includegraphics[width=\textwidth]
{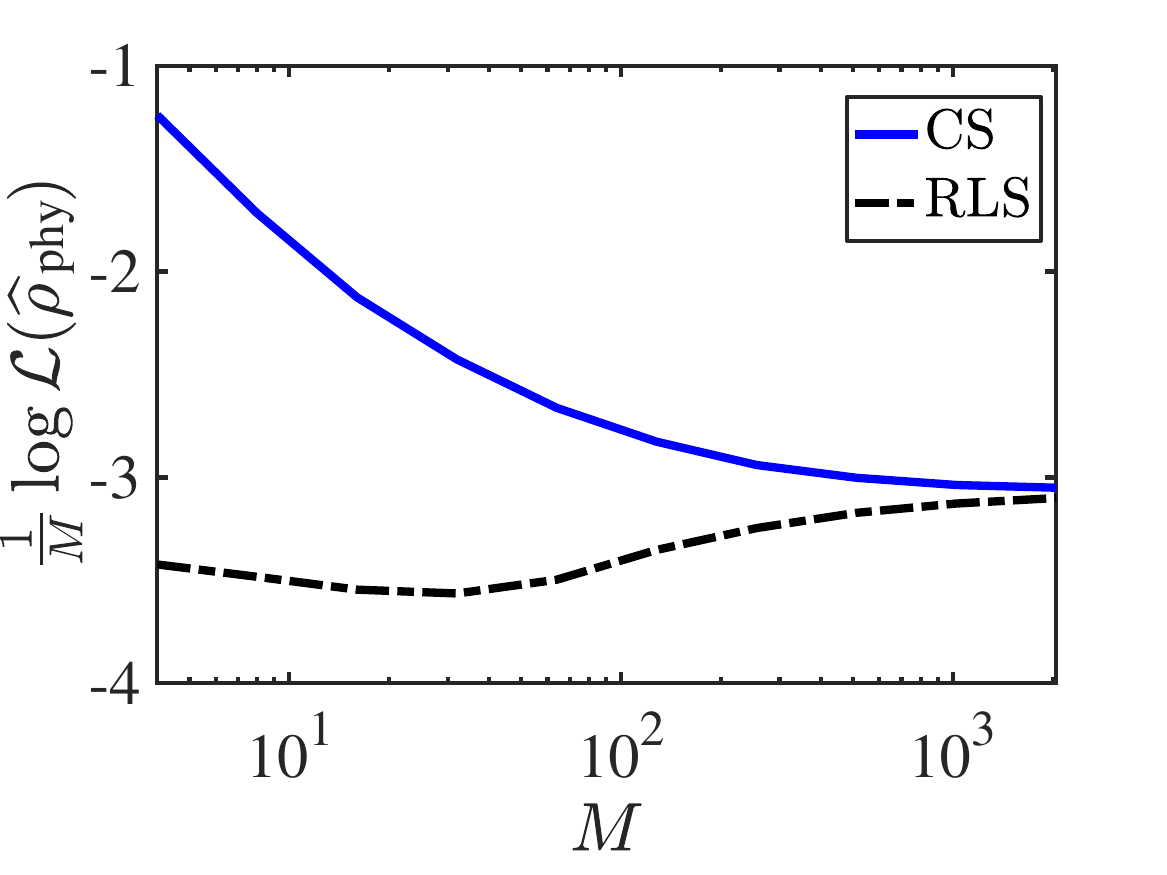}
\caption{$\wh \vrho$}
\end{subfigure}
\begin{subfigure}{0.24\textwidth}
\includegraphics[width=\textwidth]{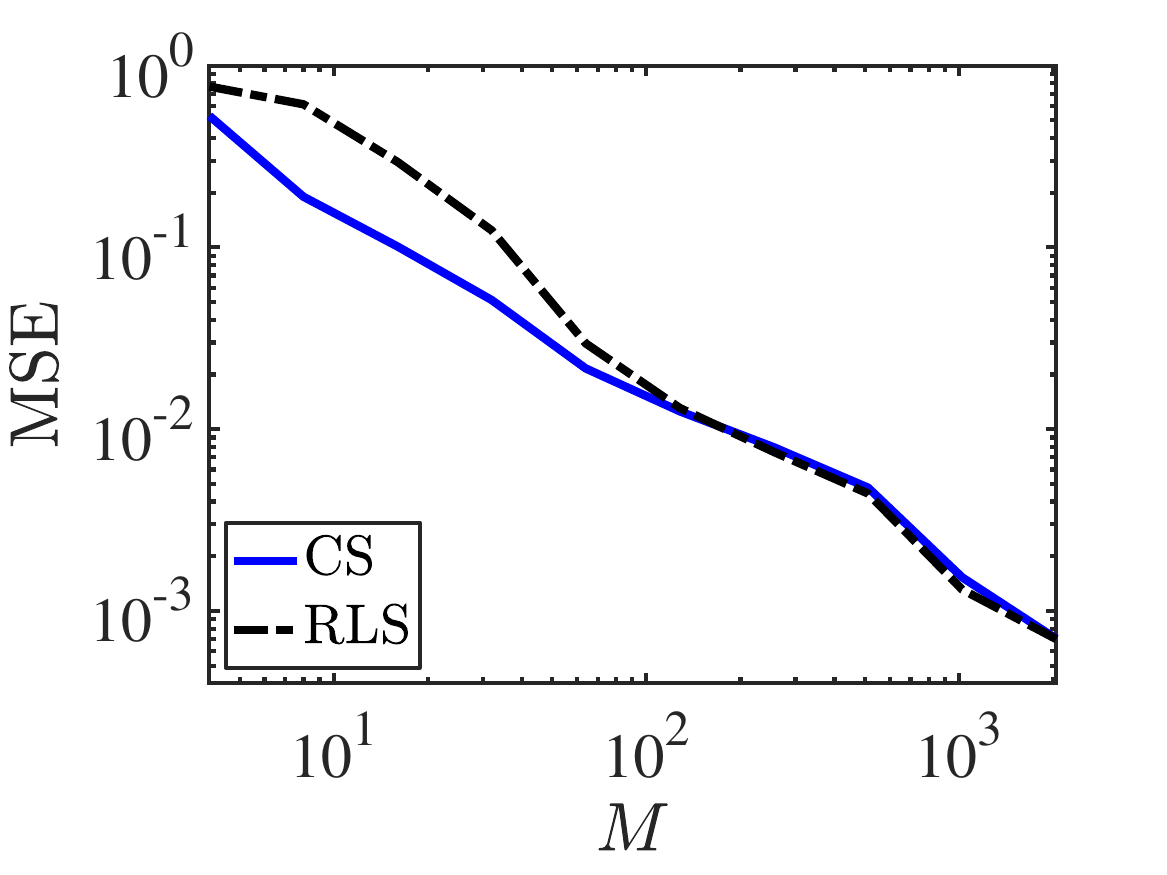}
\caption{$\wh\lambda_0$}
\end{subfigure}
\begin{subfigure}{0.24\textwidth}
\includegraphics[width=\textwidth]{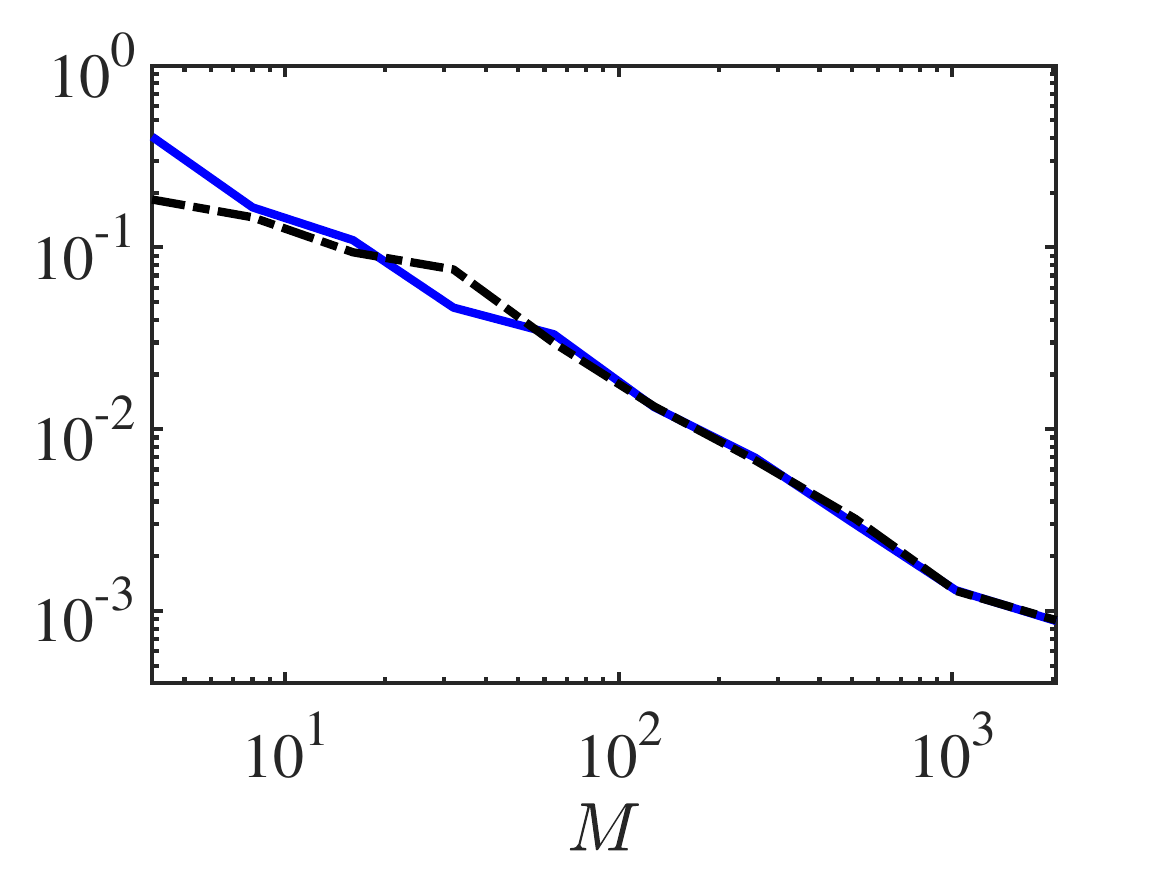}
\caption{$\wh\lambda_1$}
\end{subfigure}
\begin{subfigure}{0.24\textwidth}
\includegraphics[width=\textwidth]{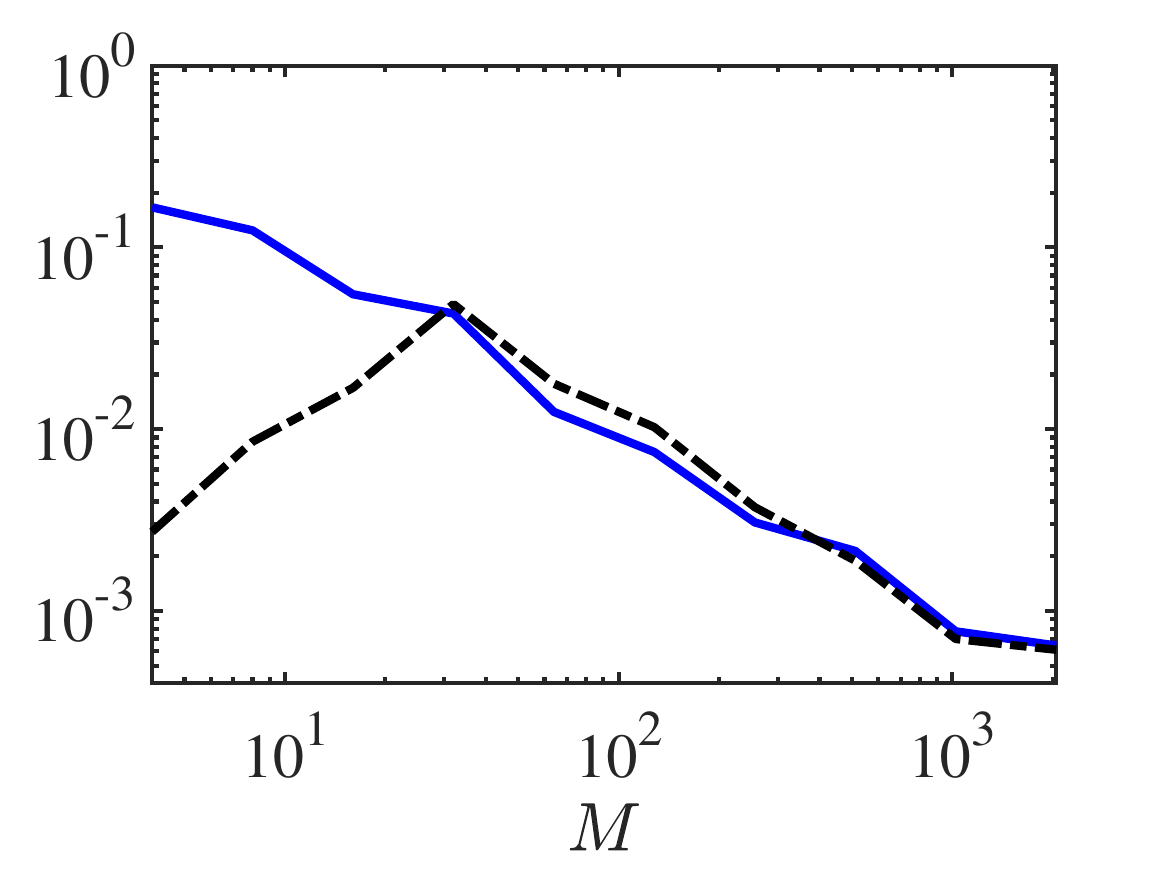}
\caption{$\wh\lambda_2$}
\end{subfigure}
\caption{Illustration of the performance of CS and RLS for estimating the state $\vrho$ and the linear observables $\lambda_i = \trace{\mLambda_i \vrho}$ with  $\mLambda_i = \vphi_i \vphi_i^\dagger$.
}
\label{fig:RLS-Shadow}
\end{figure*}

\begin{figure}[t!]
\centering
\begin{subfigure}{0.235\textwidth}
\includegraphics[width=\textwidth]{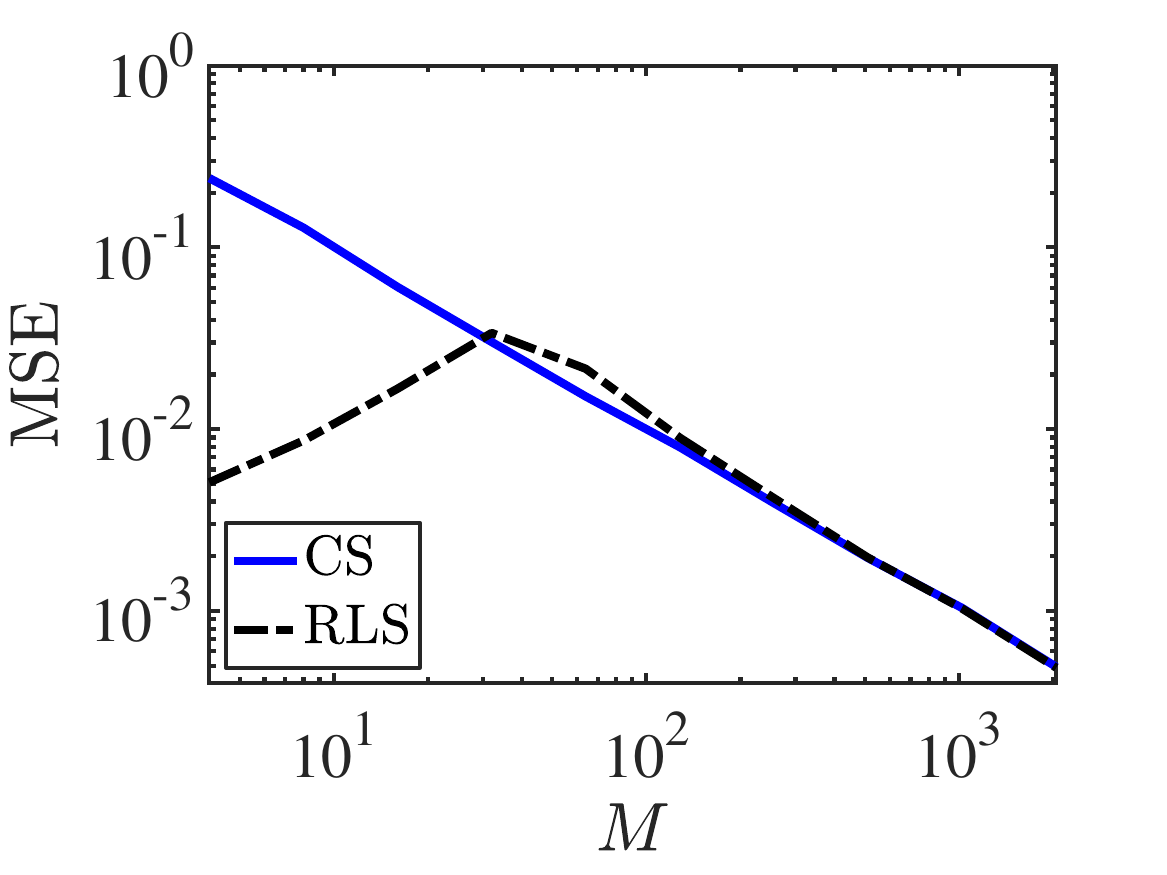}
\caption{}
\end{subfigure}
\begin{subfigure}{0.235\textwidth}
\includegraphics[width=\textwidth]{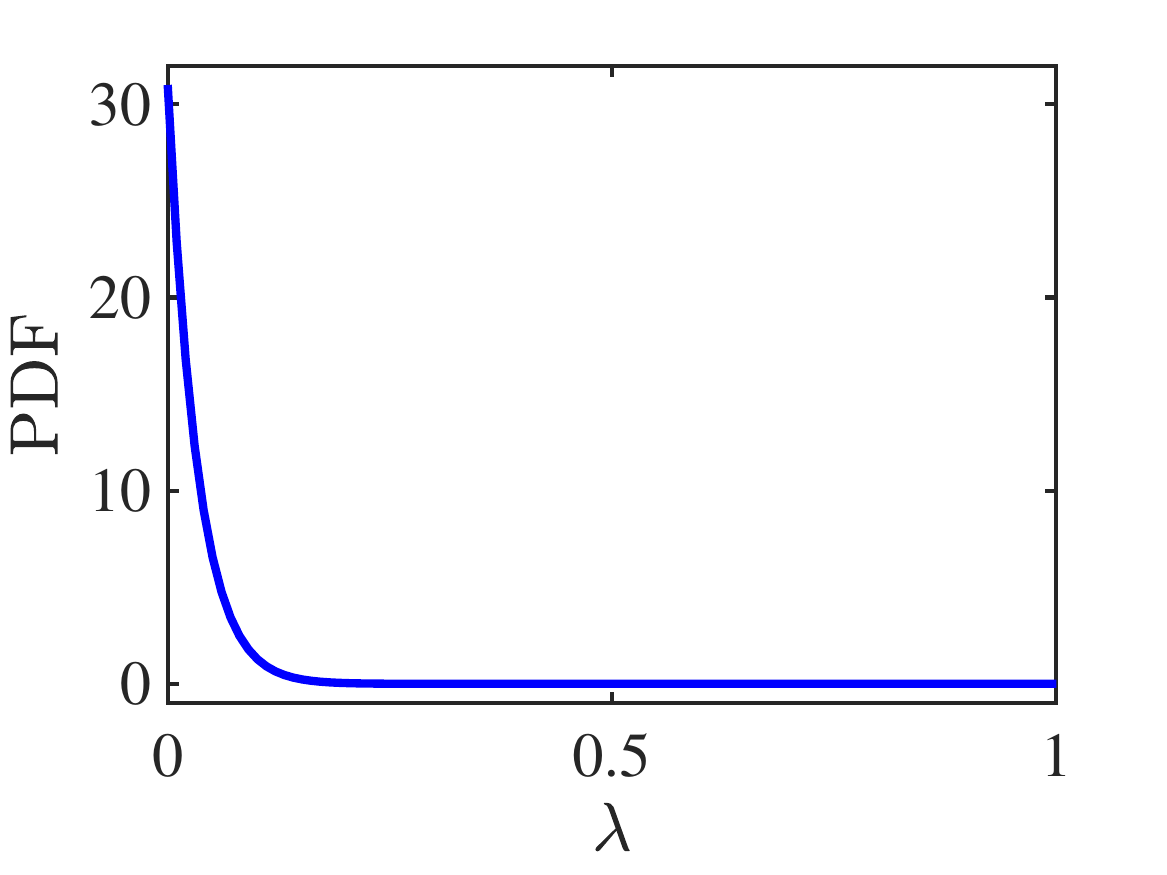}
\caption{}
\end{subfigure}
\caption{(a) Illustration of the performance of the RLS and classical shadow for estimating $50$ linear observables $\lambda = \trace(\mLambda \vrho),\mLambda = \vphi \vphi^\dagger$, where  $\vphi$ is randomly and uniformly generated from the unit sphere. (b) the probability distribution (i.e., probability density function (PDF)) $P(\lambda) = (D-1)(1-\lambda)^{D-2}$ for such a random linear observable $\lambda$.
}
\label{fig:rls-shadow-mse-randon}
\end{figure}

Considering the same ground truth state, observables, and measurements as analyzed by LS shadows in Fig.~\ref{fig:LS}, we apply RLS and CS techniques for inference and show the results in Fig.~\ref{fig:RLS-Shadow}.  The top of the first column plots the estimates of the states in terms of their eigenvalues; the middle shows the distance to the ground truth defined as $\|\wh\vrho - \vrho\|_F$; and the bottom depicts the log-likelihood defined as $\frac{1}{M}\log {\mathcal{L}}(\wh\vrho_\mathrm{phy}) = \frac{1}{M}\sum_{m=1}^M\sum_{k=1}^K  f_{m,k}\log(\trace(\mA_{m,k}\wh\vrho_\mathrm{phy}))$, where
$\wh\vrho_\mathrm{phy}$ denotes the estimator projected onto physical states by the method of Ref.~\cite{Smolin2012}, performed to remove negative probabilities that would otherwise make the log-likelihood complex~\cite{likelihoodNote}.
The last three columns [Fig.~\ref{fig:RLS-Shadow}(b--d)] plot the estimates for particular expectation value $\lambda_i$ for all 50 trials obtained by both RLS and CS; the bottom row shows the MSE with respect to the ground truth, averaged over all trials. 

\Cref{fig:RLS-Shadow}(a) compares features of the estimators themselves. While both RLS and CS yield negative eigenvalues, those from RLS are less extreme than CS in the underdetermined regime ($M<D$) before aligning closely for $M>D$ (top). Such behavior is similarly reflected in RLS's much lower error with respect to the ground truth for $M<D$, followed by close agreement for $M>D$ (middle). Interestingly, however, as the bottom plot shows, likelihood tests (typical in classical inference) strongly favor CS over RLS; for any number of measurements $M$, the likelihood evaluated at the CS estimator exceeds that of RLS, despite the fact RLS is \emph{closer} to the ground truth in much of this regime (smaller error $\|\wh\vrho - \vrho\|_F$).

This unexpected characteristic can be explained by the fact that CS shadows are constructed directly by projectors of the measured outcomes [cf. Eq.~\eqref{eq:shadow-unitary}]. Hence, even for $M\ll D$, the CS estimator automatically assigns high probabilities to prior observations---which is precisely what the likelihood computes. Accordingly, the wide deviation between expectations from a likelihood test (bottom) and the actual ground truth (middle) not only reveals an interesting feature of CS shadows; it also emphasizes the importance of testing these methods against ground truth values to avoid misleading conclusions on their relative merits.

On the observable side, Fig.~\ref{fig:RLS-Shadow}(b--d) reveals that CS shadows possess large variance for low $M$, but are unbiased and converge to ground truth values rapidly, with nearly identical rates for all observables. This enables the derivation of rigorous information-theoretic bounds for any fixed and finite number of linear observables~\cite{huang2020predicting}. In contrast, the $\ell_2$ regularization applied to RLS leads to observable estimators that are biased towards the origin but also have small variance. 
Moreover, compared to the LS shadow tests in Fig.~\ref{fig:LS}, the double-decent phenomenon and wide fluctuations around $M\approx D$ have been significantly attenuated through regularization, especially striking for the $\lambda_0$ and $\lambda_1$ tests. While the RLS estimation errors for $\vrho$ and $\lambda_2$ still increase as $M$ approaches the interpolating regime, they remain small, and they are smaller than those achieved by CS. As discussed at the end of Sec.~\ref{sec:RLS}, this increasing phase could have been mitigated by a large regularization parameter (e.g., $\mu=1$), but at the expense of higher errors for $\lambda_0$ and $\lambda_1$.

\textit{Aside: random observables.}---For the numerical experiments in this paper, we have focused on a fixed state $\vrho = \ve_0 \ve_0^\dagger$ and three rank-1, trace-1 observables with ground truth values $(\lambda_0, \lambda_1, \lambda_3) = (1,1/2,0)$. As discussed in detail in Ref.~\cite{lukens2021bayesian}, observables with ground truth values $\lambda\sim\mathcal{O}(1)$ reflect scenarios where the accuracy of CS significantly surpasses  alternative techniques like maximum likelihood and Bayesian inference for the same number of measurements. Nonetheless, from the perspective of \emph{random} observables or ground truth states, $\lambda\approx 1$ is extremely rare in large Hilbert spaces; in this regime, for example, the Bayesian mean is far more accurate on average than CS~\cite{lukens2021bayesian}.

This distinction helps explain RLS's lower MSE compared to CS for $M\ll D$ in Fig.~\ref{fig:RLS-Shadow}(d).
Taking the same simulated trials 
but selecting $50$ different projectors of the form $\mLambda = \vphi \vphi^\dagger$, where each $\vphi$ is randomly generated from the unit hypersphere according to the Haar measure (equivalent to looking at 50 random ground truth states for a fixed observable~\cite{lukens2021bayesian}),
we find the MSEs plotted in Fig.~\ref{fig:rls-shadow-mse-randon}(a). RLS provides more accurate estimates than CS on average when $M$ is small and demonstrates similar performance as $M$ becomes large. Notably, the MSEs look similar to those for $\lambda_2$ in Fig.~\ref{fig:RLS-Shadow}(d). This can be explained by the probability density function (PDF) for the random observable $\lambda = \trace(\mLambda \vrho) = \|\vphi^\dagger \ve_0\|_2^2$,  given by $P(\lambda) = (D-1)(1-\lambda)^{D-2}$~\cite{zyczkowski2005average}. As depicted in Fig.~\ref{fig:rls-shadow-mse-randon}(b), the random variable $\lambda$ has a mean of $1/D=1/32$ and mode of 0. Thus, in the random context $\lambda_2=0$ represents a much more typical value than either $\lambda_0=1$ or $\lambda_1=1/2$, justifying the strikingly similar behavior for $\wh\lambda_2$ [Fig.~\ref{fig:RLS-Shadow}(d)] compared to Haar-random cases [Fig.~\ref{fig:rls-shadow-mse-randon}(a)]. So while we do consider observables like $\lambda_0$ which reflect situations of interest in practice (e.g., verification of high-fidelity state preparation), it is important to bear in mind the extreme improbability of this situation for truly random states or observables.

\subsection{Feature 2: Distribution Mismatch}

\begin{figure*}[t]
\centering
\begin{subfigure}{0.24\textwidth}
\centering
\includegraphics[width=\textwidth]{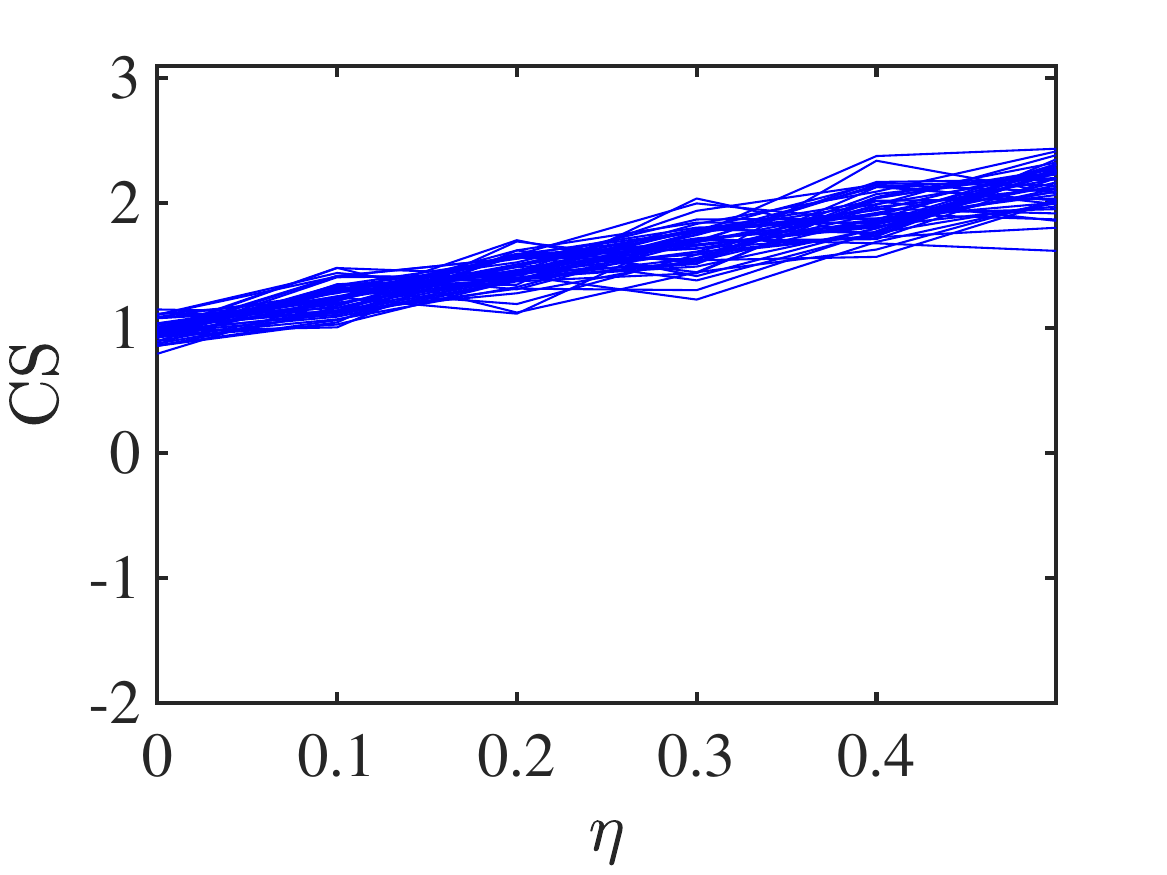}
\end{subfigure}
\begin{subfigure}{0.24\textwidth}
\includegraphics[width=\textwidth]{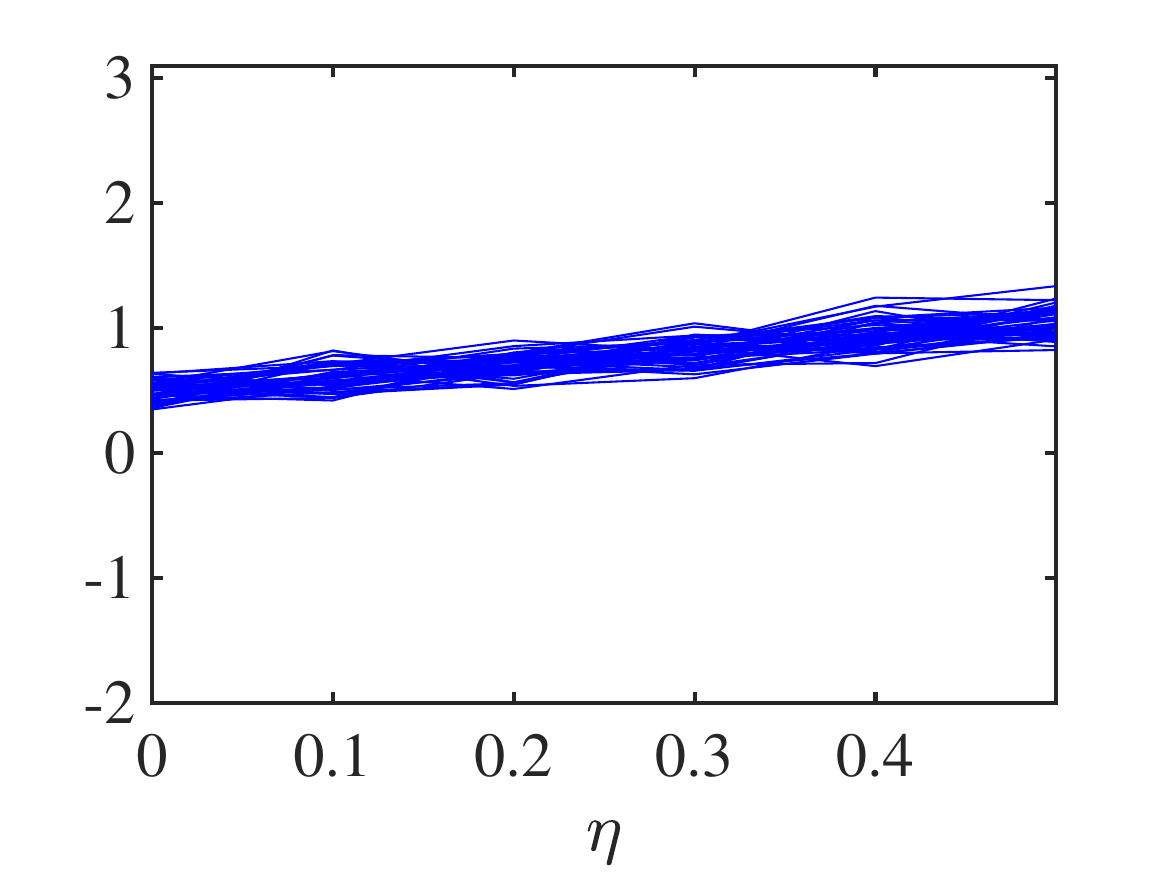}
\end{subfigure}
\begin{subfigure}{0.24\textwidth}
\includegraphics[width=\textwidth]{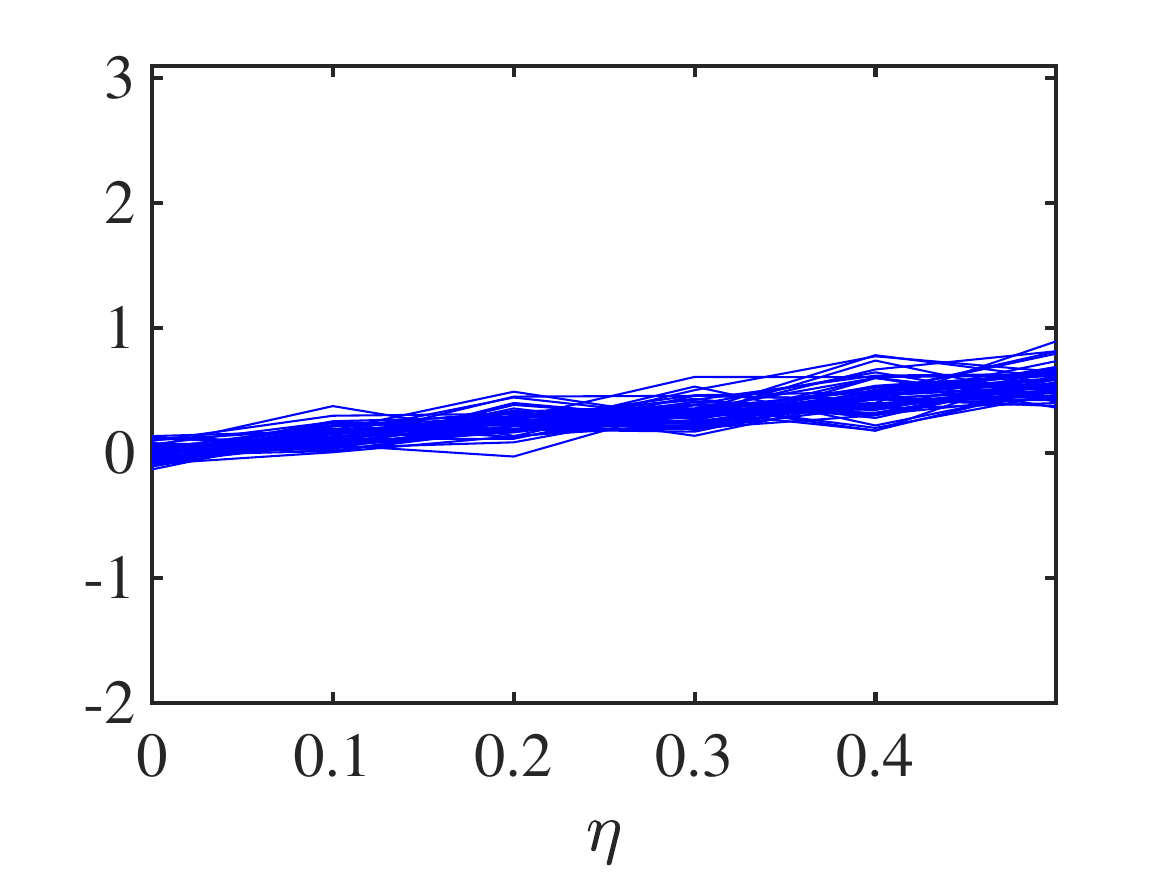}
\end{subfigure}
\vfill
\begin{subfigure}{0.24\textwidth}
\centering
\includegraphics[width=\textwidth]{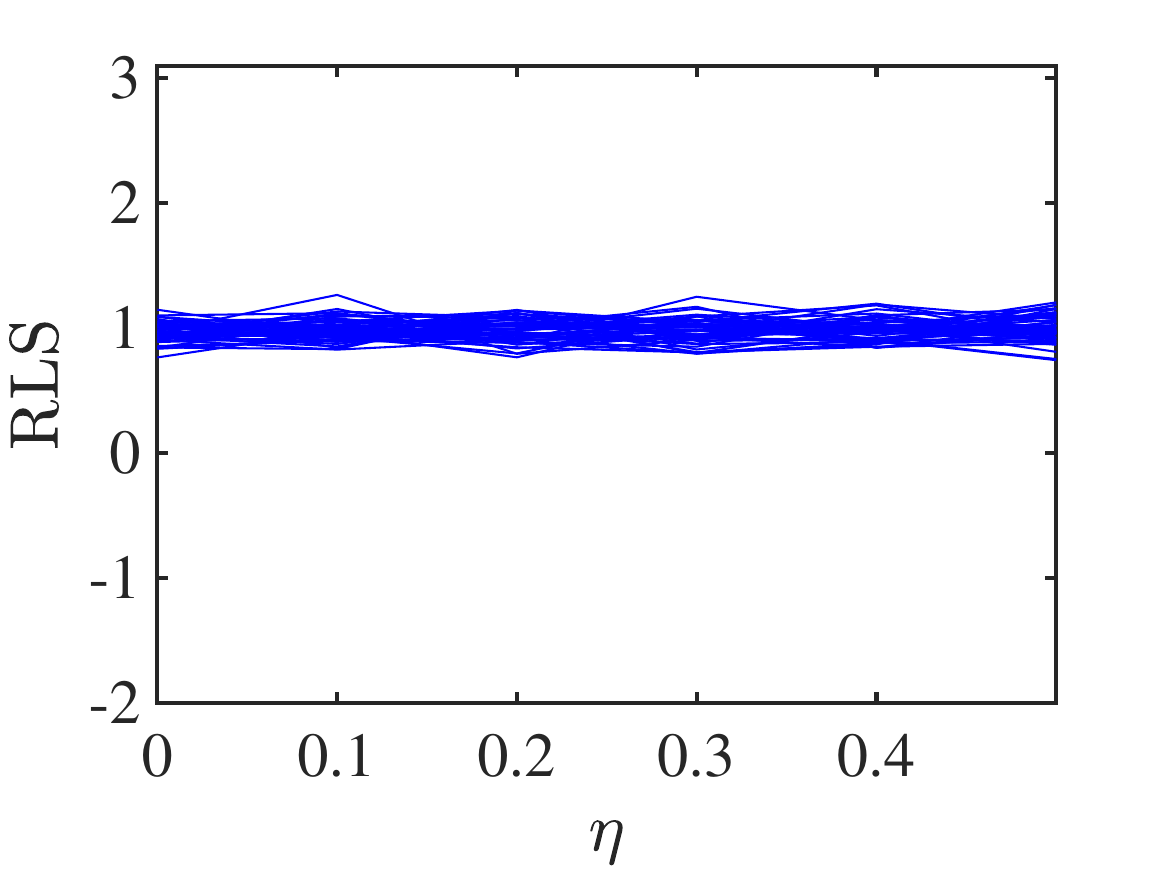}
\end{subfigure}
\begin{subfigure}{0.24\textwidth}
\includegraphics[width=\textwidth]{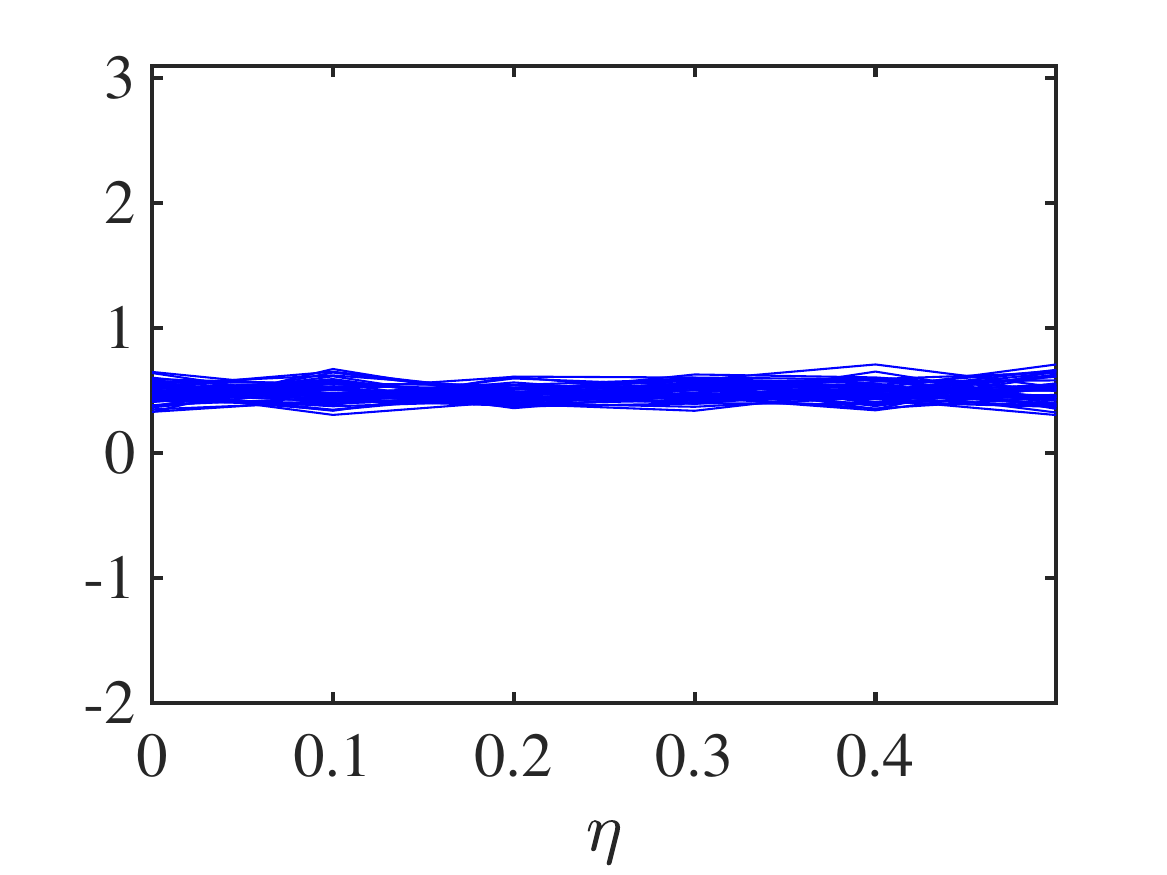}
\end{subfigure}
\begin{subfigure}{0.24\textwidth}
\includegraphics[width=\textwidth]{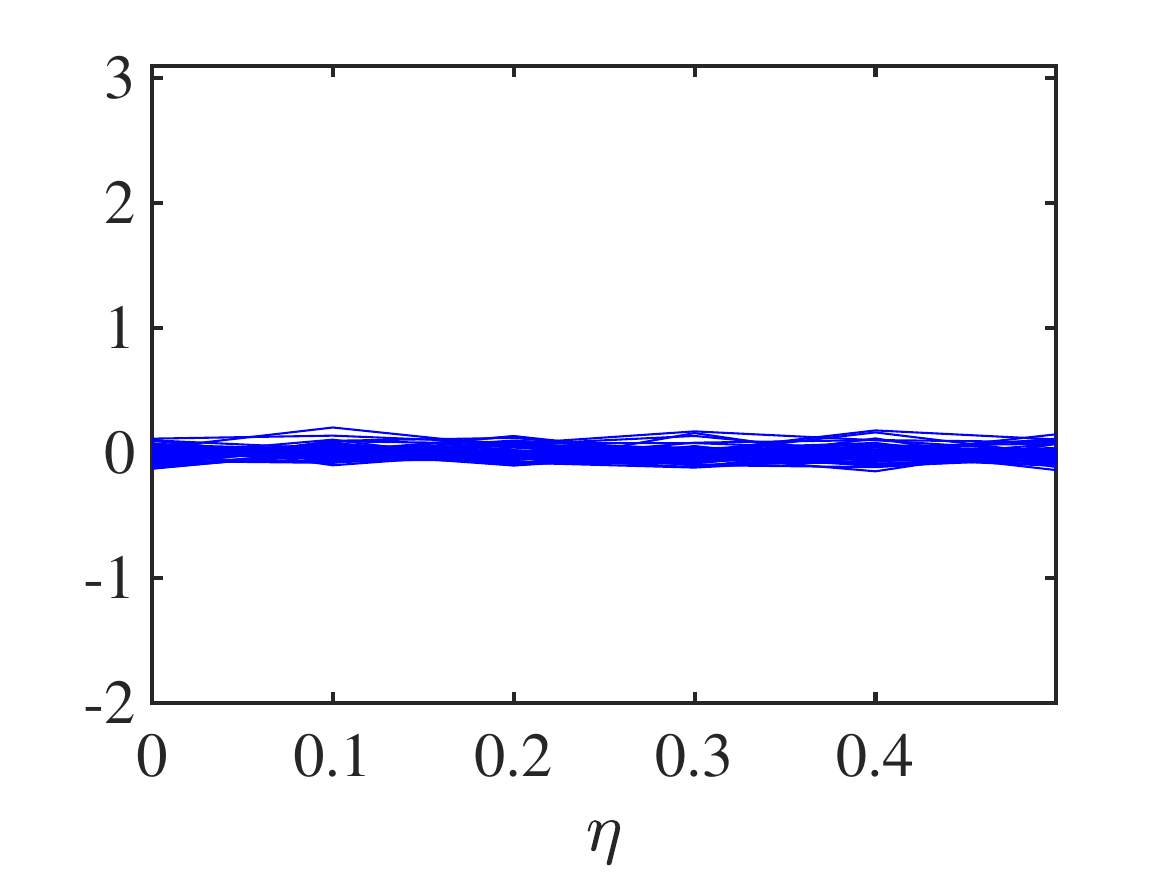}
\end{subfigure}
\vfill
\begin{subfigure}{0.24\textwidth}
\includegraphics[width=\textwidth]{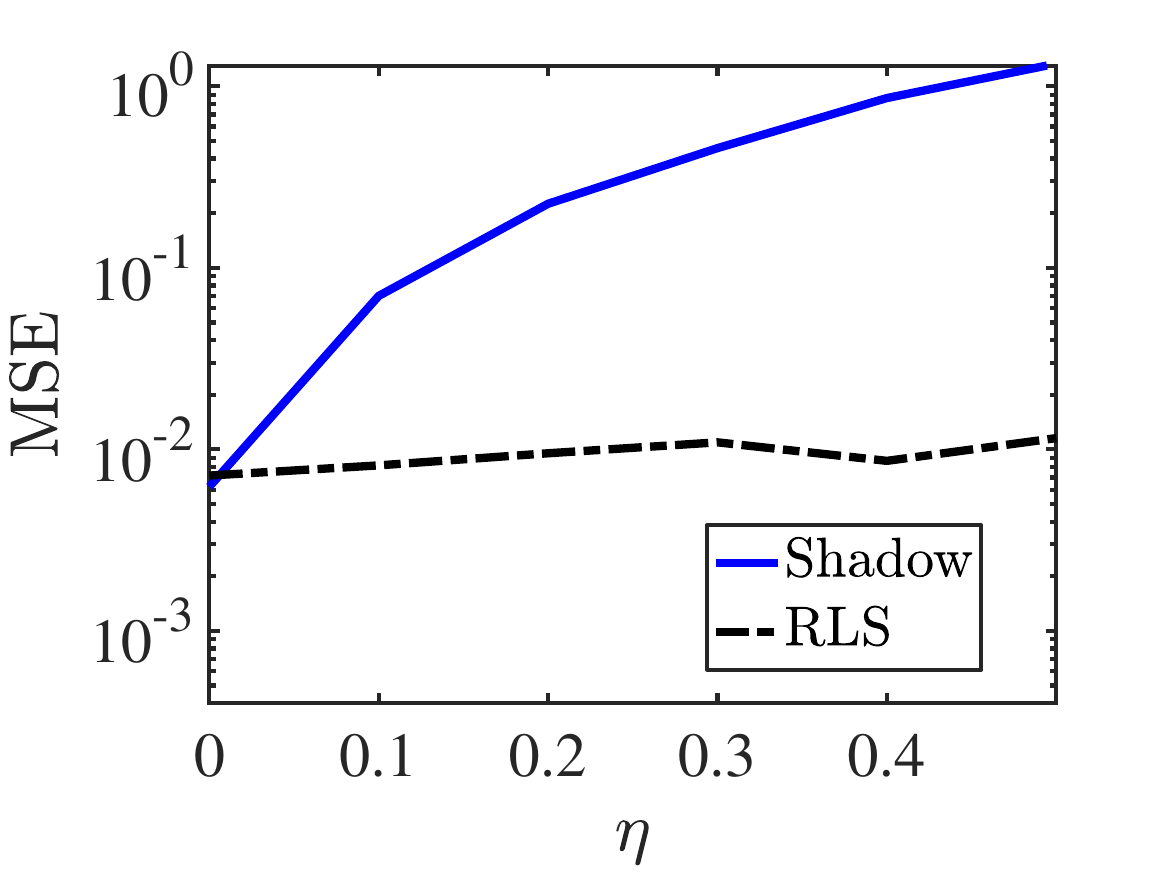}
\caption{$\wh\lambda_0$}
\end{subfigure}
\begin{subfigure}{0.24\textwidth}
\includegraphics[width=\textwidth]{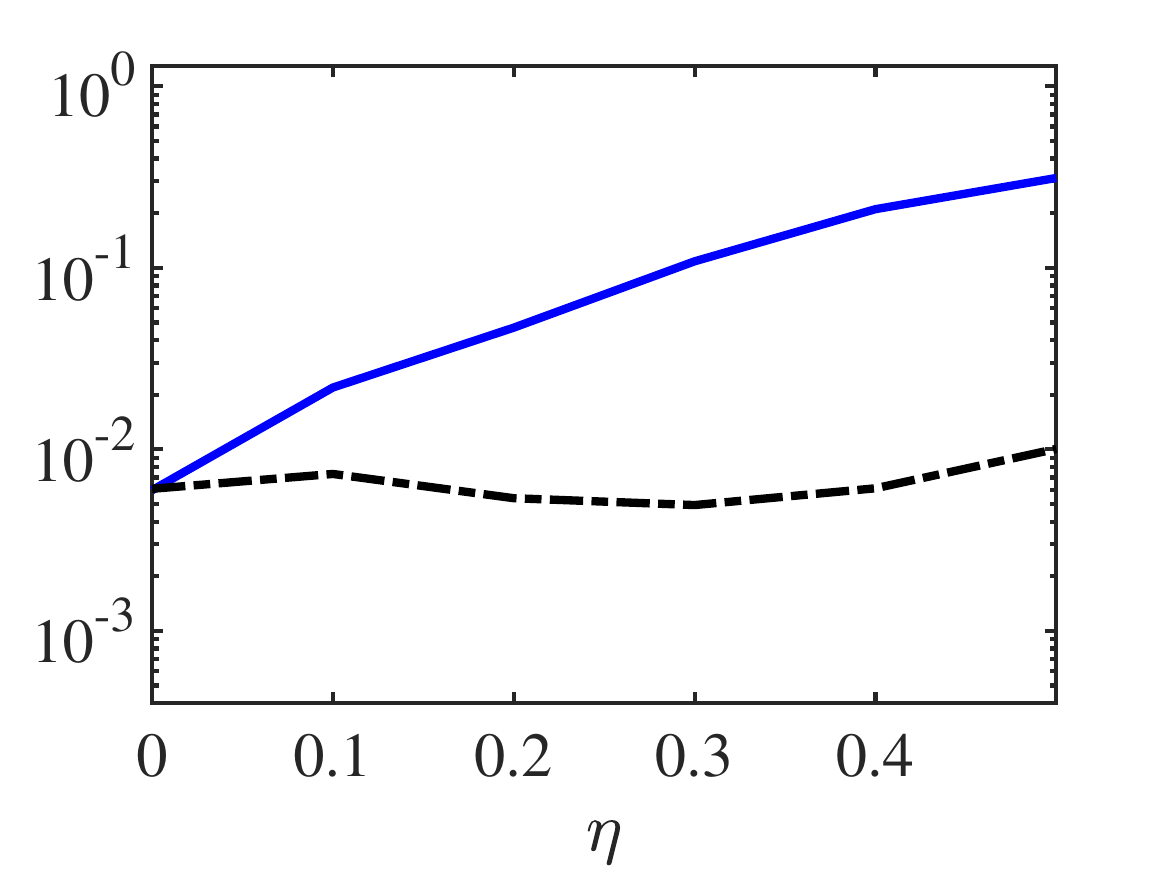}
\caption{$\wh\lambda_1$}
\end{subfigure}
\begin{subfigure}{0.24\textwidth}
\includegraphics[width=\textwidth]{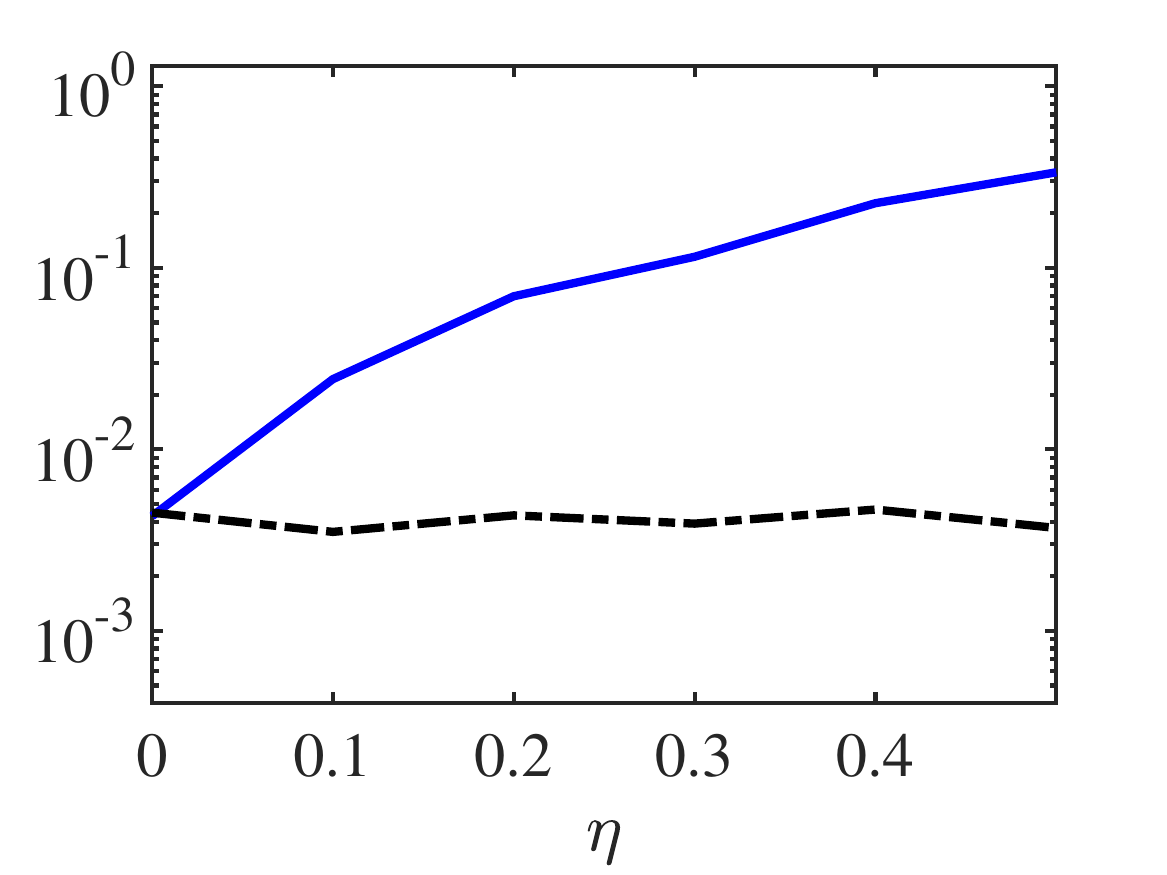}
\caption{$\wh\lambda_2$}
\end{subfigure}
\caption{Illustration of the performance of the RLS and CS approaches for estimating the  linear observables $\lambda_i = \trace(\mLambda_i \vrho)$ in the presence of distribution shifts: the unitary matrices are generated from the distribution $(1-\eta)P(\setU_{D}) + \eta P((\setU_{2})^{\otimes n})$, where $\setU_D$ and $(\setU_{2})^{\otimes n}$ denote the global and local Haar-distributed unitary matrices, respectively.  
}
\label{fig:RLS-Shadow-model-shift}
\end{figure*}

As described in \Cref{sec:CS}, CS leverages the assumption of random measurements and involves the computation of the quantum channel for the entire distribution, as in Eq.~\eqref{eq:Exp-AtA}. However, in practical experiments, the randomness is typically synthetic, and thus the measurement basis may not be perfectly generated according to the desired distribution. In this case, the distribution shift may pose robustness challenges for CS as it crucially depends on the formulation of the quantum channel. In contrast, the RLS approach does not incorporate assumptions of randomness at any point and may therefore prove more useful in this case.

To further illustrate this point, we invoke the same setup as in the previous experiments, but now generating the unitary measurement matrices from a mixture distribution $(1-\eta)P(\setU_{D}) + \eta P((\setU_{2})^{\otimes n})$, where $\setU_D$ and $(\setU_{2})^{\otimes n}$ denote global and local (tensor products of qubit) Haar-distributed unitary matrices, respectively.  In other words, for each state copy we perform either a random global basis measurement with probability $1-\eta$ or a random local basis measurement with probability $\eta$. While this toy example is not expected to reflect distribution errors in practical systems, it provides a clear showcase of the key points. Suppose that an experimenter expects the chosen operations to be sampled from a global Haar distribution, but with probability $\eta$ actually generates tensor products of local qubit unitaries. In this case, the CS shadows will be computed according to Eq.~\eqref{eq:shadow-unitary} under assumptions violated by the experiment.

\Cref{fig:RLS-Shadow-model-shift} shows numerical results for the system of interest for $M = 256$ measurements and $\eta\in(0,0.5)$. The performance of CS degrades as a result of the distribution shift, and it worsens as $\eta$ increases. In contrast, RLS depends solely on the actual measurements performed and does not rely on information about the sampling distribution, making its performance stable against such distribution shifts.

\subsection{Feature 3: Multishot Measurements}
\label{sec:multishot}

\begin{figure*}[t]
\centering
\begin{subfigure}{0.24\textwidth}
\centering
\includegraphics[width=\textwidth]{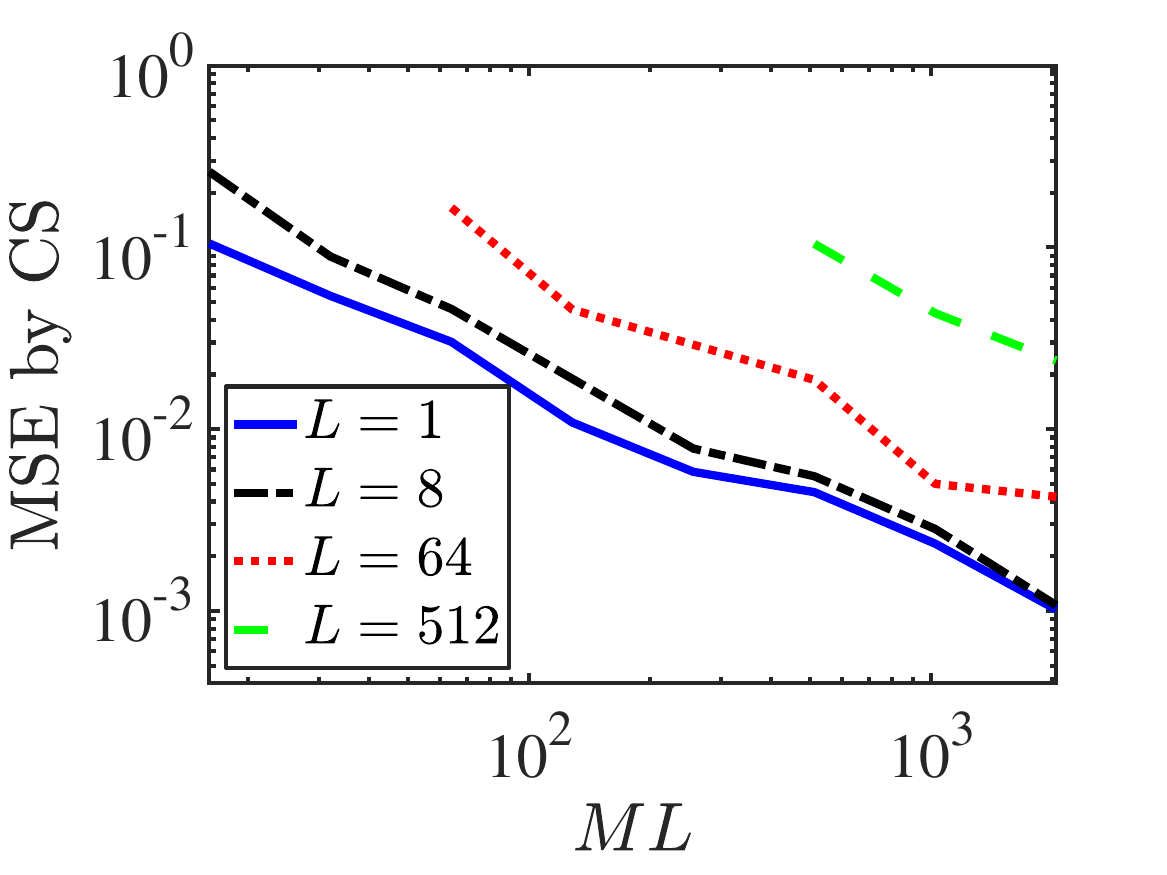}
\end{subfigure}
\begin{subfigure}{0.24\textwidth}
\includegraphics[width=\textwidth]{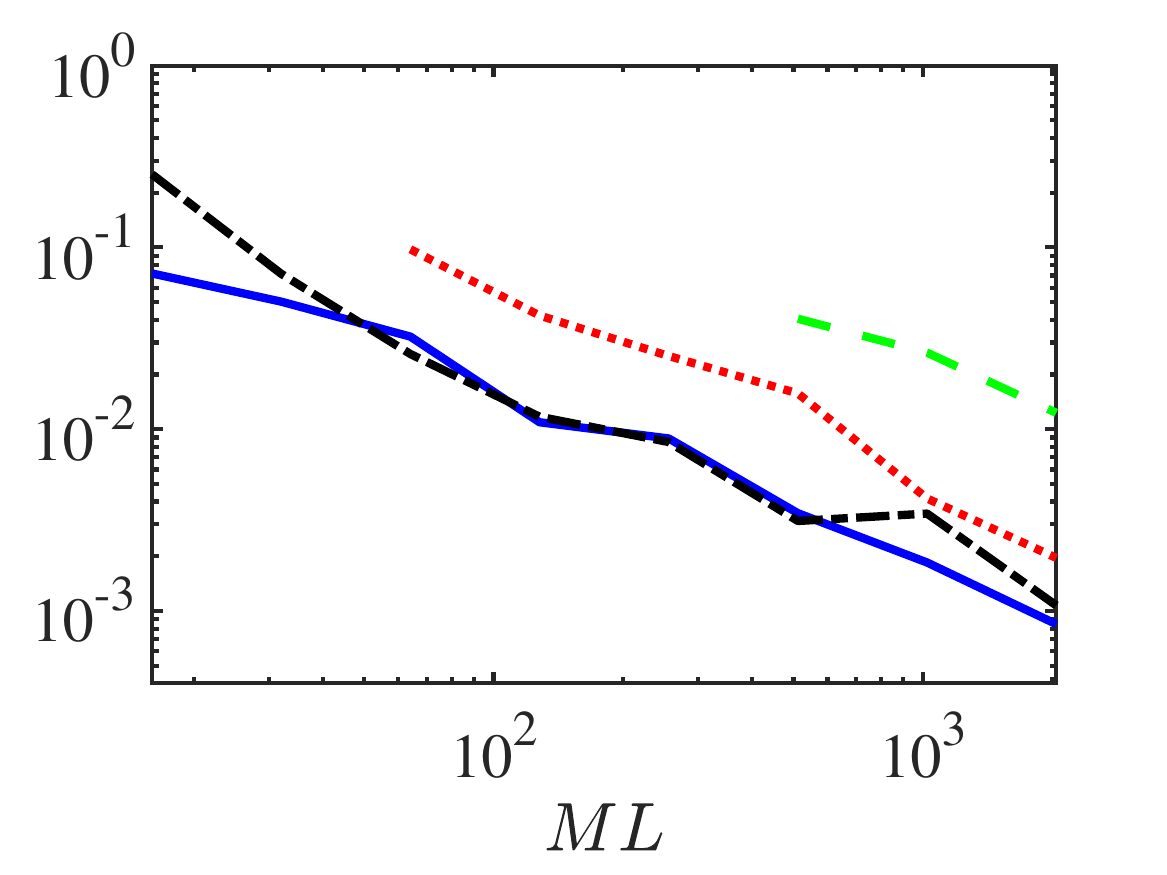}
\end{subfigure}
\begin{subfigure}{0.24\textwidth}
\includegraphics[width=\textwidth]{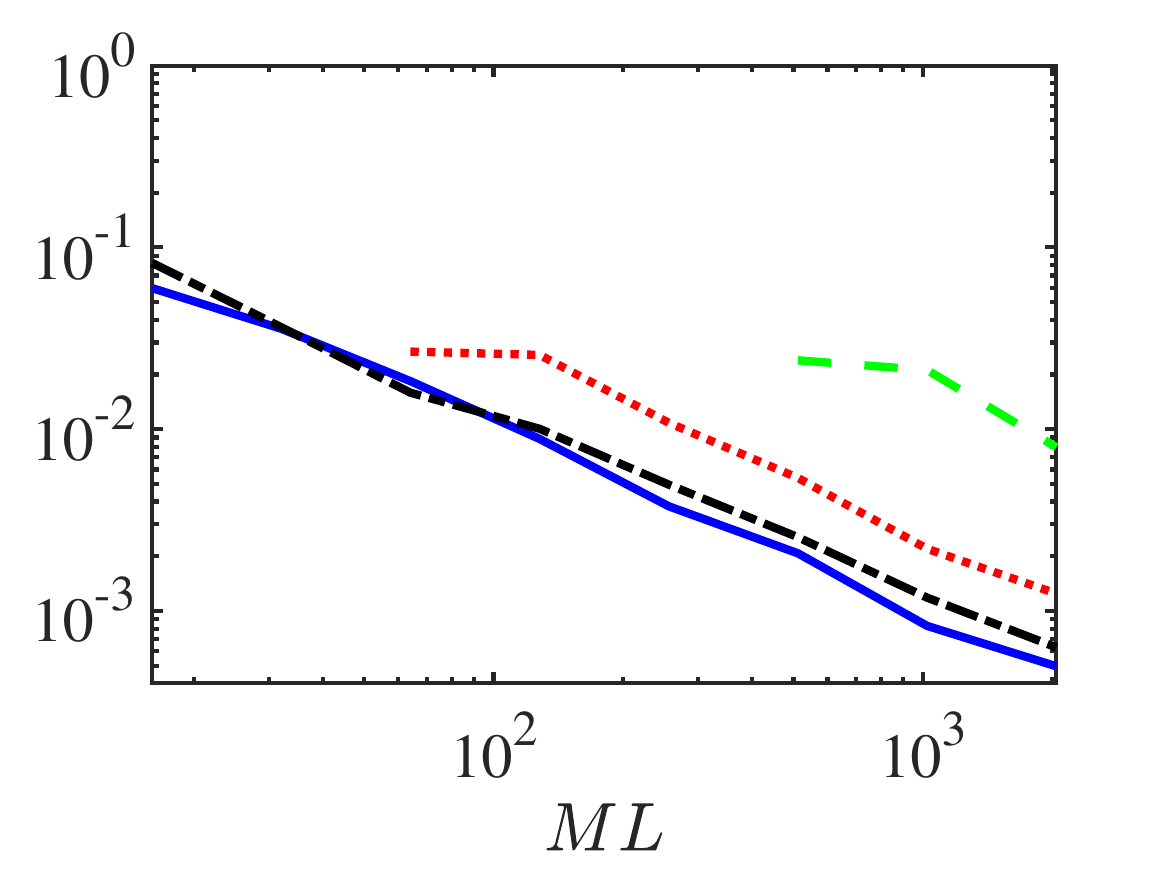}
\end{subfigure}
\begin{subfigure}{0.24\textwidth}
\includegraphics[width=\textwidth]{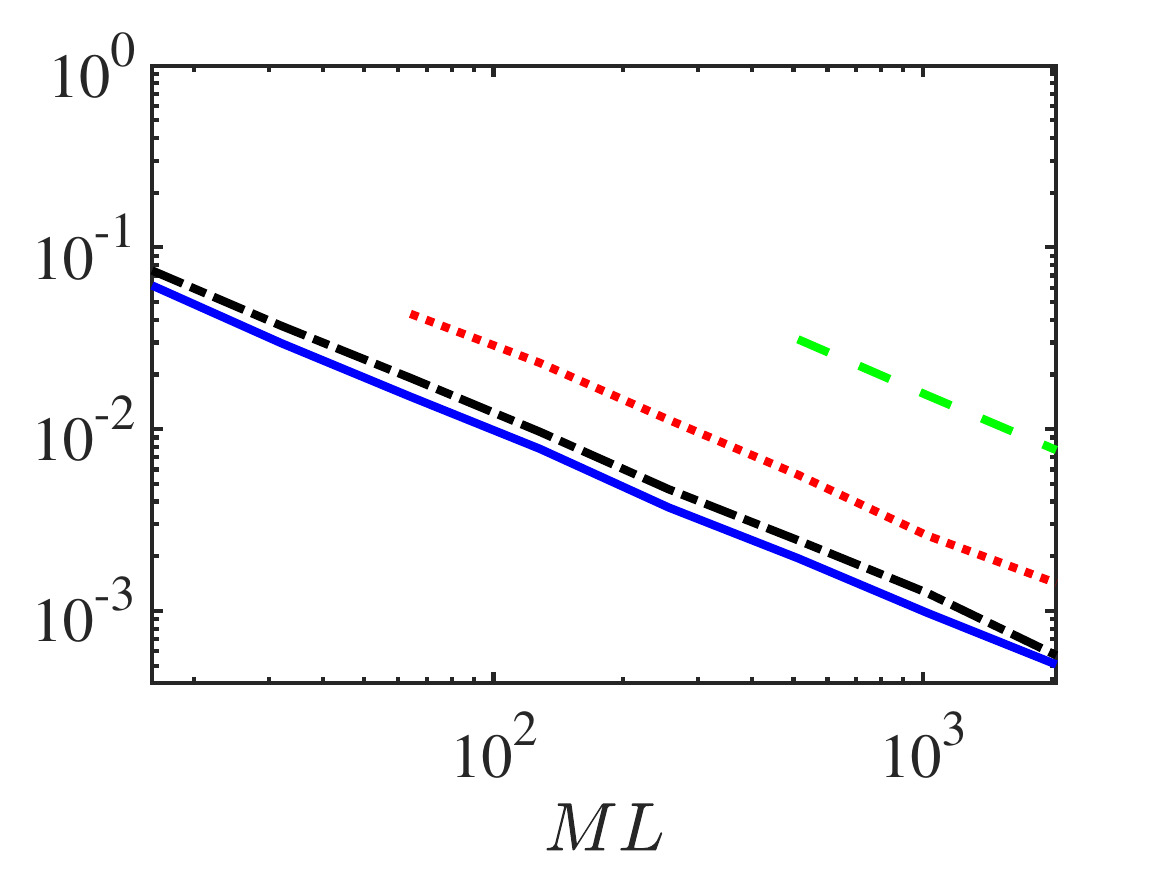}
\end{subfigure}
\vfill
\begin{subfigure}{0.24\textwidth}
\centering
\includegraphics[width=\textwidth]{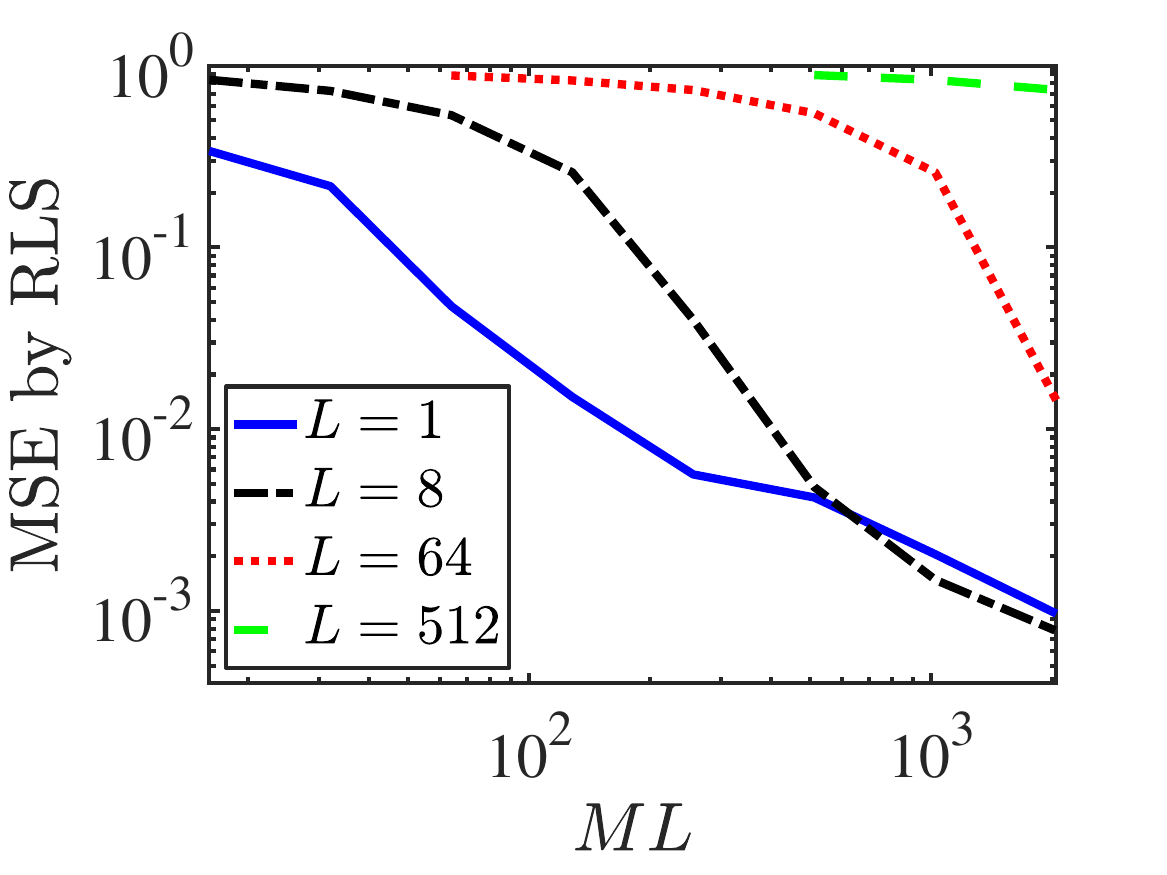}
\caption{$\wh\lambda_0$}
\end{subfigure}
\begin{subfigure}{0.24\textwidth}
\includegraphics[width=\textwidth]{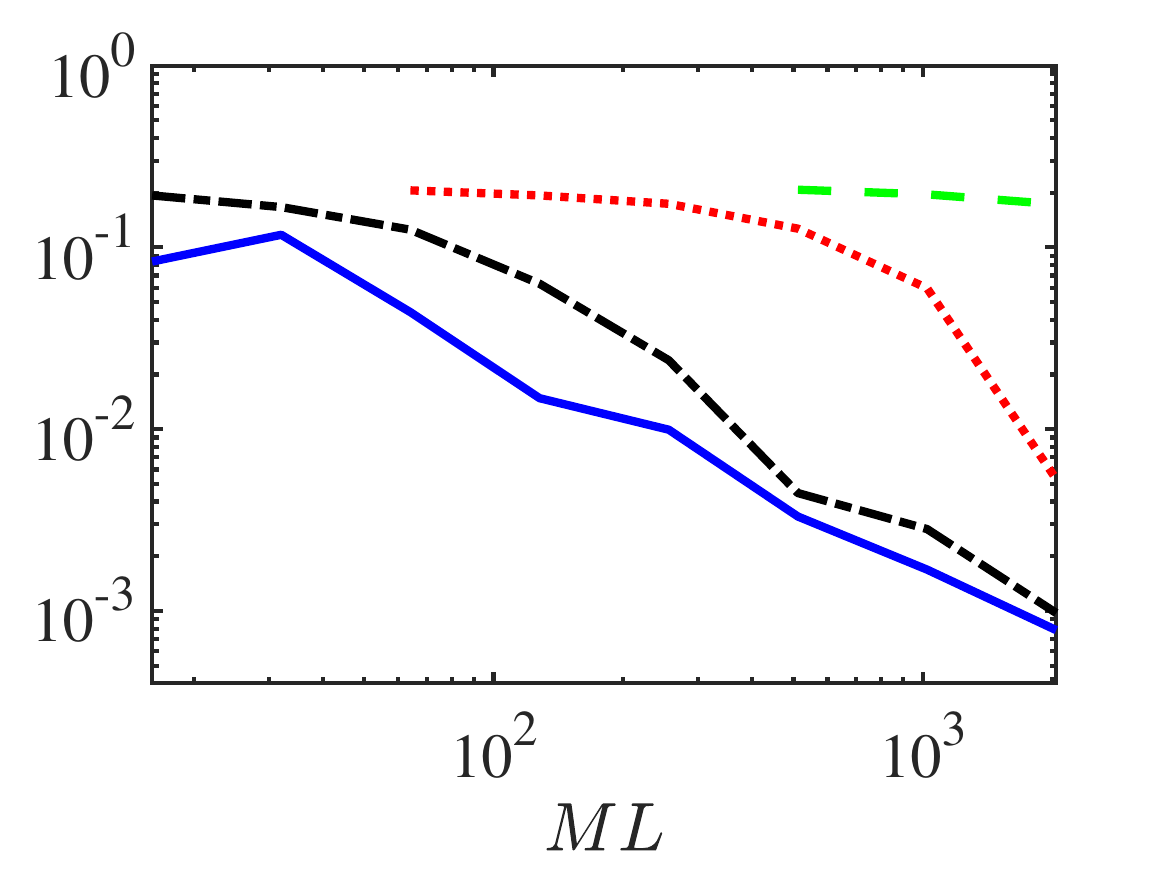}
\caption{$\wh\lambda_1$}
\end{subfigure}
\begin{subfigure}{0.24\textwidth}
\includegraphics[width=\textwidth]{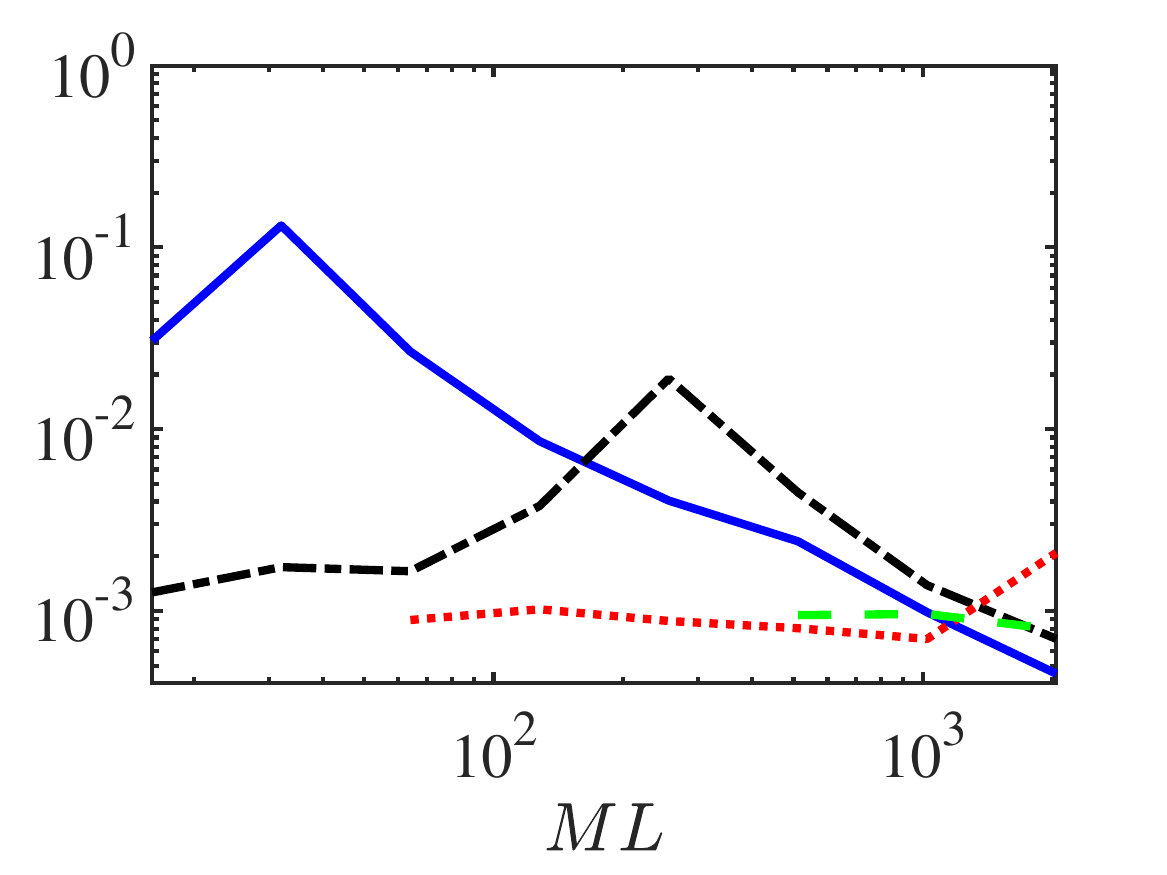}
\caption{$\wh\lambda_2$}
\end{subfigure}
\begin{subfigure}{0.24\textwidth}
\includegraphics[width=\textwidth]{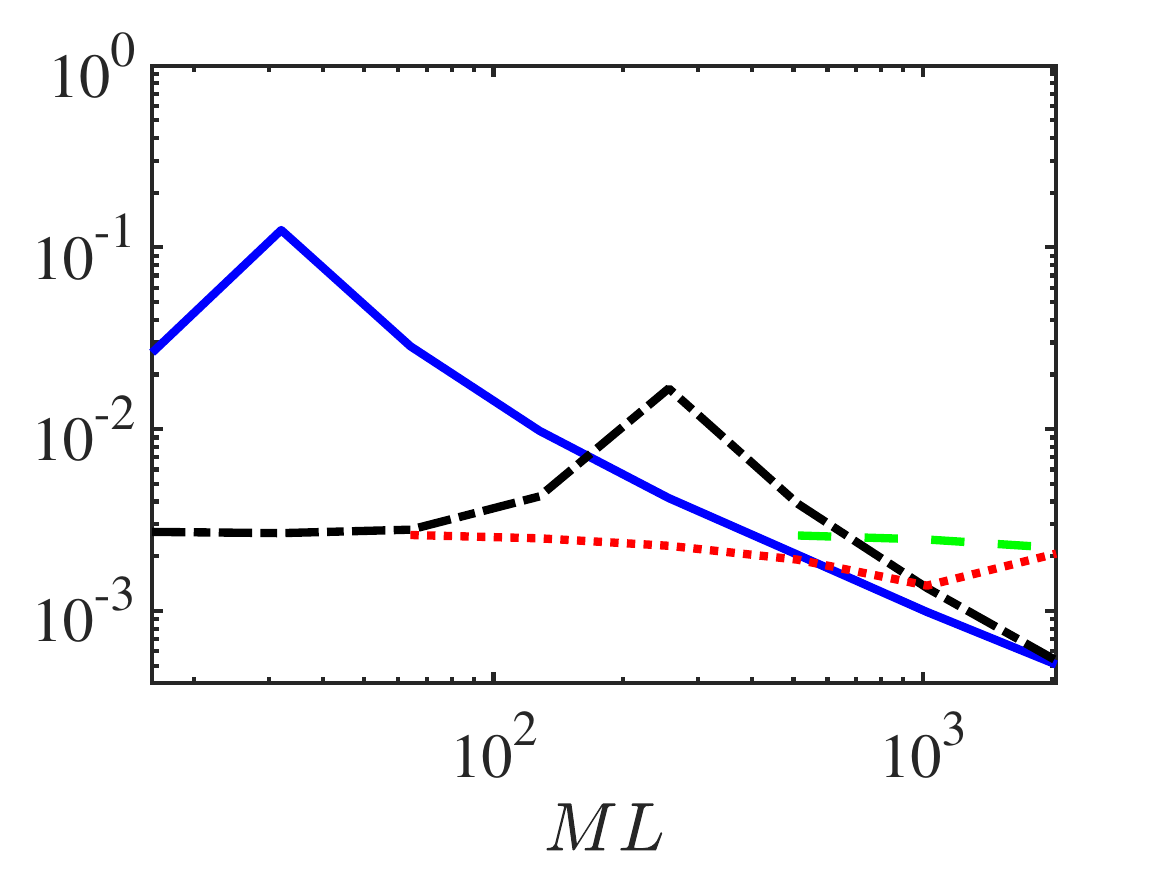}
\caption{$\wh\lambda$}
\end{subfigure}
\caption{Illustration the performance of CS (top row) and RLS (bottom row) with multishot measurements for estimating the three linear observables $\{\lambda_0, \lambda_1, \lambda_2\}$ as in Fig.~\ref{fig:LS}, and 50 random linear observables $\lambda$ as in Fig.~\ref{fig:rls-shadow-mse-randon}(a).
}
\label{fig:one-vs-more}
\end{figure*}

The derivations culminating in Eqs. (\ref{eq:shadow-v2}, \ref{eq:shadow-v3}) for RLS and CS are not confined to single-shot measurements ($L=1$) but inherently include multishot ($L>1$) scenarios as well. In the context of multishot measurements, where empirical frequencies are computed by averaging across all outcomes [Eq.~(\ref{eq:empirical-prob})], the shadows in Eqs.~(\ref{eq:shadow-v2}, \ref{eq:shadow-v3}) could always be \emph{viewed} as single-shot results but duplicated POVMs, by converting $L$ and $M$ to effective values $L_\mathrm{eff}=1$ and $M_\mathrm{eff}=ML$.
This equivalence stems from the linear nature of shadows concerning the empirical frequencies.

What then distinguishes measuring the quantum state using each POVM only once or multiple times? The primary differences are practical in origin, depending on the relative difficulty of preparing state copies compared to reconfiguring the measurement. For example, in many photonics experiments (particularly with spontaneous parametric downconversion~\cite{Mandel1995, Shih2003}), states are prepared continuously and at random, so one need only increase the integration time to push to large $L$ for a fixed measurement setting. On the other hand, for systems composed of superconducting circuits where each state copy is actively prepared, the difference in difficulty between increasing $L$ and increasing $M$ is less dramatic, since both state and measurement circuit are prepared actively and deterministically. 

For fixed number of state copies $ML$, we expect the single-shot regime $L=1$ to provide the closest agreement between LS and CS, for in that case (maximum $M$) the experimental operator $\calA^\dagger \calA$ should approach its expected value $\E[\calA^\dagger \calA]$ most rapidly. In general, for a fixed $ML$, increasing $M$ explores the Hilbert space more efficiently at the expense of greater statistical noise per setting, whereas increasing $L$ reduces statistical noise at the expense of measurement variety. 

To explore this tradeoff for CS and RLS, we conduct numerical experiments using the same setup as in Fig.~\ref{fig:RLS-Shadow} in the multishot regime. In each experiment, we keep the total number of copies of the state $ML$ fixed and vary the number of shots $L$ within the set $\{1, 8, 64, 512\}$ for measuring each POVM, which in turn varies the number of POVMs $M$. Fig.~\ref{fig:one-vs-more} illustrates the performance of CS (top row) and RLS (bottom row) in estimating the expectations of the three linear observables $\{\lambda_0, \lambda_1, \lambda_2\}$ as in Fig.~\ref{fig:LS}, and 50 random linear observables as in Fig.~\ref{fig:rls-shadow-mse-randon}(a). CS consistently achieves its best performance with single-shot measurements, and its error increases with $L$. This observation aligns with our earlier discussion: a greater number of random POVMs brings $\frac{1}{M}\calA^\dagger \calA$ closer to its expected value as described in Eq.~(\ref{eq:Exp-AtA}), and the maximum diversity of POVMs is attained in a single-shot measurement. 

Interestingly, while the primary impact of $L>1$ for CS estimation is to shift the total error up for a given $ML$, the log-log slopes of all examples remain approximately $-1$, indicating favorable scaling $\mathrm{MSE}\propto (ML)^{-1}$ for all $L$ regimes considered. On the other hand, $L>1$ examples alter the \emph{slopes} of the RLS estimator errors as well as their absolute values. For the $\wh\lambda_0$ and $\wh\lambda_1$ cases, the error increases rapidly with $L$ accompanied by an extremely  shallow initial slope [Fig.~\ref{fig:one-vs-more}(a,b)]; in contrast, the low-$ML$ regime of MSE for $\wh\lambda_2$ and random linear observables $\wh \lambda$ is actually \emph{lower} for $L>1$ compared to $L=1$ [Fig.~\ref{fig:one-vs-more}(c)]. Both of these features are likely due to the bias present in RLS shadows. With a smaller number of POVMs (low $M$), the observables computed by RLS are heavily biased to zero, 
which happens to deviate strongly from the ground truth values $\lambda_0=1$ and $\lambda_1=1/2$, yet is precisely the true expectation of $\mLambda_2$  ($\lambda_2 = 0$) and close to the expectation of most random $\mLambda$ [cf. Fig.~\ref{fig:rls-shadow-mse-randon}(b)]. In contrast, because the CS shadow is always unbiased regardless of $M$ and $L$, comparable behavior is seen for all examined observables.

To complement these numerical and qualitative findings, we now mathematically derive MSE formulas for the expectation of observable $\mLambda$ [$\lambda = \trace(\mLambda \vrho)$] with estimator $\wh\lambda = \trace(\mLambda \wh \vrho) = \frac{1}{M}\sum_{m=1}^M \trace(\mLambda \wh \vrho_m)$. We will focus on CS shadows as they are unbiased estimators, which will simplify the analysis, and exhibit consistent behavior across different observables in Fig.~\ref{fig:one-vs-more}. The following result establishes the variance (and hence MSE as CS is unbiased) of CS for estimating $\lambda$.

\begin{widetext}
\begin{theorem} Consider a ground truth state $\vrho$ which is repeatedly prepared and measured with $M$ POVMs $\{\mA_{m,k}\}_{k\in[K]},m\in[M]$  generated independently and randomly from an ensemble $\setA$ according to probability distribution $P(\setA)$. Each POVM is used to measure the state $L$ times. Then the MSE of the CS estimate of the expectation of the linear observable $\mLambda$ is given by 

\e\begin{split}
&\E_{\{\mA_{m,k}\}\sim P(\setA),\wh \vp_m}\left[\parans{\trace(\mLambda \wh \vrho) - \trace(\mLambda \vrho) }^2\right]\\
& = \frac{1}{ML}\E_{\{\mA_{k}\}\sim P(\setA)}\left[\sum_{k=1}^K \parans{{\trace(\mA_k \vrho)} + (L-1)\parans{\trace(\mA_k \vrho)}^2} \cdot \parans{\trace\parans{\mLambda \cdot{\calM^{-1}\parans{  \mA_{k} }}} }^2\right]  \\
& \quad + \frac{1 -1/L}{M}\E_{\{\mA_{k}\}\sim P(\setA)}\left[\sum_{k\neq k'}  \trace(\mA_k \vrho) \cdot \trace(\mA_{k'} \vrho)  \cdot \trace\parans{\mLambda \cdot{\calM^{-1}\parans{  \mA_{k} }}} \cdot \trace\parans{\mLambda \cdot{\calM^{-1}\parans{  \mA_{k'} }}} \right]\\
& \quad - \frac{1}{M}\parans{\trace(\mLambda \vrho)}^2.
\end{split}
\label{eq:MSE-multi-shots}\ee
\label{lemma:MSE-multi-shots}\end{theorem}

\begin{proof}[Proof of \Cref{lemma:MSE-multi-shots}]
Using the expression $\wh\vrho = \frac{1}{M}\sum_{m}\wh\vrho_m$, we have
\e\begin{split}
& \E_{\{\mA_{m,k}\}\sim P(\setA),\wh \vp_m}\left[\parans{\frac{1}{M}\sum_{m=1}^M \trace(\mLambda \wh \vrho_m) - \trace(\mLambda \vrho) }^2\right] \\
& = \frac{1}{M^2}\sum_m\E_{\{\mA_{m,k}\}\sim P(\setA),\wh \vp_m}\left[\parans{\trace(\mLambda \wh \vrho_m) }^2\right] - \frac{1}{M}\parans{\trace(\mLambda \vrho)}^2\\
& = \frac{1}{M^2}\sum_m
\E_{\{\mA_{m,k}\}\sim P(\setA),\wh \vp_m}\left[\parans{\trace\parans{\mLambda \cdot {\calM^{-1}\parans{ \sum_{k=1}^K \wh p_{m,k} \mA_{m,k} }}} }^2\right] - \frac{1}{M}\parans{\trace(\mLambda \vrho)}^2,
\end{split}
\label{eq:MSE-shadow-obsrvable}\ee
where the first equality follows from the independence and unbiasedness of the shadows $\wh\vrho_m, m\in[M]$. We now focus on the analysis of $\E_{\{\mA_{m,k}\}\sim P(\setA),\wh \vp_m}\left[\parans{\trace\parans{\mLambda \cdot{\calM^{-1}\parans{ \sum_{k=1}^K \wh p_{m,k} \mA_{m,k} }}} }^2\right]$. Since this term will be the same for each $m$, for simplicity, we drop the subscript $m$ and write it as $
\E\left[\parans{\trace\parans{\mLambda \cdot{\calM^{-1}\parans{ \sum_{k=1}^K \wh p_{k} \mA_{k} }}} }^2\right],$
where $\{\wh p_k\}$ are the empirical frequencies that are obtained by using the randomly generated POVM $\{\mA_k\}$ to measure the quantum state $L$ times.  
We note that here the expectation is taken over two types of randomness: the randomly selected POVM $\{\mA_k\}$ and the random measurements $\{\wh p_k\}$. Conditioned on $\{\mA_k\}$, $\{\wh p_k\}$ obeys a multinominal distribution with properties
\begin{align}
&\E\left[\wh p_k^2  \mid \{\mA_k\} \right] = p_k^2 + \frac{p_k (1-p_k)}{L} =  \frac{p_k + (L-1)p_k^2}{L},\\
&\E\left[\wh p_k \wh p_{k'}  \mid \{\mA_k\} \right] = p_k p_{k'} - \frac{1}{L} p_k p_{k'}  = \left(1 -\frac{1}{L}\right) p_k p_{k'}, \ \forall \ k\neq k'.
\end{align}
We now proceed by using these results and the fact that $\calM^{-1}$ is a linear operator:
\e\begin{split}
&\E_{\{\mA_k\},\wh \vp}\left[\parans{\trace\parans{\mLambda \cdot{\calM^{-1}\parans{ \sum_{k=1}^K \wh p_{k} \mA_{k} }}} }^2\right] = \E_{\{\mA_k\},\wh \vp}\left[\parans{\sum_{k=1}^K \wh p_{k}\cdot \trace\parans{\mLambda \cdot{\calM^{-1}\parans{  \mA_{k} }}} }^2\right]\\
& = \E_{\{\mA_k\},\wh \vp}\left[\sum_{k=1}^K \wh p_{k}^2 \cdot \parans{\trace\parans{\mLambda \cdot{\calM^{-1}\parans{  \mA_{k} }}} }^2\right]  \\
& \quad + \E_{\{\mA_k\},\wh \vp}\left[\sum_{k\neq k'} \wh p_{k}\wh p_{k'} \cdot \trace\parans{\mLambda \cdot{\calM^{-1}\parans{  \mA_{k} }}} \cdot \trace\parans{\mLambda \cdot{\calM^{-1}\parans{  \mA_{k'} }}} \right]\\
& = \E_{\{\mA_k\}}\left[\sum_{k=1}^K \frac{p_{k} + (L-1)p_k^2}{L} \cdot \parans{\trace\parans{\mLambda \cdot{\calM^{-1}\parans{  \mA_{k} }}} }^2\right]  \\
& \quad + \E_{\{\mA_k\}}\left[\sum_{k\neq k'} \left(1 -\frac{1}{L}\right) p_k p_{k'} \cdot \trace\parans{\mLambda \cdot{\calM^{-1}\parans{  \mA_{k} }}} \cdot \trace\parans{\mLambda \cdot{\calM^{-1}\parans{  \mA_{k'} }}} \right].
\end{split}\ee

We complete the proof by plugging the above into Eq.~\eqref{eq:MSE-shadow-obsrvable}.
\end{proof}
\end{widetext}

When $L = 1$, Eq.~(\ref{eq:MSE-multi-shots}) reduces to the formulation in  Ref.~\cite{huang2020predicting} (Lemma S1 therein) through the property that $\calM^{-1}$ is self-adjoint and the expression for rank-1 orthonormal POVMs [Eq.~\eqref{eq:shadow-unitary}]. 
During the final preparation and revision of this work, we became aware of two works with derivations similar to ours: one for rank-1 orthonormal POVMs with multishot measurements~\cite{zhou2023performance} and another that focuses on MSE for
a single unitary and then studies how the MSE varies with different choices of unitary groups~\cite{helsen2023thrifty}.
Our formulation in Eq.~(\ref{eq:MSE-multi-shots}) is distinct by holding for general POVMs.  Though Eq.~(\ref{eq:MSE-multi-shots}) may appear complex, if we disregard the last term, we can draw the following two key observations. (i) Since $\trace(\mA_k \vrho)$ is often very small, the dominant term becomes $\frac{1}{ML} \E_{{\mA_k}}\left[\sum_{k=1}^K {\trace(\mA_k \vrho)} \cdot \parans{\trace\parans{\mLambda \cdot{\calM^{-1}\parans{ \mA_{k} }}} }^2\right]$. This term decreases proportionally to $1/ML$, where $ML$ represents the total number of measurements. (ii) Conversely, if we keep $ML$ fixed, using a larger value of $L$ generally results in a larger MSE since the remaining terms increase with $L$. This explains the observed decrease in performance with increasing $L$ as shown in \Cref{fig:one-vs-more}.

\section{Conclusion}
\label{sec:conclusion}
In this paper, we have identified and formalized deep connections between traditional LS-based techniques for quantum state estimation and the disruptive methodologies of CS. Through careful derivation of the LS tomographic problem, we have shown that the LS estimator can be viewed as the average of distinct ``shadows'' $\wh\vrho_m$, each corresponding to a specific measurement, in complete analogy with CS. 
This extension of the shadow picture to LS in turn reveals a novel viewpoint for CS in connection with regularization; just like traditional techniques such as RLS, CS reduces the instabilities of LS in the underdetermined regime through replacement of $(\frac{1}{M}\calA^\dagger\calA)^{+}$ with a well-conditioned channel inverse.

Notwithstanding these intuitive similarities between RLS and CS shadows, our tests above reveal key differences. RLS shadows reduce variance at the cost of bias, are robust to errors in the distribution of random measurements, and are highly sensitive to the tradeoff in the number of POVMs $M$ and number of shots $L$. 
In contrast, CS shadows are unbiased at the expense of variance, produce high estimation errors whenever the actual measurements diverge from the expected distribution, and scale favorably with a variety of $M$ and $L$ combinations. Certainly, although not optimal in all categories of interest, the fact that CS shadows are unbiased for any number of measurements---even in the highly underdetermined regime---is a remarkable feature that distinguishes its version of regularization from alternatives such as RLS.

Irrespective of such observations, none of the various tradeoffs can minimize the exceptional \emph{computational} efficiency possible with CS shadows over both LS and RLS methods. Whenever the  quantum channel $\calM$ can be analytically inverted---as in the example in Eqs.~(\ref{eq:expectation-M},\ref{eq:M-inverse})---CS requires numerical calculation of no matrix inverses, unlike both LS [Eq.~\eqref{eq:shadow-v1}] and RLS [Eq.~\eqref{eq:shadow-v2}]. Indeed,
while examples of CS shadows up to 120 qubits were shown in Ref.~\cite{huang2020predicting}, the record dimensionality for LS tomography (specifically, LS projected onto physical states) is a comparatively meager 14 qubits~\cite{Hou2016}, and it is difficult to imagine significant increases beyond that number with existing computing technology. However, while CS shadows may face minimal competition in ultralarge Hilbert spaces, our findings connecting it to LS methods reveal a fascinating conceptual lineage with traditional methodologies, shedding further light into the secrets of the exciting and transformative tomographic procedure that is CS.

\section*{Code Availability}
The MATLAB code used to produce the results in this study is available at \url{https://github.com/ZhihuiZhu/shadow_ls}.

\acknowledgments
We acknowledge funding support from the National Science Foundation (CCF-2241298, EECS-2409701),
a Partnership Seed Award from the Center for Quantum Information and Engineering (CQISE) at the Ohio State University, and the U.S. Department of Energy, Office of Science, Office of Advanced Scientific Computing Research (ERKJ432, ERKJ353, DE-SC0024257). We thank the Ohio Supercomputer Center for providing the computational resources and the Quantum Collaborative led by Arizona State University for providing valuable expertise and resources. A portion of this work was performed at Oak Ridge National Laboratory, operated by UT-Battelle for the U.S. Department of Energy under Contract No. DE-AC05-00OR22725. We are grateful to Stephen Becker, Zhexuan Gong, Zhen Qin, Michael Wakin, Otfried G\"{u}hne and Nikolai Wyderka for many valuable discussions.

\bibliographystyle{apsrev4-1}
\bibliography{quantum.bib}

\begin{thebibliography}{73}
\bibitem [{lik()}]{likelihoodNote}%
  \BibitemOpen
  \href@noop {} {}\bibinfo {note} {Other approaches to obtain a meaningful
  likelihood from a nonphysical estimate are possible, such as removing terms
  with negative probabilities. However, we have found this approach to lead to
  the same general conclusions and so focus on physical projection only in the
  main text.}\BibitemShut {Stop}%
\end{thebibliography}

\end{document}